%% file: main.tex
\documentclass[11pt]{article}
\usepackage[margin=0.9in]{geometry}

\input{commands}

\begin{document}

\title{Analogue quantum simulation with polylogarithmic interaction strengths by extrapolating within phases of matter}
\author[1]{Dylan Harley}
\author[1]{Matthias Christandl}
\affil[1]{Department of Mathematical Sciences, University of Copenhagen, Denmark}

\date{}
\maketitle

\begin{abstract}
    Simple families of quantum Hamiltonians can simulate general many-body systems at arbitrary precision through the use of perturbative gadgets, however this generally requires interaction strengths spanning many orders of magnitude which scale polynomially in the system size and inverse precision, resulting in physically unrealisable systems. In this work, we show that for non-critical systems these required scalings can be exponentially reduced through classical post-processing, by simulating the model at smaller energy scales and extrapolating observables to the perturbative limit. In particular, we show that both local and extensive properties of thermal states with exponentially decaying correlations and ground states with a sufficiently stable gap can be simulated using gadgets whose interaction strengths scale only polylogarithmically in the inverse precision and the system size.
    As a key tool, we develop a generalised treatment of the local Schrieffer-Wolff transformation for geometrically quasi-local Hamiltonians over many energy scales, facilitating the analysis of perturbative gadget Hamiltonians without extensive global energy penalities, which may be of independent interest.
\end{abstract}

\newpage

\setcounter{tocdepth}{2}
\tableofcontents

\newpage
\section{Introduction}

\subsection{Background and motivation}

The simulation of quantum systems has long been recognised as an important use-case for quantum technologies \cite{feynman1982simulating}, and remains a key focus for applications in both the near and far term \cite{georgescu2014quantum,daley2022practical}. In contrast to the digital case, in which Hamiltonian dynamics are implemented fault-tolerantly by a universal gate-based quantum computer \cite{lloyd1996universal,dalzell2025quantum}, analogue quantum simulation involves encoding the behaviour of a Hamiltonian of interest into a tunable simulator \cite{cirac2012goals}, so that static and dynamic properties of the target system can be inferred from the native static and dynamic properties of the simulator. Compared to digital quantum simulation, this offers a potentially more experimentally tractable approach for near-term probes of many-body physics. Moreover, hybrid algorithms, using analogue simulation as a subroutine, may offer a path to reduce fault-tolerant circuit overhead for simulation tasks \cite{wang2023quantum,wang2025efficient,lloyd2025quantum}; indeed digital-analogue hybrid simulations have already been demonstrated in practice \cite{andersen2025thermalization}. Typically, we consider a situation where an experimenter has access to a tuneable but limited simulator, for instance with the ability to adjust the pairwise interaction strengths between sites whose geometry is fixed. Though these limitations appear to fundamentally restrict the regime of applicability for such a simulator, it turns out that even a very simple such model can be used to simulate arbitrary many-body Hamiltonians via perturbation theory \cite{cubitt2018universal}.

Since their initial application in establishing the $\QMA$-hardness of the $2$-local Hamiltonian problem \cite{kempe2006complexity}, perturbative simulations have become a widespread and powerful tool in the theory of Hamiltonian complexity and analogue simulation \cite{oliveira2005complexity,bravyi2008quantum,cubitt2016complexity,bravyi2017complexity,cubitt2018universal}. The basic idea is to take some target Hamiltonian and replace its interactions term-by-term with so-called perturbative gadgets, which each use an ancillary qubit to mediate interactions between several sites. In this way, the target Hamiltonian can be realised as the effective low-energy theory of the perturbative simulation Hamiltonian, which may belong to a simpler family. By iteratively applying such constructions, it is possible to reduce any many-body Hamiltonian to, for example, a $2$-local Heisenberg model in the plane as in Ref.~\cite{cubitt2018universal}. The drawback of perturbative gadgets is their scaling behaviour: in order to simulate a target Hamiltonian on $n$ sites up to some error $\epsilon > 0$, perturbative gadgets typically contain individual interaction terms which scale as $\poly(n,\epsilon^{-1})$ in order to ensure accurate convergence of the perturbative series. This leads to simulator Hamiltonians whose interactions range over many orders of magnitude, and which are unrealisable in practice even for modest system sizes.

The prior work Ref.~\cite{bravyi2008quantum} has established that $n$-independent interaction strengths (scaling only as $\poly(\epsilon^{-1})$) are sufficient to simulate the ground state energy of a Hamiltonian up to extensive error $\sim n\epsilon$, but this result does not extend to any other properties of the ground state, and the polynomial scaling remains prohibitive if high accuracy is desired. In a previous work \cite{harley2024going} we proved a no-go result showing that, for a broad class of modular simulation techniques, polynomial scaling is generally unavoidable.

Using polynomial extrapolation methods which have previously found application in similar problems for quantum error mitigation \cite{temme2017error,li2017efficient}, tensor network representations \cite{christandl2021optimization}, and recently for digital quantum simulation \cite{low2019wellconditioned,watson2025exponentially}, we show that in many cases these scalings can be exponentially reduced to $\poly\log(n\epsilon^{-1})$ through classical postprocessing. In particular, we show that properties of Gibbs states and ground states of perturbative simulator Hamiltonians can be estimated by measuring the system at several different values of the perturbative parameter, and extrapolating to the infinite interaction strength limit (in which the simulation is perfect); see Figure~\ref{fig:gadget extrapolation} for an illustration. These results require some additional structure: we assume that the system we are attempting to simulate is not close to a phase transition, to avoid attempted extrapolation through critical points. In particular, it is sufficient that the Gibbs state (respectively ground state) of interest has an exponential decay of correlations (respectively constant spectral gap) which is stable to small perturbations (see Conditions~\eqref{it:geometric locality}-\eqref{it:spectral gap} below, and surrounding discussion).

We expect this work may be of interest from several perspectives. Firstly, for practical simulations such as near-term approaches to Gibbs sampling, ground state preparation, and quantum phase estimation including analogue steps (recent proposals include e.g. Refs.~\cite{langbehn2025universal,lloyd2025quantum,hahn2025provably,ding2025endtoend,wang2023quantum,wang2025efficient,tabares2025estimating}), our techniques open the door to using perturbative gadgets without prohibitive interaction scalings to access larger families of Hamiltonians. Secondly, from the perspective of Hamiltonian complexity theory, our work gives a many-to-one reduction between non-critical local Hamiltonian problems, and may be viewed as further evidence of the role of criticality in computational hardness (see Refs.~\cite{gonzalez2018history,deshpande2022importance}). Additionally, for the theory of ground state and Gibbs state learning, the tools we develop may be used to design variations on the algorithms of e.g. Refs.~\cite{rouze2024efficient,lewis2024improved}, which infer the values of observables of parametrised Hamiltonians from random samples. These prior works approximate local observables by linearly interpolating between locally similar sample points in parameter space; meanwhile our results offer accuracy guarantees for when higher-order functions can be fitted between data points.

\begin{figure}
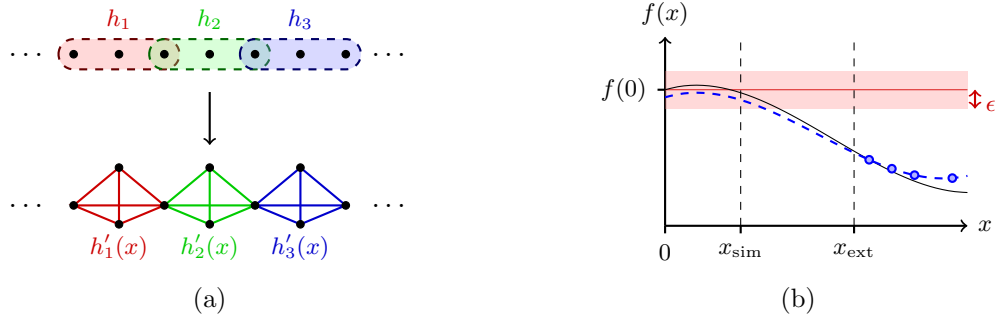

    \centering
    \include{figures/gadgetextrapolation}
    \caption{(a) A 3-local Hamiltonian $H_{\tar} = \sum_i h_i$ can be simulated by a 2-local simulator Hamiltonian $H'(x) = \sum_i h_i'(x)$ by replacing each term with a perturbative gadget containing interactions scaling polynomially with $x^{-1}$ \cite{kempe2006complexity,oliveira2005complexity}. (b) Behaviour of a simulated expectation value $f(x) := \tr[O\rho'(x)]$ (black curve), for $\rho'(x)$ a ground or Gibbs state of the simulator Hamiltonian $H'(x)$. To ensure $|f(x) - f(0)| \leq \epsilon$, one must take $ x \leq x_{\text{sim}} \sim 1/\poly(n,\epsilon^{-1})$. We show that the same accuracy can be achieved by instead noisily sampling $f(x_k)$ at $x_k \sim x_{\text{ext}} =1/\poly\log(n\epsilon^{-1})$ (blue nodes) and extrapolating (blue curve) to $x=0$.}
    \label{fig:gadget extrapolation}
\end{figure}

\subsection{Our contributions}

In this section, we give a brief outline of the tools and contributions of the work; these come in three main parts. First, we prove extrapolation results for Hamiltonians under weak analytic perturbations (not in the simulation regime). Next, we develop tools for analysing (singular) gadget Hamiltonians using local Schrieffer-Wolff perturbation theory. Finally, combining these, we prove that properties of gadget Hamiltonians can themselves be extrapolated to the limit of perfect simulation. The structure of the main results in this work is illustrated in Figure~\ref{fig:result structure}.

\begin{figure}
    \centering
    \begin{tikzpicture}
        \def\boxwidth{4};
        \def\boxheight{1.5};
        \def\boxspacing{1};
        \def\nodexspace{5};
        \def\nodeyspace{2};
        \def\buffer{0.2};

        \draw[fill=yellow,fill opacity=0.3,draw=yellow!80!black,thick,dashed,rounded corners = 0.3 cm] ($(\nodexspace,-\nodeyspace) + (\boxwidth/2 + \buffer,\boxheight / 2 + \buffer)$) rectangle ($(\nodexspace,-2*\nodeyspace) - (\boxwidth/2 + \buffer,\boxheight / 2 + \buffer)$);
        
        \foreach\x in {0,1} {
        \foreach\y in {0,1,2} {
        \draw[thick,rounded corners = 0.3 cm,fill=white] ($(\x*\nodexspace,-\y*\nodeyspace) - (\boxwidth/2,\boxheight/2)$) rectangle ($(\x*\nodexspace,-\y*\nodeyspace) + (\boxwidth/2,\boxheight/2)$);
        }
        }
    
        \node[font=\footnotesize,align=center] at (0,0) {Localising perturbations\\of Gibbs and ground states \\(Lemmas~\ref{lem:GALI gibbs} and \ref{lem:GALI ground})}; 
        
        \node[font=\footnotesize,align=center] at (\nodexspace,0) {Richardson extrapolation \\ via approximate analyticity \\(Corollary~\ref{cor:richardson extrapolation approx})};

        \node[font=\footnotesize,align=center] at (0,-\nodeyspace) {Analyticity under small \\ perturbations \\ (Theorems~\ref{thm:gibbs small pert} and \ref{thm:ground state small pert})};

        \node[font=\footnotesize,align=center] at (\nodexspace,-\nodeyspace) {Extrapolating locally \\ perturbed Hamiltonians \\(Theorem~\ref{thm:extrapolation without gadgets})};

        \node[font=\footnotesize,align=center] at (0,-2*\nodeyspace) {Local Schrieffer-Wolff \\theory and gadgets \\(Theorem~\ref{thm:simulation properties})};

        \node[font=\footnotesize,align=center] at (\nodexspace,-2*\nodeyspace) {Extrapolating simulator \\ Hamiltonians \\(Theorem~\ref{thm:extrapolation with gadgets})};

        \foreach\y in {0,1} {
        \draw[thick,->] ($(\nodexspace,-\y*\nodeyspace) - (0,\boxheight/2)$) -- ($(\nodexspace,-\y*\nodeyspace-\nodeyspace) + (0,\boxheight/2)$);
        \draw[thick,->] ($(0,-\y*\nodeyspace-\nodeyspace) + (\boxwidth/2,0)$) -- ($(\nodexspace,-\y*\nodeyspace-\nodeyspace) - (\boxwidth/2,0)$);
        }

        \draw[thick,->] ($(0,0) + (\boxwidth/2,0)$) to[out=0,in=180] ($(\nodexspace,-\nodeyspace) - (\boxwidth/2, -\boxheight/3)$);
    \end{tikzpicture}
    \caption{Structure of the key results of this work, with the main extrapolation results highlighted.}
    \label{fig:result structure}
\end{figure}
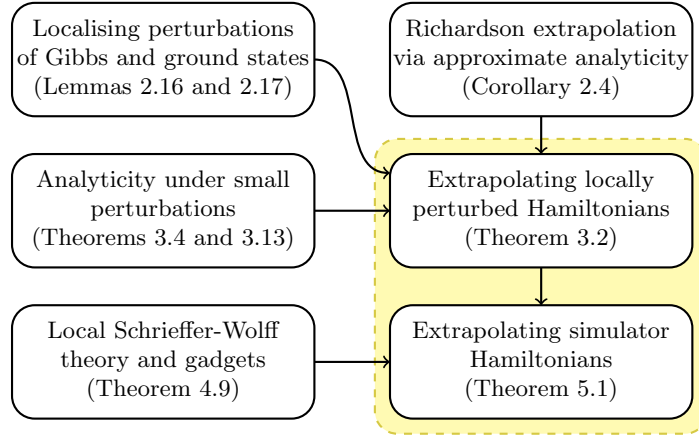

\subsubsection*{Extrapolating within phases of matter}

Initially, we consider a family of Hamiltonians $H(x)$ depending analytically on some parameter $x$ (later, we will consider simulator Hamiltonians depending polynomially on $x^{-1}$). Letting $\rho(x)$ be the corresponding Gibbs or ground state of $H(x)$, and choosing an observable $O$, we aim to estimate the value of a function $f(x) = \tr[O \rho(x)]$ at $x=0$, from a set of $m$ sample points $\{f(x_k)\}_{k=1}^m$ where $x_k > 0$ for all $k$. To this end, we use Richardson extrapolation \cite{richardson1911approximate,sidi2003practical}, which estimates the value of $f(0)$ by fitting a degree-$(m-1)$ polynomial to the samples $\{f(x_k)\}_{k=1}^m$, and simply evaluating this at $x=0$.

There are two sources of error in Richardson extrapolation. Firstly, the choice of sample points $\{x_k\}_{k=1}^m$ may lead to an ill-conditioned problem to fit the polynomial, making the procedure very sensitive to errors. To overcome this issue, we follow Ref.~\cite{low2019wellconditioned} and choose so-called Chebyshev nodes, leading to an optimally well-conditioned problem. Secondly, we need to ensure that $f(x)$ has a good polynomial approximation; to this end, it is sufficient to show that $f(x)$ is well-approximated on the real line by a complex function $\tilde{f}(z)$, which is analytic and bounded in some region of interest. 

Our problem thus reduces to showing that the functions of interest $f(x) = \tr[O\rho(x)]$ have sufficiently good analytic approximations. This is not immediately obvious: for Gibbs states $\rho_\beta(x) := e^{-\beta H(x)} / \tr[e^{-\beta H(x)}]$, extending the domain of $x$ to the complex plane may lead to zeroes in the partition function $Z(x) := \tr[e^{-\beta H(x)}]$ in turn causing non-analyticities in $\rho_\beta(x)$. Though such behaviour cannot occur on the real axis, in the thermodynamic limit $n\rightarrow\infty$ the zeroes can approach the real axis as $\sim 1/n$, limiting the domain in which extrapolation is possible, see Figure~\ref{fig:partition zeroes} for an illustration. Meanwhile for ground states, $\rho(x)$ may not be continuous or even well-defined.

\begin{figure}
    \centering
    \input{figures/partitionzeroes}
    \caption{Zeroes (blue) of the complex partition function $Z(z) := \tr[e^{-\beta H(z)}]$, for an analytic family of Hamiltonians $z \mapsto H(z)$. Zeroes of $Z(z)$ correspond to non-analytic points of the Gibbs state $\rho_\beta(z) := e^{-\beta H(z)} / \tr[e^{-\beta H(z)}]$, so that the radius of convergence for a Taylor expansion for $\rho_\beta(z)$ is given by the distance to the nearest zero. In a region where $\rho_\beta(z)$ is analytic (shaded orange), Richardson extrapolation can be performed along the real axis. At phase transitions $z_{\mathsf{crit}} \in \RR$, the complex zeroes of $Z(z)$ ``pinch'' the real axis in the thermodynamic limit $n\rightarrow \infty$, preventing extrapolation past these points. A similar problem, where $Z$ is instead viewed as a function of $\beta$, is studied in Ref~\cite{harrow2020classical}.}
    \label{fig:partition zeroes}
\end{figure}
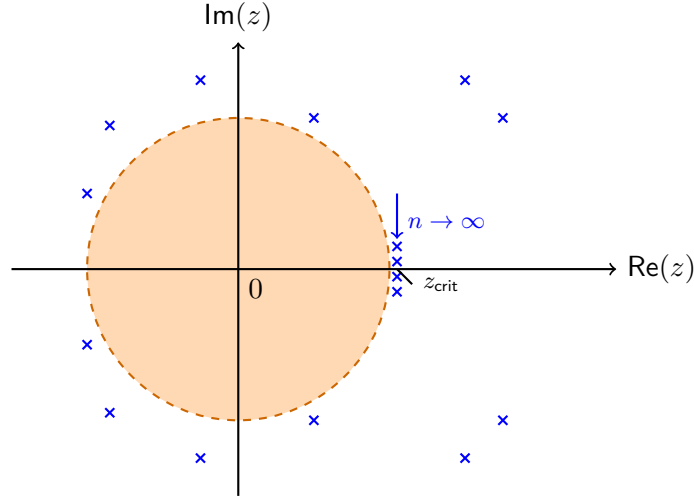

As a result, additional assumptions are necessary to guarantee that $f(x)$ can be approximated --- physically speaking, this corresponds to assuming that we do not attempt to extrapolate through a phase transition. As discussed below, these assumptions are typically hard to prove rigorously for explicit examples of Hamiltonians (partially on fundamental complexity-theoretic grounds \cite{cubitt2015undecidability}) and are slightly weaker versions of those used in prior works on learning theory \cite{huang2022provably,lewis2024improved,rouze2024efficient}. We informally describe these assumptions below, and first introduce some notation. We assume the Hamiltonian $H(x)$ acts on a set of sites $\Gamma$, where $|\Gamma| = n$, with some underlying metric $\dist : \Gamma \times \Gamma \rightarrow \RR_{\geq 0}$. We assume that this describes a $D$-dimensional space (that is, balls of radius $r\geq 0$ with respect to $\dist$ have volume scaling as $\bigO(r^D)$). We decompose $H(x) = \sum_{A\subseteq \Gamma} H_A(x)$, where each $H_A(x)$ is supported only on the sites $A\subseteq \Gamma$. Our assumptions are then as follows (see Assumption~\ref{assumption:non criticality}):
\begin{enumerate}[(I)]
    \item\label{it:geometric locality} \emph{Geometric quasi-locality:} We assume that $H(x)$ has exponentially decaying interactions, in the sense that for any sites $x,y\in \Gamma$, the total strength of interactions connecting them decays exponentially with distance:
    \begin{align}
        \sum_{\substack{A\subseteq \Gamma \\ x,y\in A}} \|H_A(s)\| \leq \exp\left( -\Omega(\dist(x,y)\right)\ .
    \end{align}
    \item\label{it:correlation decay} \emph{Uniform decay of correlations (for Gibbs states):} We assume that, for some region $x \in [0,x_{\ast}]$, $x_{\ast}> 0$, the Gibbs state $\rho_\beta(x) \propto e^{-\beta H(x)}$ has exponentially decaying correlations. For $\Gamma' \subseteq \Gamma$ and $x,y\in [0,x_{\ast}]$, we assume the same holds for the Gibbs states $\rho_\beta(x,y;\Gamma') \propto e^{-\beta H(x,y;\Gamma')}$, where $H(x,y;\Gamma')$ is the Hamiltonian where interactions acting on $\Gamma'$ are given parameter $x$, whilst all others are given parameter $y$:
    \begin{align}\label{eq:partitioned H}
        H(x,y;\Gamma') := \sum_{A\subseteq \Gamma'} H_A(x) + \sum_{\substack{A\subseteq \Gamma \\ A\cap (\Gamma\setminus \Gamma') \neq \emptyset}} H_A(y)\ .
    \end{align}

    \item\label{it:spectral gap} \emph{Uniform spectral gap (for ground states):} We assume that, for some $x \in [0,x_{\ast}]$, the Hamiltonian $H(x)$ has a unique ground state with a constant spectral gap $\gamma > 0$. We assume the same holds for the Hamiltonians $H(x,y;\Gamma')$ as defined in Eq.~\eqref{eq:partitioned H}.
\end{enumerate}

Condition~\eqref{it:geometric locality} ensures that the Hamiltonians $H(x)$ are sufficiently localised so that Lieb-Robinson bounds \cite{lieb1972finite,nachtergaele2019quasi} can be applied. Condition~\eqref{it:correlation decay} (respectively \eqref{it:spectral gap}) ensures that the Gibbs (respectively ground) state is far from phase transitions. That this behaviour holds over a range of parameters of $H(x,y;\Gamma')$ is essential to argue that the $x$-dependence can be effectively ignored far away from a local observable of interest, as we describe below. Conditions~\eqref{it:correlation decay}-\eqref{it:spectral gap} are similar to those used in Refs.~\cite{huang2022provably,lewis2024improved,rouze2024efficient} to obtain provable learning results for properties of quantum states within phases of matter. In particular, Condition~\eqref{it:correlation decay} is satisfied automatically at sufficiently high temperature \cite{alhambra2023quantum,frohlich2015some}, for translationally-invariant systems in one dimension \cite{araki1969gibbs,perezgarcia2023locality}, or for commuting or one-dimensional systems above a thermal phase transition \cite{harrow2020classical}. Condition~\eqref{it:spectral gap} is implied by the stronger condition of local topological quantum order \cite{bravyi2010topological}, which holds for toy models such as the toric code --- however in general proving whether a spectral gap persists for a family of Hamiltonians in the $n\rightarrow \infty$ limit is computationally intractable \cite{cubitt2015undecidability}.

Below we summarise the main result of this section, establishing that local properties of such systems can be extrapolated, and which is stated fully and proved in Section~\ref{sec:basic extrapolation}.

\begin{result}[Extrapolation within phases of matter --- see Theorem~\ref{thm:extrapolation without gadgets} and Corollary~\ref{cor:extrapolating without gadgets explicit}]\label{result:extrapolation without gadgets}
    Suppose $H(x)$ is a family of Hamiltonians depending analytically on $x$ and satisfying Conditions~\eqref{it:geometric locality} and \eqref{it:correlation decay} (respectively \eqref{it:geometric locality} and \eqref{it:spectral gap}), and let $O_A$ be an observable supported on sites $A\subseteq \Gamma$, $|A| = \bigO(1)$. Let $f(x) = \tr[O_A\rho(x)]$, for $\rho(x)$ a constant-temperature Gibbs state (respectively the ground state) of $H(x)$. Then the value of $f(0)$ can be calculated up to any desired accuracy $\epsilon > 0$ via Richardson extrapolation from the values of $f(x_k)$ for $x_k \geq 1/\poly\log(\epsilon^{-1})$.
\end{result}

The case of extensive observables $O = \sum_{A\subseteq \Gamma} O_A$ can also be dealt with by writing $O$ as a sum of local observables and arguing that they can be individually extrapolated (see Appendix~\ref{app:extensive properties}). 

The proof of Result~\ref{result:extrapolation without gadgets} requires two steps. Firstly, we argue that (up to a small error) the value of $f(x)$ only depends on the variation of $H(x)$ within a radius $r\sim \log\epsilon^{-1}$ of $A$ --- in other words, showing that we can assume that $\rho(x)$ is in fact the Gibbs (respectively ground) state of $H(x,0;B_r(A))$, where $B_r(A)$ contains all sites within a radius $r$ of $A$. See Figure~\ref{fig:localisation} for an illustration. This step requires our assumptions \eqref{it:geometric locality}-\eqref{it:spectral gap}, and uses the tools of quantum belief propagation \cite{hastings2007quantum} (respectively the spectral flow \cite{hastings2005quasiadiabatic,bachmann2012automorphic}) to show that the variation of $\rho(x)$ with respect to $x$ can be expressed as a local evolution, and the influence of distant interactions on $O_A$ is controlled with Lieb-Robinson bounds \cite{lieb1972finite,nachtergaele2019quasi}.

Having reduced to local perturbations supported on a small region around $O_A$, the second step of the proof involves establishing that $\rho(x)$ depends (approximately) analytically on $x$ for sufficiently small perturbations of this form. For Gibbs states, this requires establishing lower bounds on the magnitude of the partition function $\tr[e^{-\beta(H(0) + V(x))}]$ for $\|V(x)\| = \bigO(1)$, to ensure that the Gibbs state has no non-analyticities. This follows from a series of arguments involving elementary linear algebra. For ground states, we write the slightly perturbed ground state in terms of a time-ordered exponential of the spectral flow operator, and establish that this is approximately analytic and bounded for small perturbations.

\begin{figure}
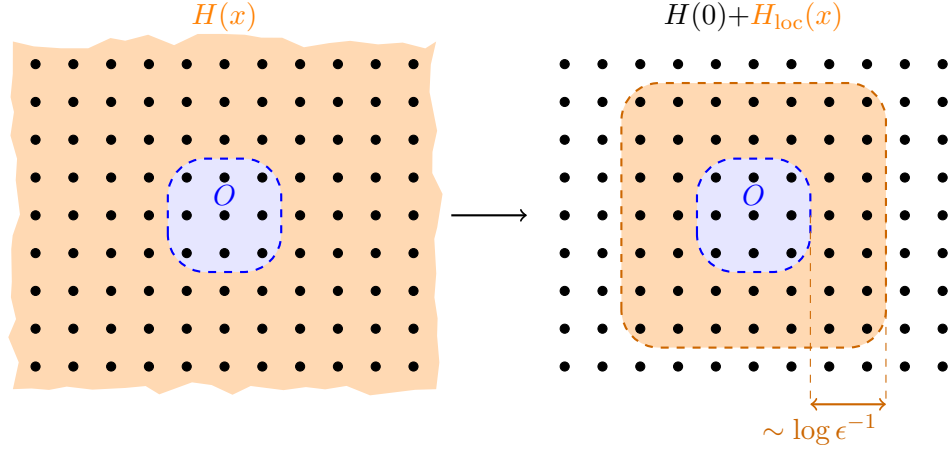

    \centering
    \include{figures/localisation}
    \caption{Away from criticality, the local expectation value $\tr[O \rho(x)]$ (for $\rho(x)$ a ground or Gibbs state) of a geometrically local Hamiltonian $H(x)$ can be approximated by ignoring the variation of the Hamiltonian outside of a region of $\log \epsilon^{-1}$ from the support of $O$ (shaded in blue). Thus, for extrapolating properties of the perturbed Hamiltonian, it is sufficient to establish that $\tr[ O \rho(x)]$ is well-approximated by an analytic function for Hamiltonian perturbations of magnitude $\|H_{\loc}(x)\| \sim \poly\log\epsilon^{-1}$.}
    \label{fig:localisation}
\end{figure}

\subsubsection*{Local Schrieffer-Wolff perturbation theory}

Note that Result~\ref{result:extrapolation without gadgets} only applies to Hamiltonians depending analytically on the perturbation parameter $x$ which we aim to send to $0$. A priori, this appears to be a different regime to the case of perturbative simulator Hamiltonians, which typically depend polynomially on $x^{-1}$ and are thus singular at $x=0$. Our next results involve a systematic analysis of the properties of such Hamiltonians, ultimately aiming to circumvent the singularity by restricting to an effective Hamiltonian which depends only analytically on $x$.

The general construction for a simulator Hamiltonian $H'(x)$ is as follows. The set of sites $\Gamma$ is partitioned into $\Gamma = \Gamma_{\eff} \cup \Gamma_{\anc}$, where the sites $i\in \Gamma_{\anc}$ are ancillary sites used to mediate interactions and induce an effective Hamiltonian on $\Gamma_{\eff}$ (for simplicity we assume these are qubits, with local basis states $\ket{0_i}$ and $\ket{1_i}$). In particular, $H'(x)$ then takes the form
\begin{align}\label{eq:hprime definition intro}
    H'(x) = \Delta x^{-d} \sum_{i\in \Gamma_{\anc}} \proj{1_i} +  \sum_{\alpha \geq 1} x^{\alpha-d} H^{(\alpha)}\ ,
\end{align}
where $\Delta > 0$ and $\proj{1_i}$ is a 1-local projector onto acting on site $i$, and for some $d\geq 0$ referred to as the degree of the simulation. The $x^{-d}$ term then incurs an energy penalty on all configurations of the ancillary sites except the all-zeroes state $\ket{\mathbf{0}_{\anc}}$. Meanwhile, the lower-order terms $H^{(\alpha)}$ are geometrically local Hamiltonians with interactions between $\Gamma_{\anc}$ and $\Gamma_{\eff}$. For sufficiently small $x > 0$ (depending on $n$), the energy penalty from the $x^{-d}$ term is large enough to induce a global energy gap, leading to the projector onto a low-energy space of $H'(x)$ given (approximately) by $\id_{\eff} \otimes \proj{\mathbf{0}}_{\anc}$. Restricting to this low-energy space, the $\Gamma_{\eff}$ system evolves under some different, effective Hamiltonian.

More concretely, $H'(x)$ is block-diagonal with respect to the projector $e^{-T(x)} (\id_{\eff} \otimes \proj{\mathbf{0}_{\anc}} ) e^{T(x)}$, for $T(x)$ an anti-hermitian operator generating the small rotation $e^{T(x)}$. This allows us to define the effective low-energy Hamiltonian $H_{\eff}(x)$ as
\begin{align}
    H_{\eff}(x) := (\id_{\eff} \otimes \bra{\mathbf{0}_{\anc}}) e^{T(x)} H'(x) e^{-T(x)} (\id_{\eff} \otimes \ket{\mathbf{0}_{\anc}}) \ .
\end{align}
Surprisingly, it is possible to construct such $T(x)$ and $H_{\eff}(x)$ even in the case when $x$ is not small enough to induce a global energy gap: one can construct a power series $T(x) = \sum_{q\geq 1} x^q T^{(q)}$ (known as the local Schrieffer-Wolff transformation \cite{datta1996low,bravyi2011schrieffer}), such that $iT(x)$ is a bounded geometrically quasi-local Hamiltonian, and $e^{T(x)} H'(x)e^{-T(x)}$ is block-diagonal with respect to $\id_{\eff} \otimes \proj{\mathbf{0}_{\anc}}$. Moreover, the power series for $T(x)$ (and hence $H_{\eff}(x)$) is well-defined and convergent for $x$ smaller than some constant $x_\ast > 0$, which is independent of the system size. This conclusion is in contrast to the usual (global) Schrieffer-Wolff transformation \cite{bravyi2011schrieffer}, for which $x$ must be taken sufficiently small that $H'(x)$ has a global energy gap above the effective space $\id_{\eff}\otimes \proj{\mathbf{0}_{\anc}}$.

Our treatment of local Schrieffer-Wolff perturbation theory is similar to Ref.~\cite{bravyi2011schrieffer}, with two major differences. Firstly, we explicitly consider Hamiltonians of the form Eq.~\eqref{eq:hprime definition intro} containing higher-order powers of $x$, allowing us to explicitly obtain $H_{\eff}(x)$ as a power series in $x$, in contrast to Ref.~\cite{bravyi2011schrieffer} in which the Hamiltonians involved have the form $H_0 + x V$. Secondly, our analysis makes use of local Hamiltonian norms introduced in Ref.~\cite{nachtergaele2019quasi}, which naturally allow the application of Lieb-Robinson bounds. In fact, using these norms turns out to make the analysis slightly simpler, leading to provably convergent power series for $T(x)$ and $H_{\eff}(x)$ which are instead truncated at finite order in Ref.~\cite{bravyi2011schrieffer} --- however we in turn require the assumption of geometric locality, which is not present in Ref.~\cite{bravyi2011schrieffer}.

If $H_{\eff}(x)$ is analytic in $x$ (i.e., the negative powers in its series expansion all cancel), we say that $H'(x)$ is a simulator Hamiltonian, and we thus recover a well-defined effective Hamiltonian in the limit $x\rightarrow 0$. We formalise this notion in Section~\ref{sec:gadgets}, and show that it recovers familiar notions of simulation formalised in Refs.~\cite{bravyi2017complexity,cubitt2018universal,harley2024going}. We also prove a general result (see Theorem~\ref{thm:gadget comb}) establishing that such Hamiltonians can be built from ``gadgets'' used in parallel: roughly speaking, given several $\{h_j'(x)\}_j$ such that each $h_j'(x)$ is a degree $d\leq 3$ simulation with effective Hamiltonian $h_{\eff,j}(x)$, the combined Hamiltonian $H'(x) = \sum_j h_j'(x)$ yields an effective Hamiltonian $H_{\eff}(x) = \sum_j h_{\eff,j}(x)$. 

Our treatment of local Schrieffer-Wolff perturbation theory allows us to prove several useful properties of simulator Hamiltonians, which we summarise below as Result~\ref{result:simulation properties}. To our knowledge, this is the first formal treatment of Hamiltonian gadgets for quantum simulation with low energies via the local Schrieffer-Wolff transformation, and we expect these tools may be independently useful.
\begin{result}[Properties of simulator Hamiltonians --- see Theorem~\ref{thm:simulation properties}]\label{result:simulation properties}
    Let $H'(x)$ be a degree-$d$ simulator Hamiltonian constructed as in Eq.~\eqref{eq:hprime definition intro}, with analytic effective Hamiltonian $H_{\eff}(x)$. There is a constant $x_{\ast} > 0$ such that, for $x \leq x_{\ast}$, the following holds:
    \begin{enumerate}[(I)]
        \item \emph{Effective Hamiltonian:} The power series for $H_{\eff}(x)$ converges, and $H_{\eff}(x)$ inherits the same quasi-locality properties as $H'(x)$ (in particular, if each $H^{(\alpha)}$ has exponentially decaying interactions, then so does $H_{\eff}(x)$).

        \item \emph{Gibbs states:} Let $\rho'_\beta(x)$ and $\rho_{\beta,\eff}(x)$ be the Gibbs states of $H'(x)$ and $H_{\eff}(x)$ respectively at inverse temperature $\beta > 0$. Then
        \begin{align}
            \rho_{\beta}'(x) \approx e^{-T(x)} (\rho_{\beta,\eff}(x) \otimes \proj{\mathbf{0}_{\anc}} ) e^{T(x)} \ ,
        \end{align}
        where ``$\approx$'' denotes approximation in the trace norm, up to an error which decays exponentially with $x^{-d}$. 

        \item \emph{Ground states:} Let $P'(x)$ and $P_{\eff}(x)$ be the ground state projectors of $H'(x)$ and $H_{\eff}(x)$ respectively. Then
        \begin{align}
            P'(x) = e^{-T(x)} (P_{\eff}(x) \otimes \proj{\mathbf{0}_{\anc}}) e^{T(x)}\ .
        \end{align}
    \end{enumerate}
\end{result}

\subsubsection*{Extrapolation of simulator Hamiltonians}

The third main result of this work can be viewed as a combination of Results~\ref{result:extrapolation without gadgets} and \ref{result:simulation properties}; it establishes that properties of simulator Hamiltonians (of the form given in Eq.~\eqref{eq:hprime definition intro}) can be extrapolated to the limit $x\rightarrow 0$ corresponding to perfect simulation. We state this result below; see Figure~\ref{fig:scalings table} for a summary of the interaction strengths necessary to compute different quantities.

\begin{figure}
    \centering
    \begin{tabular}{l|l|c}
        \textbf{Simulation task} & \textbf{State} & \textbf{Interaction strengths} \\
        \hline\hline
        Full-spectrum simulation & Gibbs & $\poly(n,\epsilon^{-1})$\\
        (See Refs.~\cite{bravyi2017complexity,cubitt2018universal}) & Ground & $\poly(n,\epsilon^{-1})$ \\
        \hline
        Ground state energy simulation \cite{bravyi2008quantum} & Ground & $\poly(n\epsilon^{-1})$ \\
        \hline
        Local observable extrapolation & Gibbs & $\poly \log( n\epsilon^{-1})$ \\
        (See Theorem~\ref{thm:extrapolation with gadgets}, Corollary~\ref{cor:extrapolation with gadgets}) & Ground & $\poly \log (\epsilon^{-1})$ \\
        \hline
        Extensive observable extrapolation & Gibbs & $\poly \log (n\epsilon^{-1})$ \\
        (See Theorem~\ref{thm:extensive extrapolation simulator}, Corollary~\ref{cor:extensive extrapolation simulator}) & Ground & $\poly \log(n\epsilon^{-1})$
    \end{tabular}
    \caption{Comparison of the interaction strengths required in a simulator Hamiltonian $H'(x)$ to estimate properties of a geometrically local target Hamiltonian $H_{\tar}$ up to (additive) error $\epsilon > 0$. In particular, we aim to estimate quantities of the form $\tr[O\rho]$, where $O$ is an observable and $\rho$ is a ground state or Gibbs state of $H_{\tar}$, by measuring corresponding observables on $H'(s)$. We separately consider local observables, for which $O$ is supported on a constant number of sites, and extensive observables where $O$ is allowed to contain quasi-local interactions throughout the system (for example, $O=H_{\tar}$). Full spectrum simulation as in Refs.~\cite{bravyi2017complexity,cubitt2018universal} captures all properties of the target Hamiltonian to any desired accuracy without any need for extrapolation, and works in full generality. Via extrapolation we establish that the same techniques can be used with exponentially weaker interaction strengths, for systems away from phase transitions.}
    \label{fig:scalings table}
\end{figure}

\begin{result}[Extrapolation of simulator Hamiltonians --- see Theorem~\ref{thm:extrapolation with gadgets} and Corollary~\ref{cor:extrapolation with gadgets}]\label{result:extrapolation with gadgets}
    Suppose $H'(x)$ is a family of simulator Hamiltonians as in Eq.~\eqref{eq:hprime definition intro}, such that the effective Hamiltonian $H_{\eff}(x)$ satisfies the conditions of Result~\ref{result:extrapolation without gadgets}. Let $O_A$ be an observable supported on sites $A \subseteq \Gamma_{\eff}$, $|A| = \bigO(1)$. Let $\rho'(x)$ and $\rho_{\eff}(x)$ be the constant-temperature Gibbs states (respectively the ground states) of $H'(x)$ and $H_{\eff}(x)$, and define $f'(x) := \tr[(O_A \otimes \id_{\anc})\rho'(x)]$. Then the value of $\tr[O_A \rho_{\eff}(0)]$ can be calculated up to any desired accuracy $\epsilon > 0$ via Richardson extrapolation, from values of $f'(x_k)$ with $x_k \geq 1/\poly\log(n\epsilon^{-1})$ (respectively $x_k \geq 1/\poly\log(\epsilon^{-1})$). In particular, the Hamiltonians $H'(x_k)$ contain interaction strengths of order $\poly\log(n\epsilon^{-1})$ (respectively $\poly\log(\epsilon^{-1})$).
\end{result}

The corresponding result for extensive quantities $O = \sum_{A\subseteq \Gamma} O_A$ (see Appendix~\ref{app:extensive simulator}) shows that interaction strengths of order $\poly\log(n\epsilon^{-1})$ are sufficient for both Gibbs states and ground states. In particular, choosing $O = H_{\eff}(0)$, the ground state energy for a non-critical system can be estimated to extensive precision $\sim \epsilon n$ using gadget interaction strengths of order $\sim \poly \log (\epsilon^{-1})$, as opposed to $\sim \poly (\epsilon^{-1})$ as in Ref.~\cite{bravyi2008quantum}. Our gadget formalism is sufficiently general to capture all mediator gadgets that we are aware of (such as Refs.~\cite{oliveira2005complexity,bravyi2017complexity,cubitt2018universal}). As a concrete example, this includes reductions to simulate $3$-local Hamiltonians with only $2$-local interactions, as we describe in Section~\ref{sec:example} --- however our results are much more general and apply to all mediator gadgets in the literature which we are aware of. We note however that geometric locality is essential, so our results do not apply to simulations which radically alter the global geometry.

\subsection{Discussion and future directions}

In this work, we have proved how the overhead for analogue quantum simulation (as measured by local interaction strengths) can be exponentially reduced from $\poly(n,\epsilon^{-1})$ to $\poly\log(n\epsilon^{-1})$, for Gibbs state and ground state properties in non-critical systems on $n$ sites up to precision $\epsilon$. To this end, we used tools from many-body theory to extend previous works on quantum learning \cite{huang2022provably,lewis2024improved,rouze2024efficient}, establishing that such properties vary approximately analytically for analytically perturbed Hamiltonians. This allows the application of Richardson extrapolation \cite{richardson1911approximate} for polynomial interpolation. Moreover, we have analysed the behaviour of perturbative gadget Hamiltonians in the low-energy regime (that is, without a global energy gap) using the local Schrieffer-Wolff transformation \cite{bravyi2011schrieffer}, in particular proving that Gibbs and ground state properties are faithfully simulated up to high precision. This gives a clean reduction from singular simulator Hamiltonians to analytic effective Hamiltonians for quantities of interest. As an example, our results can be applied for the simulation of $3$-local Hamiltonians by $2$-local simulators, though our treatment of perturbative simulation is general enough to encompass all other mediator gadget constructions we are aware of.

\paragraph{Non-criticality assumptions.} One drawback of our results is that the non-criticality conditions we assume are difficult to prove (and even potentially undecidable \cite{cubitt2015undecidability}) for families of Hamiltonians other than simple toy models. This problem is also present in the related works \cite{huang2022provably,rouze2024efficient,lewis2024improved}, which use similar techniques to ours to localise Hamiltonian perturbations and infer local properties. Whilst it seems unlikely that this problem can be solved entirely, as phase transitions pose a fundamental barrier for extrapolation, it may be possible to establish weaker assumptions which are still sufficient. A potential extension in this direction comes from Ref.~\cite{onorati2023provably}, which establishes learning results under the phase definition introduced by Ref.~\cite{coser2019classification}, wherein two states are said to lie in the same phase if they can be related by a short-time dissipative evolution. Furthermore, the numerics of Ref.~\cite{huang2022provably} suggest that predicting properties within phases may work in practice in many situations where the required assumptions cannot be rigorously established in theory; it would be interesting to see whether similar empirical conclusions can be obtained in our case.

\paragraph{Applications for simulating dynamics.} Our results apply to static properties of simulator Hamiltonians, that is, properties of their thermal and ground states. It is natural to ask whether the local Schrieffer-Wolff and classical post-processing techniques might also be useful for analysing the dynamics of such systems. Further work in this direction could build on recent works (see for example Refs.~\cite{cai2023stochastic,rao2025stability}) analysing and bounding the propagation of errors in noisy analogue simulators, or using perturbative methods to encode the target dynamics into a space protected by a quantum error-detecting code \cite{cao2024robust}.

\paragraph{Analyticity and phase transitions.} One more fundamental mathematical question raised by this work involves the relationship between phase transitions and zeroes of the partition function. In particular, given a family of Hamiltonians $H(x)$ depending analytically on $x$, the Gibbs states $e^{-\beta H(x)} / \tr[e^{-\beta H(x)}]$ are analytic when extended to complex $x$ except at zeroes of the partition function $Z(x) := \tr[e^{-\beta x}]$. Distance from the nearest zeroes of $Z(x)$ controls the radius of convergence of the Taylor series, and thus whether Richardson extrapolation can be directly applied (see Figure~\ref{fig:partition zeroes}). In our proof of Result~\ref{result:extrapolation without gadgets}, we circumvent this issue using the assumption of exponential correlation decay, which allows us to approximately restrict the variation of $H(x)$ and bound the zeroes of $Z(x)$ away from $x=0$. Nonetheless, it is an interesting theoretical question whether our assumption~\eqref{it:correlation decay} is sufficient to directly prove exact analyticity of the Gibbs state, for local variations independent of the system size. Previous work \cite{harrow2020classical} has examined a closely related problem, in which the partition function is viewed as a complex function of the inverse temperature $\beta$ with the Hamiltonian fixed. In particular, the authors show that under some additional assumptions, an absence of zeroes near the real line of the partition function implies that distant correlations are exponentially small. In the classical case these conditions are known to be equivalent \cite{dobrushin1987completely,harrow2020classical}, and it remains an open question to establish an analogous quantum result.

\section{Preliminaries}

\subsection{Richardson extrapolation}\label{sec:richardsonextrapolation}

On the classical postprocessing side, we will use Richardson extrapolation \cite{richardson1911approximate,sidi2003practical}, a method which has previously been applied in the different but related context of digital quantum simulation via Trotterisation \cite{low2019wellconditioned,watson2025exponentially}. In this section we sketch the main ideas of this approach, and state the results necessary for our purposes.

We aim to compute the value of $f(0)$, for some unknown function $f : \CC \rightarrow \CC$ which we can only access at limited precision for inputs $x \in (0,x_{\max}]$. That is, we are given some set of estimates $\{\hat{f}(x_1),\dots,\hat{f}(x_m)\}$ such that $|\hat{f}(x_k) - f(x_k)| \leq \delta$ for all $k$, for some error rate $\delta > 0$. Sampling from values of $x$ closer to zero will be more expensive (in our case, this will ultimately correspond to implementing a simulator Hamiltonian with interaction strengths scaling as $\sim 1/\poly (x)$), so ideally we would like an extrapolation scheme such that $\min_k |x_k|$ need not be too small.

The strategy is as follows: we approximate $f(x)$ with a degree-$(m-1)$ polynomial $f(x) = c_0 + c_1 x + \dots + c_{m-1} x^{m-1} + E_m(x)$ (up to some higher-order errors $E_m(x) = \bigO(x^m)$), and estimate the coefficients $c_k$ by some $\hat{c}_k$, calculated by inverting the resulting set of linear equations
\begin{align}\label{eq:linear system for richardson extrap}
    \hat{f}(x_k) = \hat{c}_0 + \hat{c}_1 x_k + \dots + \hat{c}_{m-1} x_k^{m-1}\ ,\quad \text{for $k=1,\dots,m$}\ .
\end{align}
The estimated coefficient $\hat{c}_0$ is then our approximation for $f(0)$. Notice that there are two sources of error in this procedure: the error $E_m(x)$ incurred by our polynomial approximation, and the sampling errors $\hat{f}(x_k) - f(x_k)$.

The former error can be controlled by bounding the Taylor series truncation error of the function $f(x)$, expanded around $x=0$. To this end, we can use the following corollary of Cauchy's integral theorem: 
\begin{lemma}[Taylor series truncation error --- see \cite{harrow2020classical}, Proposition 18]\label{lem:taylor series truncation}\noproofref
    Suppose $f : \CC \rightarrow \CC$ is an analytic function bounded as $|f(x)|\leq M$ for $|x| \leq b$, where $b>1$. Then the error of truncating $f(x)$ by a Taylor series of degree $m$ in $|x| \leq 1$, $f(x) = c_0 + c_1 x + \dots + c_m x^m + E_{m+1}(x)$, is bounded by
    \begin{align}
        \big|E_{m+1}(x) \big| \leq \frac{M}{b^m(b-1)}\ ,\quad \text{for $|x|\leq 1$}\ .
    \end{align}
\end{lemma}

The latter error is slightly more subtle; it depends on the conditioning of the linear system defined by \cref{eq:linear system for richardson extrap}, which is in turn very sensitive to the choice of sampling points $x_j$. The linear system corresponds to a Vandermonde matrix, which we can solve explicitly for $\hat{c}_0$ as
\begin{align}
    \hat{c}_0 &= \begin{pmatrix}
        1 & 0 \dots & 0
    \end{pmatrix} \begin{pmatrix}
        1 & x_1 & \dots & x_1^{m-1} \\
        1 & x_2 & \dots & x_2^{m-1} \\
        \vdots & & \ddots & \vdots \\
        1 & x_m & \dots & x_m^{m-1}
    \end{pmatrix}^{-1} \begin{pmatrix}
        \hat{f}(x_1) \\ \hat{f}(x_2) \\ \dots \\ \hat{f}(x_m)
    \end{pmatrix} \notag \\
    &= \sum_k \bigg( \prod_{j\neq k} \frac{x_j}{x_j - x_k}\bigg) \hat{f}(x_k)\ .\label{eq:vandermonde system inversion}
\end{align}
Meanwhile, the true value of $f(0) = c_0$ is given by
\begin{align}
    c_0 &= \begin{pmatrix}
        1 & 0 & \dots & 0
    \end{pmatrix} \begin{pmatrix}
        1 & x_1 & \dots & x_1^{m-1} \\
        1 & x_2 & \dots & x_2^{m-1} \\
        \vdots & & \ddots & \vdots \\
        1 & x_m & \dots & x_m^{m-1}
    \end{pmatrix}^{-1} \begin{pmatrix}
        f(x_1) - E_m(x_1) \\ f(x_2) - E_m(x_2) \\ \dots \\ f(x_m) - E_m(x_m)
    \end{pmatrix}\ \\
    &= \sum_k \bigg(\prod_{j\neq k} \frac{x_j}{x_j - x_k} \bigg) \big(f(x_k) - E_m(x_k) \big)\ ,
\end{align}
and hence the error in the extrapolated value can be bounded by
\begin{align}
    |c_0 - \hat{c}_0| &= \bigg| \sum_k \bigg(\prod_{j\neq k} \frac{x_j}{x_j - x_k} \bigg) \big( f(x_k) - \hat{f}(x_k) - E_m(x_k) \big)\bigg| \notag\\
    &\leq \big(\max_k |f(x_k) - \hat{f}(x_k) | + \max_k | E_m(x_k) | \big) \sum_k \prod_{j\neq k} \bigg| \frac{x_j}{x_j - x_k}\bigg|\notag \\
    &= \bigg(\delta + \sup_{|z|\leq x_{\max}}|E_m(z)| \bigg) \alpha(x_1,\dots,x_m)\ , \label{eq:extrapolation error}
\end{align}
where we have used our assumption that $|f(x_k) - \hat{f}(x_k)| \leq \delta$, and defined the dimensionless condition number
\begin{align}
    \alpha(\vx) = \sum_k \prod_{j\neq k} \bigg| \frac{x_j}{x_j - x_k} \bigg| \ .
\end{align}
A naive choice of sampling points $\vx$, such as uniformly spaced $x_k = x_{\max} k /m$, will lead to poor conditioning $\alpha(\vx) = e^{\Omega(m)}$, making the computation extremely sensitive to noise. Better choices exist, however: in particular the Chebyshev nodes defined by
\begin{align}\label{eq:chebyshev nodes}
    x_k = x_{\max} \sin^2\bigg(\frac{(2k-1)\pi}{4m} \bigg)\ ,
\end{align}
lead to a scaling of $\alpha(\vx) = \bigO(\log m)$ (see Ref~\cite{low2019wellconditioned}). In Appendix~\ref{app:condition number} we give an elementary proof of this fact, with the concrete upper bound $\alpha(\mathbf{x})\leq 3\log(m)$ for $m\geq 2$. Combining the bound \cref{eq:extrapolation error} with the particular choice of Chebyshev nodes and \cref{lem:taylor series truncation}, we arrive at the following general extrapolation result:

\begin{theorem}[Richardson extrapolation with Taylor series truncation]\label{thm:richardson extrapolation}\noproofref
    Let $f : \CC \rightarrow \CC$ be an analytic function which is bounded as $|f(x)| \leq M$ whenever $x\in \CC$ satisfies $|x|\leq R$. Using the Chebyshev nodes
    \begin{align}
        x_k = \frac{R}{2} \sin^2\bigg(\frac{(2k-1) \pi}{4m} \bigg)\quad\text{for $k = 1,\dots,m$}\ ,
    \end{align}
    assume we have noisy samples $\hat{f}(x_k)$ such that $|\hat{f}(x_k) - f(x_k)| \leq \delta$ for all $k$, and compute $\hat{c}_0$ using \cref{eq:vandermonde system inversion}. Then the extrapolation error from the true value of $f(0)$ is bounded by
    \begin{align}\label{eq:richardson extrapolation error}
        \big|f(0) - \hat{c}_0 \big| \leq \big(\delta + 2^{-m} M \big) \cdot 3\log (m)\ .
    \end{align}
\end{theorem}

In particular, in order to bound the right-hand side of \cref{eq:richardson extrapolation error} below some $\epsilon > 0$, it is sufficient to take $m = \bigO(\log M\epsilon^{-1})$ samples at precision $\delta = \epsilon/\bigO(\log\log\epsilon^{-1})$. This leads to
\begin{align}
    x_{\min} = \min_k |x_k| = \frac{R}{\poly \log(M\epsilon^{-1})}\ .
\end{align}

\subsubsection*{Analytic approximations}

In this work, we will often be in the situation where the function $f$ we wish to extrapolate is not analytic, and hence \cref{thm:richardson extrapolation} cannot immediately be applied. Instead, our strategy will be to show that $f$ can be well-approximated by a different, genuinely analytic, function $\tilde{f}$, and to then apply \cref{thm:richardson extrapolation} to this function. It will be useful to introduce the following definition.

\begin{definition}[Analytic approximation]\label{def:analytic approximation}
    Let $f : \RR \rightarrow \CC$ be a function, and let $\delta,M,R \geq 0$. We say that $f$ has a $(\delta,M,R)$-analytic approximation if there exists an analytic function $\tilde{f} : \CC \rightarrow \CC$ such that
    \begin{align}
        \sup_{x\in [0,R]} |f(x) - \tilde{f}(x)| \leq \delta\ ,\quad \sup_{|z| \leq R} |\tilde{f}(z)| \leq M\ .
    \end{align}
\end{definition}

\begin{corollary}[Richardson extrapolation with approximate analyticity]\label{cor:richardson extrapolation approx}
    Let $f : \RR \rightarrow \CC$ be a function. Assume that $f$ has a $(\delta,M,R)$-analytic approximation. Then, for any $m\geq 1$, an estimate $\hat{c}_0$ of the value of $f(0)$ can be extrapolated from the Chebyshev nodes $\{x_k\}_{k=1}^m$ as in \cref{thm:richardson extrapolation} with sampling error $\delta$, with error
    \begin{align}
        |f(0) - \hat{c}_0| \leq (\delta + 2^{-m} M ) \bigO(\log m)\ ,
    \end{align}
    and where where $x_{\min} := \min_k x_k$ is bounded as
    \begin{align}
        x_{\min} = \Theta(R/m^2)\ .
    \end{align}
\end{corollary}

\begin{proof}[*cor:richardson extrapolation approx]
    We view the samples of $f$ as noisy samples of $\tilde{f}$, where the error is bounded by $2\delta$ by the triangle inequality. Applying \cref{thm:richardson extrapolation}, this yields an estimate $\hat{c}_0$ satisfying
    \begin{align}
        |\tilde{f}(0) - \hat{c}_0| \leq (2\delta + 2^{-m} M) \bigO(\log m)\ .
    \end{align}
    The result then follows using the fact that $|f(0) - \tilde{f}(0)|\leq \delta$, and absorbing the constant factor into the $\bigO(\log m)$.
\end{proof}

\subsection{Many-body states}\label{sec:many body}

\subsubsection*{Local Hamiltonians and $F$-norms} 

In this section we will establish some notation and technical tools for the treatment of geometrically local Hamiltonians. We focus on many-body Hamiltonians on systems of $n$ sites, with local Hilbert spaces denoted by $\cH_i$, $i \in \Gamma$, $|\Gamma| = n$; usually we will think of qubits $\cH_i \cong \CC^2$, but this restriction is not necessary. The full Hilbert space is then denoted by $\cH := \bigotimes_{i\in \Gamma} \cH_i$. Any Hamiltonian $H\in \Herm(\cH)$ can be decomposed into local terms of the form
\begin{align}
    H = \sum_{A\subseteq \Gamma} H_A\ ,
\end{align}
where each term $H_A$ acts non-trivially only on the sites $A\subseteq \Gamma$, and acts as a tensor product of identity operators on all other sites. Though such a decomposition is generally not unique, there exists a canonical choice: write $H$ as a sum over tensor products of Pauli matrices (which form a basis for $\Herm(\cH)$) and choose $H_A$ to be the sum of those acting only on $A$. The Hamiltonian $H$ is said to be \emph{$k$-local} if $H_A = 0$ whenever $|A| > k$.

Following Ref.~\cite{nachtergaele2019quasi}, we are interested in geometrically local Hamiltonians which may not be strictly $k$-local but with interactions whose strength generally decays over long distances according to some function $F$. We therefore assume the structure of a metric on $\Gamma$ denoted by $\dist : \Gamma\times \Gamma \rightarrow [0,\infty)$ encoding the physical positions of the $n$ sites. The ball of radius $r$ centred at $i\in \Gamma$ is defined as
\begin{align}
    B_r(i) := \{ j \in \Gamma : \dist(i,j) \leq r\}\ .
\end{align}
We then say that $\Gamma$ has dimension $D \geq 1$ if there exists a constant $k_D > 0$ independent of $n$ such that 
\begin{align}\label{eq:ball volume bound}
    \sup_{i\in \Gamma,r\geq 1} |B_r(i)| r^{-D} \leq k_D\ .
\end{align}

Given a non-increasing function $F : [0,\infty) \rightarrow [0,\infty)$, the $\|\cdot\|_F$-norm of a Hamiltonian is defined as
\begin{align}\label{eq:F norm}
    \|H\|_F := \sup_{i,j\in \Gamma} \frac{1}{F(\dist(i,j))} \sum_{\substack{A\subseteq \Gamma \\ i,j \in A}} \|H_A\|\ .
\end{align}
Intuitively, this definition ensures that the total strength of interactions coupling sites $i,j\in \Gamma$ decays as $F(\dist(i,j))$, since
\begin{align}\label{eq:total interaction F norm}
    \sum_{\substack{A\subseteq \Gamma \\ i,j \in A}} \|H_A\| \leq F(\dist(i,j)) \|H\|_F\ .
\end{align}

\subsubsection*{Physical states and criticality}

Given a family of Hamiltonians $\{H(x)\}_{x \in \cX}$ (for some possibly multidimensional parameter space $\cX$), we will typically denote the ground state of $H(x)$ (when this is well-defined) by $\ket{\psi_0(x)}$, and the Gibbs state at inverse temperature $\beta \geq 0$ by $\rho_\beta(x) := e^{-\beta H(x)} / Z_\beta(x)$, where $Z_\beta(x) = \tr[e^{-\beta H(x)}]$ is the partition function. We can generally only expect extrapolation of physical properties to be reliable away from phase transitions (see e.g. Ref.~\cite{harrow2020classical} for a different connection between analyticity and correlation decay in thermal states). For ground states, we characterise criticality in terms of the spectral gap; ground state $\ket{\psi_0(x)}$ does not undergo a phase transition if the Hamiltonians $H(x)$ have a constant gap between the ground and first excited energies. We summarise this with the following definition.

\begin{definition}[Uniformly gapped ground states]\label{def:uniformly gapped ground}
    A family of Hamiltonians $\{H(x)\}_{x\in \cX}$ is uniformly gapped with gap $\gamma > 0$ if, for all $x\in \cX$, the ground space of $H(x)$ is one-dimensional and its two lowest eigenvalues $\lambda_0(x) < \lambda_1(x)$ are separated by $\lambda_1(x) - \lambda_0(x) > \gamma$.
\end{definition}

The assumption that the ground space is one-dimensional is necessary to extrapolate ground state properties, however it is not necessary if one aims to extrapolate the ground state energy (as $\lambda_0(x)$ may be is well-defined even when $\ket{\psi_0(x)}$ is not). For Gibbs states, we characterise non-criticality in terms of correlation decay:

\begin{definition}[Exponential decay of correlations]\label{def:exp correlation decay}
    A state $\rho$ on $\otimes_{x\in \Gamma} \cH_x$ is said to satisfy an exponential decay of correlations with parameters $K,\xi > 0$ if, for all observables $M_A$ and $N_B$ supported on disjoint subsystems $A,B\subseteq \Gamma$, we have
    \begin{align}
        \frac{|\tr[\rho(M_A\otimes N_B)] - \tr[\rho M_A] \tr[\rho N_B] |}{\|M_A\| \|N_B\|} \leq K|A| |B | e^{-\dist(A,B) / \xi}\ .
    \end{align}
\end{definition}

Often (see e.g. Ref.~\cite{alhambra2023quantum}), a stronger definition is used with the sizes of the boundaries $|\partial A|$ and $|\partial B|$ in place of $|A|$ and $|B|$, but this is not necessary for our purposes as we will typically use this assumption in situations when one or both of $|A|$ and $|B|$ is constant. Note that gapped ground states exhibit exponential decay of correlations as in Definition~\ref{def:exp correlation decay} by Ref.~\cite{hastings2006spectral}, as do Gibbs states at sufficiently high temperature \cite{alhambra2023quantum,frohlich2015some}. The following definition, analogously to \cref{def:uniformly gapped ground}, extends this assumption to hold uniformly over a family of states.

\begin{definition}[Uniform exponential decay of correlations]\label{def:uniformcorrelationdecay}
    A family of states $\{\rho(x)\}_{x\in \cX}$ is said to satisfy a uniform exponential decay of correlations with parameters $K,\xi>0$ if $\rho(x)$ has exponentially decaying correlations as in Definition~\ref{def:exp correlation decay} for all $x \in \cX$.
\end{definition}

\subsubsection*{Quantum belief propagation and the spectral flow}

Consider the Gibbs (respectively gapped ground) states of a family of Hamiltonians $\{H(x)\}_{x\in \cX}$ on the sites $\Gamma$, whose local terms vary smoothly with $x$. It turns out that, under such variation, the Gibbs (gapped ground states) change in a predictable manner described by the framework of quantum belief propagation \cite{hastings2007quantum,kim2012perturbative,anshu2021sample,rouze2024efficient} (respectively the spectral flow \cite{hastings2005quasiadiabatic,bachmann2012automorphic,nachtergaele2019quasi}). In both cases, the change in the Gibbs (gapped ground) state is controlled by an operator obtained by locally transforming $\partial_x H(x)$. We begin by stating the result for Gibbs states:

\begin{lemma}[Quantum belief propagation \cite{hastings2007quantum,kim2012perturbative}]\label{lem:qbp}
    Let $\{H(x)\}_{x \in \cX}$ be a family of Hamiltonians which smoothly depends on $x$ on the interval $\cX \subseteq \RR$. Then the derivative of the Gibbs state $\rho_\beta(x) = e^{-\beta H(x)}/\tr[e^{-\beta H(x)}]$ is given by
    \begin{align}\label{eq:qbp evolution}
        \frac{\diff}{\diff x} \rho_\beta(x) = -\frac{1}{2} \beta \left\{ \rho_\beta(x),\Phi_{H(x)}(\partial_x H(x))\right\} + \beta \rho_\beta(x) \tr\left[\rho_\beta(x) \Phi_{H(x)} (\partial_x H(x))\right]\ ,
    \end{align}
    where $\Phi_H(\cdot)$ is the quantum belief propagation operator defined by
    \begin{align}\label{eq:qbp operator}
        \Phi_H(X) := \int_{-\infty}^\infty \diff t \kappa_\beta(t) e^{-itH} X e^{itH}\ ,
    \end{align}
    and $\kappa_\beta(t)$ is a function which decays exponentially away from $t=0$, explicitly given by (see Ref.~\cite{anshu2021sample})
    \begin{align}\label{eq:kappa function}
        \kappa_\beta(t) = \frac{2}{\pi \beta} \log\frac{e^{\pi |t| / \beta} + 1}{e^{\pi|t| / \beta} - 1} \leq \frac{4}{\pi\beta} \cdot \frac{1}{e^{\pi |t| /\beta} - 1}\ .
    \end{align}
\end{lemma}
\begin{proof}[*lem:qbp]
    An elementary proof of this fact is given, for example, in Ref.~\cite{alhambra2023quantum} Appendix B. The statement there only covers the case where $H(x)$ is linear, $H(x) = H(0) + x A$, but this is not a necessary restriction: the first step in the proof involves using Duhamel's identity to write
    \begin{align}
        \frac{\diff}{\diff x} e^{-\beta H(x)} = -\beta \int_0^1\diff t e^{-\beta t H(x)} \partial_x H(x) e^{-\beta (1-t) H(x)}\ .
    \end{align}
    This is valid for any smooth function $H(x)$, and the remainder of the proof follows unchanged with $\partial_x H(x)$ in place of $A$.
\end{proof}

It is an immediate consequence (see Ref.~\cite{rouze2024efficient}) of Eq.~\eqref{eq:qbp evolution} that, for any observable $O$ with expectation value $f_\beta(x) := \tr[O\rho_\beta(x)]$, the derivative of $f_\beta(x)$ is given by
\begin{align}\label{eq:qbp observable}
    \frac{\diff}{\diff x} f_\beta(x) = -\beta \Cov_{\rho_\beta(x)} \left(O,\Phi_{H(x)} (\partial_x H(x)) \right) = -\beta \Cov_{\rho_\beta(x)} \left( \Phi_{H(x)}(O), \partial_x H(x) \right) \ ,
\end{align}
where $\Cov_\rho(X,Y)$ is the operator covariance
\begin{align}
    \Cov_\rho(X,Y) := \frac{1}{2} \tr[\rho\{X,Y\}] - \tr[\rho X] \tr[\rho Y]\ .
\end{align}
Since uniform correlation decay as in \cref{def:uniformcorrelationdecay} gives an upper bound on the covariance between spatially separated observables, this form can be leveraged to localise the dependence of $f_\beta(x)$ on $\partial_x H(x)$ to terms in a small region around $O$ (we formalise this fact in Lemma~\ref{lem:GALI gibbs}).

For ground states, we have the following qualitatively similar (though technically distinct) result.

\begin{lemma}[Spectral flow --- see Ref.~\cite{bachmann2012automorphic}, Proposition 2.4]\label{lem:spectral flow}\noproofref
    Let $\{H(x)\}_{x\in \cX}$ be a uniformly gapped family of Hamiltonians with continuous first derivative with respect to $x$ on the interval $\cX$. Let $\ket{\psi_0(x)}$ denote the ground state of $H(x)$, and let $\gamma > 0$ be the uniform gap. Then the derivative of $\ket{\psi_0(x)}$ with respect to $x$ is given by
    \begin{align}\label{eq:spectral flow evolution}
        \frac{\diff}{\diff x} \ket{\psi_0(x)} = i\Psi_{H(x)}(\partial_x H(x)) \ket{\psi_0(x)}\ ,
    \end{align}
    where $\Psi_H(\cdot)$ is the spectral flow operator defined by
    \begin{align}\label{eq:spectral flow operator}
        \Psi_H(X) := \int_{-\infty}^\infty \diff t w_\gamma(t) \int_0^t \diff u e^{iuH} X e^{-iuH}\ ,
    \end{align}
    and $w_\gamma(t)$ is any real-valued $L_1$ function satisfying $\int_{-\infty}^\infty \diff t w_\gamma(t) = 1$ and whose Fourier transform $\hat{w}_\gamma$ is supported in the interval $[-\gamma,\gamma]$.
\end{lemma}

Note that both the quantum belief propagation and spectral flow operators $\Psi_H(X)$ correspond to a ``smearing'' of the operator $X$, under local dynamics controlled by $H$, and hence their supports are approximately localised around the support of $X$ (we formalise this intuition with Lemma~\ref{lem:qbp and spectral flow truncation}). Although we can intuitively view the spectral flow as the extension of quantum belief propagation to the $\beta \rightarrow \infty$ case, note that this is not how it is obtained (Eq.~\eqref{eq:qbp evolution} is unbounded in this limit); the spectral gap is crucial.

In \cite{bachmann2012automorphic}, an explicit family of such functions $w_\gamma$ is given which decay quickly away from $t=0$. Choosing this family, we can assume without loss of generality that $w_\gamma(t)$ is even, non-negative, and that when $|t| \geq e^{-1/\sqrt{2}} \gamma^{-1}$, we have the bounds
\begin{align}\label{eq:wgamma inequality}
    0 \leq w_\gamma(t) \leq 2(e\gamma)^2 |t| \exp \left(-\frac{2\gamma|t|}{7\log^2 (\gamma|t|)} \right)\ .
\end{align}
Moreover, for an observable $O$ with ground state expectation value $f_{\ground}(x) := \bra{\psi_0(x)} O \ket{\psi_0(x)}$, we can see from Eq.~\eqref{eq:spectral flow evolution} that the derivative of $f_{\ground}(x)$ is given by
\begin{align}\label{eq:spectral flow observable derivative}
    \frac{\diff}{\diff x} f_{\ground}(x) = \bra{\psi_0(x)} i[O,\Psi_{H(x)}(\partial_x H(x))] \ket{\psi_0(x)} = \bra{\psi_0(x)} i [\Psi_{H(x)}(O),\partial_x H(x)]\ket{\psi_0(x)}\ .
\end{align}
The absolute value of this derivative can thus be bounded by the operator norm of the commutator
\begin{align}
    \left[ \Psi_{H(x)}(O), \partial_x H(x) \right]\ .
\end{align}
In \cref{sec:correlation spread}, we will use Lieb-Robinson bounds to show that $\Psi_{H(x)}(O)$ is approximately localised around the support of $O$, and hence that Eq.~\eqref{eq:spectral flow observable derivative} only depends on the variation of $H(x)$ supported near to $O$. For this purpose, it will be useful to prove the following alternative representation of the spectral flow operator:

\begin{lemma}[Alternative representation of spectral flow --- see Lemma~\ref{lem:spectral flow alternative app}]\label{lem:spectral flow alternative}\noproofref
    The spectral flow operator $\Psi_H(O)$ as defined in \cref{lem:spectral flow} can be written as
    \begin{align}
        \Psi_H(X) = \int_{-\infty}^\infty \tilde{w}_\gamma(t) e^{itH} X e^{-itH}\ ,
    \end{align}
    where $\tilde{w}_\gamma : \RR \rightarrow \RR$ is an odd $L_1$ function with $|\tilde{w}_\gamma(t)| \leq 1/2$ for all $t$. Moreover, for $|t|\geq e^3 \gamma^{-1}$,
    \begin{align}
        |\tilde{w}_\gamma(t)| &\leq W_1 \left(\frac{\gamma t}{\log^2(\gamma t)} \right)^2  \exp\left( -\frac{2\gamma t}{7\log^2(\gamma t)}\right) \ ,\\
        \int_t^\infty \diff s \tilde{w}_\gamma(s) &\leq W_2 \gamma^{-1} \left( \frac{\gamma t}{\log^2(\gamma t)}\right)^3 \exp\left(-\frac{2\gamma t}{7\log^2(\gamma t)}\right)\ ,
    \end{align}
    for some constants $W_1,W_2 > 0$. Moreover,
    \begin{align}
        \int_{-\infty}^\infty \diff t | \tilde{w}_\gamma(t) |\leq W_3 \gamma^{-1} \ ,
    \end{align}
    for some constant $W_3 > 0$.
\end{lemma}

\subsection{Spread of correlations}\label{sec:correlation spread}

\subsubsection*{Lieb-Robinson bounds}

In order to analyse the locality properties of the quantum belief propagation and spectral flow operators introduced above, we require some general results about the localisation of operators under short-time dynamics \cite{lieb1972finite}. We will essentially follow the formalism of Ref.~\cite{nachtergaele2019quasi}, in which the authors prove very general results (which are also applicable to infinite lattices, though this case is not necessary for our purposes). See also Ref.~\cite{cubitt2015stability}, in which the authors use similar techniques with generalisation to dissipative dynamics.

Before we can state the Lieb-Robinson bounds, we adopt the terminology of Ref.~\cite{nachtergaele2019quasi} and define an $F$-function as below.

\begin{definition}[$F$-function \cite{nachtergaele2019quasi}]\label{def:f function}
    A non-increasing function $F : [0,\infty) \rightarrow [0,\infty)$ is called an $F$-function on $(\Gamma,\dist)$ with parameters $(\|F\|,C_F)$, where we define
    \begin{align}
        \|F\| &:= \sup_{i\in \Gamma} \sum_{j \in \Gamma} F(\dist(i,j))\ , \\
        C_F &:= \sup_{i,j\in \Gamma} \sum_{k\in \Gamma} \frac{F(\dist(i,k))F(\dist(k,j))}{F(\dist(i,j))}\ .\label{eq:cf constant}
    \end{align}
    When $\|F\|$ and $C_F$ are both constants independent of the size of the lattice $n$, we will simply refer to $F$ as an $F$-function. If $C_F = 1$, we say that $F$ is normalised.
\end{definition}

For our results in Sections~\ref{sec:perturbation theory}-\ref{sec:simulator extrapolation}, we will assume that all $F$-functions are normalised (and hence ignore any factors of $C_F$). This only affects our conclusions up to constant factors, since any $F$-function $F$ can be be normalised via $F \mapsto C_F^{-1} F$.

The constants $\|F\|$ and $C_F$ are both necessary to ensure that a Hamiltonian with bounded $\|\cdot\|_F$-norm is sufficiently localised to satisfy Lieb-Robinson bounds. Informally, $\|F\|$ gives a bound on the total strength of interactions acting on any given $i\in \Gamma$ for a Hamiltonian $H = \sum_{A\subseteq \Gamma} H_A$, as
\begin{align}
    \sum_{\substack{A \subseteq \Gamma \\ i \in A}} \|H_A\| &\leq \sum_{j\in \Gamma} \sum_{\substack{A\subseteq \Gamma \\ i,j \in A}} \|H_A\| \\
    &\leq \sum_{j\in \Gamma} F(\dist(i,j)) \|H\|_F \\
    &\leq \|H\|_F \|F\|\ .\label{eq:local term bound}
\end{align}
Meanwhile, the constant $C_F$ will later be important to ensure that higher-order functions of local Hamiltonians remain local (for a concrete statement, see Lemma~\ref{lem:Fnormcommutator}).

The phrase ``constants independent of the size of the lattice'' is not entirely well defined, as we consider a finite system $\Gamma$. This can be understood by viewing $\Gamma$ as a member of a family of systems of varying size $n$, or as a subset of size $n$ of the $D$-dimensional square lattice $\Gamma \subseteq \ZZ^D$. In either case, we then interpret these conditions as meaning that $\|F\|$ and $C_F$ are bounded as $n$ is taken infinitely large.

It is not immediately clear how to construct functions satisfying \cref{def:f function}, but there are two important examples for our purposes, as noted in~\cite{nachtergaele2019quasi}. Firstly, for the $D$-dimensional systems that we consider, the function $F(r) = (1+r)^{-(D+\epsilon)}$ defines an $F$-function for any $\epsilon > 0$, where $C_F \leq 2^{D+\epsilon} \|F\|$. Moreover, given any $F$-function $F$ and a non-decreasing and subadditive (i.e. satisfying $g(r + s) \leq g(r) + g(s)$) function $g : [0,\infty) \rightarrow [0,\infty)$, the weighted function
\begin{align}
    F_g(r) := F(r) e^{-g(r)}
\end{align}
also defines an $F$-function. Choosing $g(r) = ar$ for constant $a$, this allows us to describe exponentially decaying interactions. We formalise this in the below definition:
\begin{definition}[Exponentially decaying interactions]\label{def:exp decaying interactions}
    A Hamiltonian $H$ on $\Gamma$ in $D$ dimensions is said to have exponentially decaying interactions with decay rate $a > 0$ if $\|H\|_{F_g}$ is bounded by a constant, where $F_g$ is the following $F$-function:
    \begin{align}
        F_g(r) := (1+r)^{-(D+1)} e^{-ar}\ .
    \end{align}
\end{definition}

We can generalise the bound in Eq.~\eqref{eq:local term bound} to bound the error incurred in operator norm when the Hamiltonian $H$ is restricted to a sublattice $\Gamma' \subseteq \Gamma$, as below:
\begin{lemma}[Hamiltonian restriction]\label{lem:hamiltonian restriction}
    Let $H$ be a Hamiltonian on $\Gamma$ with bounded $\|\cdot\|_F$-norm for an $F$-function $F$. For any subset $\Gamma'\subseteq\Gamma$, we define the restriction of $H$ to $\Gamma'$ by
    \begin{align}
        H|_{\Gamma'} := \sum_{A\subseteq \Gamma'} H_A\ .
    \end{align}
    Then the difference between $H$ and $H|_{\Gamma'}$ is bounded by
    \begin{align}
        \|H - H|_{\Gamma'}\| \leq |\Gamma \setminus \Gamma'| \|H\|_F \|F\|\ .
    \end{align}
\end{lemma}

\begin{proof}[*lem:hamiltonian restriction]
    We have
    \begin{align}
        H - H|_{\Gamma'} = \sum_{\substack{A\subseteq \Gamma \\ A \cap (\Gamma \setminus \Gamma') \neq \emptyset}} H_A\ ,
    \end{align}
    and hence
    \begin{align}
        \|H - H|_{\Gamma'}\| &\leq \sum_{x\in \Gamma\setminus \Gamma'} \sum_{\substack{A \subseteq \Gamma \\ x \in A}} \|H_A \| \\
        &\leq |\Gamma \setminus \Gamma'| \|H\|_F \|F\|\ .
    \end{align}
    where in the second line we used Eq.~\eqref{eq:local term bound}.
\end{proof}

With these definitions in hand, we are ready to state the Lieb-Robinson bounds we will use. We consider a time-dependent Hamiltonian $H(t) = \sum_{A\subseteq \Gamma} H_A(t)$, an $F$-function $F(r)$, and a function $g(r)$ as above. We use the shorthand $\|H\|_{F_g} := \sup_t \|H(t)\|_{F_g}$ (and assume that this is constant). Let $O_A$ and $O_B$ be two observables with disjoint supports $A,B \subseteq \Gamma$ separated by distance $\dist(A,B) = r$, and denote by $O_A(t)$ the time-evolved observable for $t\geq 0$, satisfying
\begin{align}
    O_A(0) = O_A\ ,\quad \frac{\diff}{\diff t} O_A(t) = i[H(t),O_A(t)]\ .
\end{align}

\begin{lemma}[Lieb-Robinson bounds --- see Ref.~\cite{nachtergaele2019quasi}, Theorem 3.1]\label{lem:lr bounds}\noproofref
    There exist constants (depending on $C_{F_g}$ and $\|F_g\|$) $c,\nu > 0$ such that, for all $t > 0$,
    \begin{align}
        \| [O_A(t),O_B] \| \leq c\|O_A\| \|O_B\| \min\{|A|,|B|\} \left( e^{\nu \|H\|_{F_g} t} - 1\right) e^{-g(r)}\ ,
    \end{align}
    where $r = \dist(A,B)$.
\end{lemma}

\subsubsection*{Local truncations}

We will now use the Lieb-Robinson bounds of the previous section to show that, for a local observable $O$ and quasi-local Hamiltonian $H$, the operators $\Phi_H(O)$ and $\Psi_H(O)$ are localised around the support of $O$. For a set of sites $A \subseteq \Gamma$, we write $B_r(A)$ to denote the ball of radius $r$ around $A$, that is
\begin{align}
    B_r(A) = \{ i \in \Gamma : \dist(i,A) \leq r\}\ .
\end{align}
The below result, Lemma~\ref{lem:qbp and spectral flow truncation}, establishes that whenever $O_A$ is supported in $A\subseteq \Gamma$, $\Phi_H(O_A)$ and $\Psi_H(O_A)$ have good approximations which only act on $B_r(A)$. See also Ref.~\cite{rouze2024efficient}, Lemma III.2 for essentially the same result in the case of quantum belief propagation (a similar statement also appears in Ref.~\cite{kim2012perturbative}, Corollary 3), and Ref.~\cite{nachtergaele2019quasi} Section 6.5 for similar results establishing quasi-locality of the spectral flow. We defer the proof of this result to Appendix~\ref{app:truncation results}.

\begin{lemma}[Truncation of the quantum belief propagation and spectral flow operators --- see Lemma~\ref{lem:qbp and spectral flow truncation app}]\label{lem:qbp and spectral flow truncation}\noproofref
    Let $H$ be a Hamiltonian with bounded $\|\cdot\|_{F_g}$-norm, and let $O_A$ be an observable supported on $A \subseteq \Gamma$. Let $\Phi_H(O_A)$ be the quantum belief propagation operator defined by Eq.~\eqref{eq:qbp operator}, and $\Psi_H(O_A)$ be the spectral flow operator defined by Eq.~\eqref{eq:spectral flow operator}. Then for every $r\geq0$ there exist operators $\Phi_H^{[r]}(O_A)$ and $\Psi_H^{[r]}(O_A)$ which act as the identity outside of $B_r(A)$, and constants $a_1,b_1> 0$ (depending on $c,\beta,\nu,\|H\|_{F_g}$) and $a_2,b_2 > 0$ (depending on $c,\gamma,\nu,\|H\|_{F_g}$), such that
    \begin{align}
        \|\Phi_H(O_A) - \Phi_H^{[r]}(O_A) \| &\leq a_1 |A| \|O_A\| e^{-b_1 g(r)}\ , \label{eq:truncation w constants qbp}\\
        \|\Psi_H(O_A) - \Psi_H^{[r]}(O_A) \| &\leq a_2 |A| \|O_A\| e^{-b_2 g(r) /\log^2 g(r)}\ ,\label{eq:truncation w constants sf}
    \end{align}
    where $\omega(x)$ is defined by
    \begin{align}
        \omega(x) := \left( \frac{x}{\log^2(x)}\right)^3 \exp\left( -\frac{2x}{7\log^2 x}\right)\ .
    \end{align}
\end{lemma}

\subsubsection*{Generalised approximate local indistinguishability}

This ability to truncate the action of the quantum belief propagation and spectral flow operators allows us to formalise the notion that local observables are not strongly affected by a distantly changing Hamiltonian in thermal and ground states away from criticality. This phenomenon is known as generalised approximate local indistinguishability (GALI), and was shown in Ref.~\cite{rouze2024efficient} to hold for Gibbs states with exponentially decaying correlations and gapped ground states. The results below are essentially an adaptation of their results to our framework.

\begin{lemma}[GALI for Gibbs states --- \cite{rouze2024efficient} Proposition V.4 paraphrased --- see Lemma~\ref{lem:GALI gibbs app}]\label{lem:GALI gibbs}\noproofref
    Let $\{H(x)\}_{x\in \cX}$ be a family of Hamiltonians with continuous first derivative with respect to $x$ on the interval $\cX$, such that $\|H\|_{F_g}$ and $\|\partial_x H\|_{F_g}$ are both bounded. Let $O_A$ be an observable supported on $A\subseteq \Gamma$, and assume that $\partial_s H(s)$ contains no terms with support in $B_{r_0}(A)$ for some $r_0\geq 0$.

    Let $\rho_\beta(x)$ be the associated family of Gibbs states at temperature $\beta = \bigO(1)$, and assume that these satisfy a uniform exponential decay of correlations with parameters $K,\xi > 0$. Then, for all $x$, it holds that
    \begin{align}
        \left|\Cov_{\rho_\beta(x)} (\Phi_{H(x)}(O_A),\partial_x H(x)) \right| \leq c_1\|O_A\| |A|^3 \sum_{r = r_0}^\infty \left( r^{3D} e^{-r/2\xi} + r^D e^{-b_1 g(r/2)} \right)\ ,
    \end{align}
    for positive constants $b_1,c_1>0$.
    Hence, defining $f_\beta(x) := \tr[O_A\rho_\beta(x)]$, we have for all $x_0,x_1 \in \cX$ that
    \begin{align}\label{eq:GALI gibbs function}
        |f_\beta(x_1) - f_\beta(x_0)| &\leq \beta |x_1 - x_0| c_1 \|O_A\| |A|^3 \sum_{r = r_0}^\infty \left( r^{3D} e^{-r/2\xi} + r^D e^{-b_1 g(r/2)} \right)\ .
    \end{align}
\end{lemma}
Notice that in the case of exponentially decaying correlations (that is, when $g(r) = \Omega(r)$), the right-hand side of Eq.~\eqref{eq:GALI gibbs function} decays exponentially with $r_0$. This will allow us to estimate $f_\beta(x)$ to accuracy $\epsilon > 0$ whilst ignoring the variation of $H(x)$ outside a ball of radius $r_0 \sim \log \epsilon^{-1}$. We state the corresponding result for gapped ground states below.

\begin{lemma}[GALI for ground states --- \cite{rouze2024efficient} Proposition V.5 paraphrased (see also \cite{lewis2024improved}) --- see Lemma~\ref{lem:GALI ground app}]\label{lem:GALI ground}\noproofref
    Let $\{H(x)\}_{x\in \cX}$ and $O_A$ be as in \cref{lem:GALI gibbs}, and assume that $\{H(x)\}_{x\in \cX}$ has a uniform gap $\gamma > 0$ above its ground state $\ket{\psi_0(x)}$. Then
    \begin{align}
        \left\| [\Psi_{H(x)}(O_A), \partial_x H(x)]\right\| \leq c_2 \|O_A\| |A|^2 \sum_{r=r_0}^\infty r^D e^{-b_2g(r-1)/\log^2 g(r-1)}\ ,
    \end{align}
    for positive constants $b_2,c_2 >0$. Hence, defining $f_{\ground}(x) := \bra{\psi_0(x)} O_A \ket{\psi_0(x)}$, we have for all $x_0,x_1 \in \cX$ that
    \begin{align}\label{eq:GALI ground function}
        |f_{\ground}(x_1) - f_{\ground}(x_0)| \leq |x_1 - x_0| c_2 \|O_A\| |A|^2 \sum_{r=r_0}^\infty r^D e^{-b_2g(r-1)/\log^2 g(r-1)}\ .
    \end{align}
\end{lemma}

In particular, in the case of exponential decaying interactions $g(r) = \Omega(r)$, the right-hand side of Eq.~\eqref{eq:GALI ground function} decays as $\sim \exp(-\Omega(r/\log^2 r))$.

\section{Extrapolation within phases of matter}\label{sec:basic extrapolation}

\subsection{Assumptions and main statement}

In this section, we establish the first set of results for this work: namely, that local properties of Gibbs states and ground states have good analytic approximations (and hence can be extrapolated) along paths of non-critical Hamiltonians. This will establish some basic results and proof techniques which we will ultimately combine with the perturbation theory machinery in Section~\ref{sec:perturbation theory} to prove our main results about extrapolating simulator Hamiltonians in Section~\ref{sec:simulator extrapolation}. The main conclusions of this section are stated below as Theorem~\ref{thm:extrapolation without gadgets}. The proof of the result is divided into its two parts: Theorems~\ref{thm:gibbs state extrapolation} and \ref{thm:ground state extrapolation} deal with the Gibbs state and ground state results respectively. The results are then proved separately in Sections~\ref{sec:extrapolation gibbs} and \ref{sec:extrapolation ground}. A generalisation of Theorem~\ref{thm:extrapolation without gadgets} to the case of extensive observables is given in the appendix; see Theorem~\ref{thm:extensive extrapolation}.

Firstly, we make precise the non-criticality assumptions required for our main results, collected below as Assumption~\ref{assumption:non criticality} for convenience. In addition to geometric locality and analytic dependence on the parameter $x$, we assume that the family of Hamiltonians has a uniform exponential decay of correlations (respectively a spectral gap, for the ground state case) which is robust even when the interactions contained within some region $\Gamma' \subseteq \Gamma$ are parametrised differently to those outside of $\Gamma'$. Ultimately, this will be necessary for us to argue that we can continuously ``turn off'' the interactions far from the local observable of interest, using the GALI properties from Section~\ref{sec:correlation spread}.

\begin{assumption}\label{assumption:non criticality}
    Let $H(x) = \sum_{A\subseteq \Gamma} H_A(x)$ be a family of Hamiltonians parametrised by $x\in \cX$. Assume that each $H_A(x)$ is an analytic function of $x$, and that both $H(x)$ and $\partial_x H(x)$ have exponentially decaying interactions as in Definition~\ref{def:exp decaying interactions} (that is, $\|H(x)\|_{F_g}, \|\partial_x H(x)\|_{F_g} = \bigO(1)$ for an $F$-function $F$ and for $g(r) = ar$ where $a > 0$) whenever $x \in \{z\in \CC : |z| \leq x_{\ast}\}$, for some $x_{\ast}>0$. For every $\Gamma'\subseteq \Gamma$, define the family of Hamiltonians $\{H(x,y;\Gamma')\}_{(x,y)\in \cX\times \cX}$ by
    \begin{align}\label{eq:Hxygammaprime}
        H(x,y;\Gamma') := \sum_{A\subseteq \Gamma'} H_A(x) + \sum_{\substack{A\subseteq \Gamma \\
        A\nsubseteq \Gamma'}} H_A(y)\ .
    \end{align}
    Then we further assume that one of the following holds, depending on the extrapolation task:
    \begin{enumerate}[(I)]
        \item\label{it:gibbs assumption} \emph{(For Gibbs states)}: There exist constants $K,\xi > 0$ such that, for all $\Gamma'\subseteq \Gamma$, the family $\{\rho_\beta(x,y;\Gamma')\}_{(x,y)\in \cX\times \cX}$ satisfies a uniform exponential decay of correlations with parameters $K,\xi$ (see Definition~\ref{def:uniformcorrelationdecay}). 
        \item\label{it:ground assumption} \emph{(For ground states)}: There exists a constant $\gamma > 0$ such that, for all $\Gamma' \subseteq \Gamma$, the family $\{H(x,y;\Gamma')\}_{(x,y)\in \cX\times \cX}$ is uniformly gapped above the unique ground state with gap $\gamma$ (see Definition~\ref{def:uniformly gapped ground}).
    \end{enumerate}
\end{assumption}

We can now state the main result of the section. This is stated in terms of analytic approximations of local observables: for any chosen $\delta > 0$, we establish the existence of $(\delta,M,R)$-analytic approximations, for $M,R>0$ depending on $\delta$. 

\begin{theorem}[See Theorems~\ref{thm:gibbs state extrapolation} and \ref{thm:ground state extrapolation}]\label{thm:extrapolation without gadgets}\noproofref
    Let $\{H(x)\}_{x\in \cX}$ be a family of Hamiltonians on $\Gamma$ in $D$ dimensions satisfying the conditions of Assumption~\ref{assumption:non criticality}. Let $O_A$ be an observable supported on $A \subseteq \Gamma$, where $|A|,\|O_A\| = \bigO(1)$. Then the following holds:
    \begin{enumerate}[(I)]
        \item\label{it:gibbs extrapolation without gadgets} Let $\rho_\beta(x)\sim e^{-\beta H(x)}$ be the family of Gibbs states associated to $H(x)$ at inverse temperature $\beta > 0$. Then (with Assumption~\ref{assumption:non criticality}\eqref{it:gibbs assumption}) the function
        \begin{align}\label{eq:fbeta definition without gadgets}
            f_\beta(x) := \tr[O_A \rho_\beta(x))]\ ,\quad x\in [0,x_{\ast}]\ ,
        \end{align} 
        has a $(\delta,M,R)$-analytic approximation for any $\delta > 0$, where
        \begin{align}
            M = 3\|O_A\| \ ,\quad x_{\ast} \geq R = 1/\bigO(\log^D(\delta^{-1}))\ .
        \end{align}
        \item\label{it:gs extrapolation without gadgets} Let $\ket{\psi_0(x)}$ be the unique ground state of $H(x)$. Then (with Assumption~\ref{assumption:non criticality}\eqref{it:ground assumption}) the function
        \begin{align}\label{eq:fground definition without gadgets}
            f_{\ground}(x) := \bra{\psi_0(x))} O_A \ket{\psi_0(x)} \ ,\quad x\in [0,x_{\ast}]\ ,
        \end{align}
        has a $(\delta,M,R)$-analytic approximation for any $\delta > 0$, where
        \begin{align}
            M = 16 \|O_A\| \ ,\quad x_{\ast} \geq R = 1/\bigO(\log^{D+1}(\delta^{-1}))\ .
        \end{align}
    \end{enumerate}
\end{theorem}

In particular, this theorem implies --- by Corollary~\ref{cor:richardson extrapolation approx} --- that the functions $f_\beta(x)$ and $f_{\ground}(x)$ can be extrapolated to $x=0$ with accuracy $\epsilon > 0$ from their evaluations at values of $x$ lower bounded by $1/\poly\log (\epsilon^{-1})$. We make this concrete with the following corollary:

\begin{corollary}[Extrapolating local observables within phases of matter]\label{cor:extrapolating without gadgets explicit}
    The value of $f_\beta(0)$ (respectively $f_{\ground}(0)$) can be calculated up to any desired accuracy $\epsilon > 0$, given the values of $f_\beta(x_k)$ (respectively $f_{\ground}(x_k)$) at $m$ Chebyshev sample points $\{x_k\}_{k=1}^m$, where each $x_k$ is bounded above zero by $x_{\min} := \min_k x_k$, where:
    \begin{enumerate}[(I)]
        \item For $f_\beta$,
        \begin{align}
            m = \bigO(\log(\epsilon^{-1}))\ ,\quad x_{\min} = \frac{1}{\bigO(\log^{D+2}(\epsilon^{-1}))}\ .
        \end{align}
        \item For $f_{\ground}$,
        \begin{align}
            m = \bigO(\log(\epsilon^{-1}))\ ,\quad x_{\min} = \frac{1}{\bigO(\log^{D+3}(\epsilon^{-1}))}\ .
        \end{align}
    \end{enumerate}
    The result still holds using noisy estimates of $f_\beta(x_k)$ (respectively $f_{\ground}(x_k)$) each with additive error $\delta = \Theta(\epsilon/\log\log(\epsilon^{-1}))$.
\end{corollary}

\begin{proof}[*cor:extrapolating without gadgets explicit]
    Corollary~\ref{cor:richardson extrapolation approx} ensures that we can extrapolate to within error $\epsilon > 0$, given a $(\delta,M,R)$-analytic approximation and $m$ Chebyshev samples, as long as
    \begin{align}
        \epsilon = (\delta + 2^{-m} M) \bigO(\log m)\ .
    \end{align}
    For both $f_{\beta}$ and $f_{\ground}$, we can obtain the required approximations from Theorem~\ref{thm:extrapolation without gadgets} with $M = \bigO(1)$, and hence a choice of $m = \bigO(\log(\epsilon^{-1}))$ and $\delta = \Theta(\epsilon / \log\log(\epsilon^{-1}))$ suffices. This leads to $R = 1/\bigO(\log^D(\epsilon^{-1}))$ and $R = 1/\bigO(\log^{D+1}(\epsilon^{-1}))$ for $f_{\beta}$ and $f_{\ground}$ respectively. We then have, from Corollary~\ref{cor:richardson extrapolation approx}, that $x_{\min} \sim R/m^2$, from which the stated result follows.
\end{proof}

\subsection{Gibbs states}\label{sec:extrapolation gibbs}

\subsubsection*{Analyticity of perturbations on Gibbs states}

In order to establish the Gibbs state case of \cref{thm:extrapolation without gadgets}, the proof will proceed in two parts. Firstly, we use the GALI property of \cref{lem:GALI gibbs} to show that the Hamiltonian perturbation can be approximately truncated to a spatially localised region. Next, we need to show that sufficiently small (in operator norm) perturbations to the Hamiltonian result in perturbations to the Gibbs state which can be extended analytically. To this end, we need to prove the following theorem:

\begin{theorem}\label{thm:gibbs small pert}
    Let $H \in \Herm(\mathcal{H})$ be a Hamiltonian, and let $V : \CC \rightarrow \Lin(\mathcal{H})$ be an analytic family of perturbations. Let $\beta > 0$, and let $R > 0$ be such that $V(z)$ is bounded as
    \begin{align}
        \|V(z)\| \leq \beta^{-1} \log(3/2)\quad\text{whenever}\quad |z| \leq R\ .\label{eq:gibbsstatesmallpertbound}
    \end{align}
    Fix an observable $O \in \Herm(\mathcal{H})$. Then the function
    \begin{align}
        f_\beta(z) := \frac{\tr[Oe^{-\beta(H + V(z))}]}{\tr[e^{-\beta(H+V(z))}]}
    \end{align}
    is analytic in the disc $|z| \leq R$, and moreover is bounded by
    \begin{align}
        \sup_{|z|\leq R} |f_\beta(z)| \leq 3\|O\|\ .
    \end{align}
\end{theorem}

\cref{lem:traceineq} is the key technical fact for this purpose; it is quite a general trace inequality which may be of independent interest:

\begin{lemma}\label{lem:traceineq}
    Let $X\in \Herm(\mathcal{H})$ be a Hermitian matrix, and let $Y \in \Lin(\mathcal{H})$ be a general matrix. Let $\omega(X,Y) = |\tr[e^{X+Y}]|/|\tr[e^X]|$. Then
    \begin{align}\label{eq:traceineq}
        2 - e^{\|Y\|} \leq \omega(X,Y)\leq e^{\|Y\|}\ .
    \end{align}
\end{lemma}

We think of $X = -\beta H$ proportional to an unperturbed Hamiltonian, and $Y = -\beta V(z)$ being a family of perturbations depending on a complex parameter $z$ --- thus $\omega(X,Y)$ has the interpretation as the ratio of the corresponding perturbed partition function to an unperturbed partition function. The fact that $\omega(X,Y) > 0$ for some region of $z$ is sufficient to deduce that $\tr[e^{-\beta(H + V(z))}]$ has no complex zeroes in this region, but is not sufficient to bound the Taylor series truncation error of the corresponding Gibbs state. For this, we need an explicit positive lower bound for $\omega(X,Y)$, ensuring that the perturbed partition function cannot get ``too small'' and hence that expectation values of the perturbed ground state are bounded for the region of complex $z$. We note that the lower bound in \eqref{eq:traceineq} is not tight (and is only non-trivial when $\|Y\|\leq \log 2$) but it is sufficient for our purposes.

\begin{proof}[*lem:traceineq]
    The upper bound is immediate from \cref{lem:expsum1normbound}, since
    \begin{align}
        \omega(X,Y) = \frac{|\tr[e^{X+Y}]|}{\tr[e^X]} \leq \frac{\|e^{X+Y}\|_1}{\tr[e^X]} \leq e^{\|Y\|}\ .
    \end{align}
    For the lower bound, we introduce a function $f(t)$ defined by
    \begin{align}
        f(t) = \real\bigg(\frac{\tr[e^{X+Yt}]}{\tr[e^X]} \bigg)\ ,\quad t\in [0,1] \ .
    \end{align}
    Note that $f(0) = 1$ and $f(1) \leq \omega(X,Y)$. Differentiating, we have
    \begin{align}
        \frac{\diff}{\diff t} f(t) &= \frac{\real \tr[Ye^{X+Yt}]}{\tr[e^X]} \notag \\
        &\geq -\frac{\|Y\| \|e^{X+Yt}\|_1}{\tr[e^X]} \notag \\
        &\geq -\|Y\| e^{\|Y\|t}\ ,
    \end{align}
    where in the second line we have used H\"older's inequality and in the third line we have applied \cref{lem:expsum1normbound}. Hence
    \begin{align}
        f(1) = f(0) - \int_0^1 \diff t \frac{\diff}{\diff t} f(t) \geq 1 - \|Y\|\int_0^1\diff t e^{\|Y\|t} = 2 - e^{\|Y\|}\ .
    \end{align}
\end{proof}

We now show how \cref{lem:traceineq} implies \cref{thm:gibbs small pert}. We will need the following two lemmas on the singular values of matrices.

\begin{lemma}[\cite{so2000singular}, Corollary 2.3]\label{lem:exponentialsingularvalues}\noproofref
    Let $X \in \Lin(\mathcal{H})$ be an arbitrary complex matrix. Let $\{s_j(X)\}_j$ be the singular values of $X$, and let $\{s_j(e^X)\}_j$ be the singular values of $e^X$. Then
    \begin{align}
        \sum_j s_j(e^X) \leq \sum_j e^{s_j(X)}\ .
    \end{align}
\end{lemma}

\begin{lemma}[Weyl's inequality for singular values]\label{lem:weylsingularvalues}\noproofref
    Let $X$ and $Y$ be arbitrary matrices, with singular values $0\leq s_0(X) \leq s_1(X) \leq \dots$ and $0\leq s_0(Y) \leq s_1(Y) \leq \dots$ respectively. Then
    \begin{align}
        \max_j |s_j(X) - s_j(Y) | \leq \|X - Y\|\ .
    \end{align}
\end{lemma}

From here we can prove the following useful bound on the 1-norm of a matrix exponential.

\begin{lemma}\label{lem:expsum1normbound}
    Let $X \in \Herm(\mathcal{H})$ be Hermitian, and let $Y \in \Lin(\mathcal{H})$ be an arbitrary matrix. Then
    \begin{align}
        \|e^{X+Y}\|_1 \leq e^{\|Y\|} \tr[e^X]\ .
    \end{align}
\end{lemma}
\begin{proof}[*lem:expsum1normbound]
    Using the notation of \cref{lem:exponentialsingularvalues}, we have
    \begin{align}
        \|e^{X+Y}\|_1 &= \sum_j s_j(e^{X+Y}) \notag \\
        &\leq \sum_j e^{s_j(X+Y)} \notag \\
        &\leq e^{\|Y\|} \sum_j e^{s_j(X)} \notag \\
        &= e^{\|Y\|} \tr[e^X]\ ,
    \end{align}
    where in the second line we used \cref{lem:exponentialsingularvalues} and in the third line we used \cref{lem:weylsingularvalues}.
\end{proof}

\begin{proof}[*thm:gibbs small pert]
    The bound \eqref{eq:gibbsstatesmallpertbound}, along with \cref{lem:traceineq}, guarantees that $\tr[e^{-\beta(H + V(z))}]$ has no zeroes in the region $|z| \leq R$, and so (using the analyticity of $V$) it is immediate that $f_\beta(z)$ is analytic in this disc. 

    For the upper bound, we can compute
    \begin{align}
        |f_\beta(z)| = \bigg| \frac{\tr[e^{-\beta H}]}{\tr[e^{-\beta (H + V(z))}]}\bigg| \cdot \bigg| \frac{\tr[O e^{-\beta(H + V(z))}]}{\tr[e^{-\beta H}]} \bigg| \leq (2-e^{\|\beta V(z)\|})^{-1} \|O\|\frac{\|e^{-\beta (H + V(z))}\|_1}{\tr[e^{-\beta H}]}\ ,
    \end{align}
    where we bounded the first term using \cref{lem:traceineq} and the second term using H\"older's inequality. Using \cref{lem:expsum1normbound}, we can conclude
    \begin{align}
        |f_\beta(z)| \leq (2 - e^{\|\beta V(z)\|})^{-1} \|O\| e^{\|\beta V(z)\|} \leq 3 \|O\|\ ,
    \end{align}
    where in the second inequality we used that $e^x / (2-e^x) \leq 3$ for $0\leq x \leq \log(3/2)$.
\end{proof}

\subsubsection*{Analytic approximation of observables on non-critical Gibbs states}

We are now ready to prove the Gibbs state part of \cref{thm:extrapolation without gadgets}. We restate the result as the following theorem:

\begin{theorem}[Extrapolation of local properties of Gibbs states --- restatement of Theorem~\ref{thm:extrapolation without gadgets}\eqref{it:gibbs extrapolation without gadgets}]\label{thm:gibbs state extrapolation}
    Let $\{H(x)\}_{x \in \cX}$ be a family of Hamiltonians satisfying Assumption~\ref{assumption:non criticality}\eqref{it:gibbs assumption}. Then, for any $\delta > 0$, the function $f_\beta(x)$ defined in Eq.~\eqref{eq:fbeta definition without gadgets} has a $(\delta,M,R)$-analytic approximation, where
    \begin{align}
        M = 3\|O_A\| \ ,\quad x_{\ast}\geq R = 1/\bigO(\log^D(\delta^{-1}))\ .
    \end{align}
\end{theorem}

\begin{proof}[*thm:gibbs state extrapolation]
    Let $r_0 > 0$ (which will later be set to $r_0 \sim \log \delta^{-1}$), and write $f_\beta(x,y;B_{r_0}(A)^c)$ for the function
    \begin{align}
        f_{\beta}(x,y;B_{r_0}(A)^c) := \tr[O_A \rho_\beta(x,y;B_{r_0}(A)^c) ]\ ,
    \end{align}
    where $\rho_\beta(x,y;B_{r_0}(A)^c)$ is defined, as in Assumption~\ref{assumption:non criticality}\eqref{it:gibbs assumption}, as the Gibbs state corresponding to the Hamiltonian
    \begin{align}
        H(x,y;B_{r_0}(A)^c) := \sum_{A'\subseteq B_{r_0}(A)^c} H_{A'}(x) + \sum_{\substack{A'\subseteq \Gamma \\ A'\nsubseteq B_{r_0}(A)^c}} H_{A'}(y)\ .
    \end{align}
    That is, the Hamiltonian where all interactions whose support intersects $B_{r_0}(A)$ are parametrised by $y$, whilst all others are parametrised by $x$. By the GALI property for Gibbs states (Lemma~\ref{lem:GALI gibbs}), $f_{\beta}(x,y;B_{r_0}(A)^c)$ only depends very slightly on its first argument, in the sense that we can bound:
    \begin{align}
        \left| f_\beta(x,x;B_{r_0}(A)^c) - f_\beta(0,x;B_{r_0}(A)^c)\right| &\leq \beta |x| \|O_A\| |A|^3 \sum_{r = r_0}^\infty \left( r^{3D} e^{-r/2\xi} + r^D e^{-b_1 a r/2}\right) \\
        &\leq |x| a' e^{-b' r_0}\ ,
    \end{align}
    for some constants $b_1,c_1 > 0$, and where on the second line we have observed that the summand is exponentially decaying in $r_0$ and absorbed all the constants in the expression into the constants $a',b' > 0$.

    We write $\tilde{f}_\beta(x) := f_\beta(0,x;B_{r_0}(A)^c)$. For any complex $|x|\leq x_{\ast}$, the Hamiltonian $H(0,x;B_{r_0}(A)^c)$ can be viewed as a perturbation from $H(0)$ of magnitude
    \begin{align}
        \left\|H(0,x;B_{r_0}(A)^c) - H(0)\right\| &\leq \sum_{\substack{A'\subseteq \Gamma \\ A' \cap B_{r_0}(A) \neq \emptyset}} \left\| H_{A'}(x) - H_{A'}(0)\right\| \\
        &\leq \sum_{i \in B_{r_0}(A)} \sum_{\substack{A'\subseteq \Gamma \\ i\in A'}} \left\| H_{A'}(x) - H_{A'}(0)\right\|  \\
        &\leq \sum_{i \in B_{r_0}(A)} \sum_{\substack{A'\subseteq \Gamma \\ i\in A'}} |x| \sup_{|x'|\leq R} \left\|\partial_{x'} H_{A'}(x') \right\| \\
        &\leq |B_{r_0}(A)| |x| \sup_{|x'|\leq R} \|\partial_{x'}H(x')\|_{F_g} \|F_g\|\ .
    \end{align}
    In the first and second lines we have used the triangle inequality, in the third line we used the mean value theorem, and in the fourth line we have used the definition of the $\|\cdot\|_{F_g}$-norm. Now, since $|B_{r_0}(A)| \leq |A| k_D r_0^D$, and treating $|A|,k_D,D,\|F_g\|$, and $\sup_{|s'|\leq R}\|\partial_{s'}H(s')\|_{F_g}$ as constants, we have
    \begin{align}
        \left\|H(0,x;B_{r_0}(A)^c) - H(0)\right\| \leq c' |x| r_0^D\ ,
    \end{align}
    for some constant $c' > 0$. Thus, using Theorem~\ref{thm:gibbs small pert} with $V(x) := H(0,x;B_{r_0}(A)^c) - H(0)$, we deduce that $\tilde{f}_\beta(x)$ is analytic and has $\tilde{f}_\beta(x) \leq 3\|O_A\|$ in the complex disc defined by
    \begin{align}
        c' |x| r_0^D \leq \beta^{-1} \log(3/2) \Rightarrow |x| \leq R = \frac{\log(3/2)}{\beta b' r_0^D} = \bigO(r_0^{-D})\ .
    \end{align}
    In other words, $\tilde{f}_\beta(x)$ is a $(\delta,M,R)$-analytic approximation for $f_\beta(x)$, where
    \begin{align}
        \delta = \exp(-\Theta(r_0))\ ,\quad M = 3\|O_A\| \ ,\quad R = \Theta(r_0^{-D})\ .
    \end{align}
    Hence, for any chosen $\delta > 0$ we may fix $r_0 = \bigO(\log(\delta^{-1}))$ to obtain the required result.
\end{proof}

\subsection{Gapped ground states}\label{sec:extrapolation ground}

\subsubsection*{Approximating the spectral flow}

The proof idea for the gapped ground state part of Theorem~\ref{thm:extrapolation without gadgets} is similar to the Gibbs state case: we use GALI to assume that the Hamiltonian perturbations only occur locally (and therefore have bounded norm), and then prove that the effect of such small perturbations on observables can be analytically extended on the complex plane. The second part of this argument quickly runs into a fundamental problem: the spectral flow operator, defined through Lemma~\ref{lem:spectral flow alternative} as
\begin{align}\label{eq:spectral flow restatement}
    \Psi_H(X) = \int_{-\infty}^\infty \tilde{w}_\gamma(t) e^{itH} X e^{-itH}\ ,
\end{align}
is ill-defined whenever $H$ is not Hermitian, as the integral will not necessarily converge. As a result, even very small complex Hamiltonian perturbations may lead to singular behaviour. Below, we prove a few preliminary results to circumvent this issue. Namely, in Lemma~\ref{lem:spectral flow integral truncation}, we show that the integral in Eq.~\eqref{eq:spectral flow restatement} can be truncated to finite range whilst incurring some mild errors. The following results, Lemmas~\ref{lem:to exp analyticity} and \ref{lem:perturbed pure state evolution}, ensure that this approximate version of the spectral flow operator can be used to construct the desired analytic approximation to the function of interest $f_{\ground}(s)$.

\begin{lemma}[Integral truncation for spectral flow]\label{lem:spectral flow integral truncation}
    Let $H \in \Herm(\cH)$ be a Hamiltonian, and let $V : \RR \rightarrow \Herm(\cH)$ be an family of Hamiltonian perturbations which extends to an analytic function $V : \CC \rightarrow \Lin(\cH)$. For any analytic $X : \CC \rightarrow \Lin(\cH)$, let
    \begin{align}
        \Psi_{H + V(x)} (X(x)) = \int_{-\infty}^\infty \diff t \tilde{w}_\gamma(t) e^{it(H+V(x))} X(x) e^{-it(H+V(x))}\ ,\quad x \in \RR\ ,
    \end{align}
    be the spectral flow operator as in Lemma~\ref{lem:spectral flow alternative}. For any $T \geq 0$, define the truncated spectral flow operator by
    \begin{align}
        \tilde{\Psi}_{H + V(z)}(X(z)) := \int_{-T}^T \diff t \tilde{w}_\gamma(t) e^{it(H+V(z))} X(z) e^{-it(H+V(z))}\ ,\quad z \in \CC\ .
    \end{align}
    Then $\tilde{\Psi}_{H+V(z)}(X(z))$ is an analytic function of $z$, bounded by
    \begin{align}
        \|\tilde{\Psi}_{H+V(z)}(X(z)) \| &\leq 2T \|X(z)\| \exp(2T\|V(z)\|)\ ,\quad z\in \CC\ ,
    \end{align}
    and approximates $\Psi_{H+V(x)}(X(x))$ for real $x$ up to error
    \begin{align}
        \|\tilde{\Psi}_{H+V(x)}(X(x)) - \Psi_{H + V(x)}(X(x))\| &\leq 2\|X(x)\| W_2 \gamma^{-1} \left(\frac{\gamma T}{\log^2(\gamma T)} \right)^3 \exp\left(-\frac{2\gamma T}{7\log^2 (\gamma T)} \right) \ ,\quad x\in \RR\ ,
    \end{align}
    where $W_2>0$ is the constant from Lemma~\ref{lem:spectral flow alternative}.
\end{lemma}

\begin{proof}[*lem:spectral flow integral truncation]
    The analyticity of $\tilde{\Psi}_{H+V(z)}(X(z))$ follows from the analyticity of the exponential function, and the standard fact that a weighted integral of analytic functions is itself analytic (see, e.g., Ref.~\cite{stein2010complex} Theorem 5.4). Moreover, the bound follows from the triangle inequality
    \begin{align}
        \left\|\tilde{\Psi}_{H+V(z)}(X(z)) \right\| &\leq \int_{-T}^T \diff t \left\| \tilde{w}_\gamma(t) e^{it(H+V(z))} X(z) e^{-it(H+V(z))}\right\| \\
        &\leq 2T \|X(z)\| \exp\left(2T\|V(z)\| \right)\ ,
    \end{align}
    where in the second line we used that $|\tilde{w}_\gamma(t)| \leq 1$ (from Lemma~\ref{lem:spectral flow alternative}), and the fact that $H \in \Herm(\cH)$. Finally, we can use the triangle inequality again, along with the fact that $\tilde{w}_\gamma$ is an odd function, to bound
    \begin{align}
        \|\tilde{\Psi}_{H+V(x)}(X(x)) - \Psi_{H + V(x)}(X(x))\| &\leq2\|X(x)\| \int_T^\infty \diff t |\tilde{w}_\gamma(t)| \\
        &\leq 2\|X(x)\| W_2 \gamma^{-1} \left(\frac{\gamma T}{\log^2(\gamma T)} \right)^3 \exp\left(-\frac{2\gamma T}{7\log^2 (\gamma T)} \right)\ ,
    \end{align}
    as required, where in the second line we used the upper bound from Lemma~\ref{lem:spectral flow alternative}.
\end{proof}

The next two results concern time-ordered exponentials and pure state evolution. Given a time-dependent Hamiltonian $H : \RR \rightarrow \Herm(\cH)$, we define the time-ordered exponential by
\begin{align}\label{eq:time ordered exp}
    U(x) = \texp\left(i\int_0^x H(t)\diff t \right) := \sum_{q\geq 0} i^q \int_{0\leq t_1\leq \dots \leq t_q \leq x} \diff t_1 \dots \diff t_q H(t_q)\dots H(t_1)\ .
\end{align}
Then, given a family of states $\ket{\psi(x)}$ obeying
\begin{align}
    \frac{\diff}{\diff x} \ket{\psi(x)} = iH(x) \ket{\psi(x)}\ ,\quad x\in \RR\ ,
\end{align}
the unique solution for $\ket{\psi(x)}$ given a starting state $\ket{\psi(0)}$ is
\begin{align}
    \ket{\psi(x)} = U(x) \ket{\psi(0)}\ .
\end{align}
The next result establishes that $U(x)$ (and hence $\ket{\psi(x)}$) extends to an analytic function, provided that $H(x)$ does.

\begin{lemma}[Analyticity of time-ordered exponentials]\label{lem:to exp analyticity}
    Let $H : \CC \rightarrow \Lin(\cH)$ be an analytic matrix-valued function in the disc $D_R := \{z\in \CC : |z| \leq R\}$, which is bounded as $\sup_{|z|\leq R}\|H(z)\| = M$. Then $U(s)$ (as defined in Eq.~\eqref{eq:time ordered exp}) extends to a function $U : \CC \rightarrow \Lin(\cH)$ which is analytic for $|z|\leq R$, and which satisfies
    \begin{align}
        \sup_{|z| \leq R} \|U(z)\| \leq e^{MR}\ .
    \end{align}
\end{lemma}

\begin{proof}[*lem:to exp analyticity]
    We aim to define, for $z\in D_R$, 
    \begin{align}
        U(z) := \sum_{q\geq 0} i^q \int_0^z \diff t_1 \int_{t_1}^z \diff t_2 \dots \int_{t_{q-1}}^z \diff t_q H(t_q) \dots H(t_1)\ ,
    \end{align}
    though we must check that this is well-defined, in the sense that the integrals are independent of the path taken through $\CC$ (or at least $D_R$). To this end, we prove by induction that each term of the form
    \begin{align}
        \int_0^z \diff t_1 \int_{t_1}^z \diff t_2 \dots \int_{t_{q-1}}^z \diff t_q H(t_q) \dots H(t_1)
    \end{align}
    is an analytic function of $z$, and independent of the path that $t_1$ takes. Indeed, this is trivially true for $q=1$ by the analyticity of $H$. Moreover, we can write this as
    \begin{align}
        \int_0^z \diff t_1 \int_{t_1}^z \diff t_2 \dots \int_{t_{q-1}}^z \diff t_q H(t_q) \dots H(t_1) &= \int_0^z \diff t_1 \bigg( \int_0^z \diff t_2 \dots \int_{t_{q-1}}^z \diff t_q H(t_q) \dots H(t_2) \\
        &\quad - \int_0^{t_1} \diff t_2 \dots \int_{t_{q-1}}^z \diff t_q H(t_q)\dots H(t_2)\bigg) H(t_1)\ .
    \end{align}
    By induction, the bracketed term is an analytic function of both $z$ and $t_1$, and --- using the analyticity of $H$ --- so is the entire integrand. Since $U(z)$ is a sum of such terms, whose magnitude decays quickly (as $\sim \|zH\|^q/q!$), this proves that $U(z)$ is both well-defined and analytic (see Ref.~\cite{stein2010complex}, Theorem 5.2).

    Using the independence of paths in the integral, we can without loss of generality choose the straight-line paths
    \begin{align}
        U(|z|e^{i\theta}) = \sum_{q\geq 0} i^q e^{iq\theta} \int_0^{|z|} \diff s_1 \int_{s_1}^{|z|} \diff s_2 \dots \int_{s_{q-1}}^{|z|} \diff s_q H(s_q e^{i\theta}) \dots H(s_1 e^{i\theta})\ .
    \end{align}
    Thus, using that the total volume of the integral through $\RR^q$ is $|z|^q/q!$, we can bound
    \begin{align}
        \|U(z)\| = \sum_{q\geq 0} \frac{|z|^q}{q!} \max_{w \leq |z|} \|H(w)\|^q \leq e^{M|z|}\ ,
    \end{align}
    from which the result follows.
\end{proof}

The next result bounds the difference between the trajectories of two states evolving under slightly different Hamiltonians (for our purposes, this will be relevant to bound the evolution error between an ideal spectral flow operator, and its truncated approximation).

\begin{lemma}[Perturbed pure state evolution]\label{lem:perturbed pure state evolution}
    Let $H,V : \RR \rightarrow \Herm(\cH)$ be two continuous families of Hamiltonians. Fix an initial state $\ket{\psi_0}\in \cH$, and define the trajectories $\ket{\psi(x)}$ and $\ket{\tilde{\psi}(x)}$ by
    \begin{align}
        \ket{\psi(0)} &= \ket{\tilde{\psi}(0)} = \ket{\psi_0}\ , \\
        \frac{\diff}{\diff x} \ket{\psi(x)} &= iH(x) \ket{\psi(x)}\ , \\
        \frac{\diff}{\diff x} \ket{\tilde{\psi}(x)} &= i(H(x) + V(x)) \ket{\tilde{\psi}(x)}\ .
    \end{align}
    Then, for all $x\in \RR$,
    \begin{align}
        \left\|\ket{\psi(x)} - \ket{\tilde{\psi}(x)} \right\| \leq \int_0^x \diff t\|V(t)\| \ .
    \end{align}
\end{lemma}

\begin{proof}[*lem:perturbed pure state evolution]
    We define the vector $\ket{\epsilon(x)} := \ket{\psi(x)} - \ket{\tilde{\psi}(x)}$. We can compute
    \begin{align}
        \frac{\diff}{\diff x} \ket{\epsilon(x)} &= iH(x)\ket{\psi(x)} - i(H(x) + V(x)) \ket{\tilde{\psi}(x)} \\
        &= i(H(x) + V(x)) \ket{\epsilon(x)} -iV(x) \ket{\psi(x)}\ .
    \end{align}
    One may verify that the unique solution of this differential equation with $\ket{\epsilon(0)} = 0$ is given by
    \begin{align}
        \ket{\epsilon(x)} = -i\int_0^x \diff t \texp\left( i\int_t^s \diff \tau (H(\tau) + V(\tau) )\right) V(t) \ket{\psi(t)}\ .
    \end{align}
    Hence, using the triangle inequality, we have for all $x \geq 0$ that
    \begin{align}
        \|\ket{\epsilon(x)}\| &\leq \int_0^x \diff t \left\| \texp\left( i\int_t^x \diff \tau (H(\tau) + V(\tau) )\right) V(t) \ket{\psi(t)}\right\| \\
        &\leq \int_0^x \diff t \|V(t)\|\ .
    \end{align}
\end{proof}

\subsubsection*{Analytic approximation of observables on gapped ground states}

We now state the below theorem, which establishes an analogous result to Theorem~\ref{thm:gibbs small pert} in the case of gapped ground states: for weakly perturbed Hamiltonians, ground state properties depend approximately analytically on the perturbed parameter. 

\begin{theorem}\label{thm:ground state small pert}
    Let $H \in \Herm(\cH)$ be a Hamiltonian and let $V : \CC \rightarrow \Lin(\cH)$ be an analytic family of perturbations such that $H + V(x)$ has a unique ground state $\ket{\psi_0(x)}$ with spectral gap $\gamma > 0$ for all $x \in [-R,R]$, for some $R > 0$. Fix an observable $O \in \Herm(\cH)$. Then for every $T > 0$, the function
    \begin{align}
        f_{\ground}(x) := \bra{\psi_0(x)} O \ket{\psi_0(x)}
    \end{align}
    has a $(\delta,M,R)$-analytic approximation where
    \begin{align}
        \delta = 2\|O\|c''\lambda R \exp\left(-a''\frac{T}{\log^2(T)}\right) \ ,\quad M = \|O\| \exp\left( 2RT \lambda e^{2RT\lambda} \right)\ ,
    \end{align}
    where $a'',c''>0$ are constants and $\lambda := \sup_{|z|\leq R} \|\partial_z V(z)\|$.
\end{theorem}

\begin{proof}[*thm:ground state small pert]
    We can use the spectral flow as prescribed in Lemmas~\ref{lem:spectral flow} and \ref{lem:spectral flow alternative} to write
    \begin{align}
        \frac{\diff}{\diff x} \ket{\psi_0(x)} = i\Psi_{H+V(x)} (\partial_x V(x)) \ket{\psi_0(x)}\ .
    \end{align}
    We define $\tilde{\Psi}_{H+V(x)}(\partial_xV(x))$ to be the spectral flow operator with the integral truncated to the range $[-T,T]$, as in Lemma~\ref{lem:spectral flow integral truncation}. Then $\tilde{\Psi}_{H+V(x)}(\partial_x V(x))$ extends to an analytic function for $|z|\leq R$, bounded by
    \begin{align}
        \left\|\tilde{\Psi}_{H+V(z)}(\partial_z V(z)) \right\| &\leq 2T \|\partial_z V(z)\| \exp\left(2T\|V(z)\|\right) \leq 2T\lambda e^{2T R\lambda}\ ,
    \end{align}
    where we defined $\lambda := \sup_{|z| \leq R} \|\partial_z V(z)\|$. Furthermore, for $x\in [-R,R]$ we have
    \begin{align}
        \left\|\tilde{\Psi}_{H+V(x)}(\partial_x V(x)) - \Psi_{H+V(x)}(\partial_x V(x)) \right\| \leq c''\lambda \exp\left(-a'' \frac{T}{\log^2(T)} \right)\ ,
    \end{align}
    where we have absorbed the constants from Lemma~\ref{lem:spectral flow integral truncation} into the constants $c'',a''>0$.

    Defining $\ket{\tilde{\psi}_0(z)}$ for $|z| \leq R$ by
    \begin{align}
        \ket{\tilde{\psi}_0(0)} = \ket{\psi_0(0)}\ ,\quad \frac{\diff}{\diff z} \ket{\tilde{\psi}_0(z)} = i\tilde{\Psi}_{H+V(z)}(\partial_z V(z))\ ,
    \end{align}
    we have, by Lemma~\ref{lem:perturbed pure state evolution}, that
    \begin{align}
        \left\| \ket{\psi_0(x)} - \ket{\tilde{\psi}_0(x)}\right\| \leq c'' \lambda R \exp\left(-a'' \frac{T}{\log^2 (T)} \right)\quad\text{for}\quad x\in [-R,R]\ .
    \end{align}
    Moreover, by Lemma~\ref{lem:to exp analyticity}, we are guaranteed that $\ket{\tilde{\psi}_0(x)}$ extends to an analytic function for $|z|\leq R$ with
    \begin{align}
        \sup_{|z|\leq R} \|\ket{\tilde{\psi}_0(z)}\| &\leq \exp\left(2R T\lambda e^{2R T\lambda}\right) \ .
    \end{align}
    Defining the analytic function $\tilde{f}_{\ground}(z) := \bra{\tilde{\psi}_0(z)} O \ket{\tilde{\psi}_0(z)}$, we have
    \begin{align}
        \left| f_{\ground}(x) - \tilde{f}_{\ground}(x)\right| &\leq 2\|O\| \cdot \left\| \ket{\psi_0(x)} - \ket{\tilde{\psi}_0(x)}\right\|\\
        &\leq 2\|O\| c'' \lambda R\exp\left(-a''\frac{T}{\log^2(T)}\right)\quad\text{for}\quad x\in[-R,R]\ ,
    \end{align}
    and
    \begin{align}
        |f_{\ground}(z)| &\leq \|O\| \|\ket{\tilde{\psi}_0(z)}\|^2 \\
        &\leq \|O\| \exp \left(4RT \lambda e^{2R T\lambda} \right)\ \quad\text{for}\quad |z|\leq R\ .
    \end{align}
    Hence $\tilde{f}_{\ground}$ is a $(\delta,M,R)$-analytic approximation for $f_{\ground}$, where
    \begin{align}
        \delta = 2\|O\| c'' \lambda R \exp\left(-a'' \frac{T}{\log^2(T)} \right)\ ,\quad M = \|O\| \exp\left( 2RT\lambda e^{2R T\lambda} \right)\ .
    \end{align}
\end{proof}

We restate the ground state part of Theorem~\ref{thm:extrapolation without gadgets} below:
\begin{theorem}[Restatement of Theorem~\ref{thm:extrapolation without gadgets}\eqref{it:gs extrapolation without gadgets}]\label{thm:ground state extrapolation}
    Let $\{H(x)\}_{x\in \cX}$ be a family of Hamiltonians satisfying Assumption~\ref{assumption:non criticality}\eqref{it:ground assumption}. Then, for any $\delta > 0$, the function $f_{\ground}(x)$ defined in Eq.~\eqref{eq:fground definition without gadgets} has a $(\delta,M,R)$-analytic approximation, where
    \begin{align}
        M = 16\|O_A\| \ ,\quad x_{\ast} \geq R = 1/\bigO(\log^{D+1}(\delta^{-1})) \ .
    \end{align}
\end{theorem}
\begin{proof}[*thm:ground state extrapolation]
    Let $r_0,T>0$ be positive parameters which will later be fixed. Similarly to the proof of Theorem~\ref{thm:gibbs state extrapolation}, we write $f_{\ground}(x,y;B_{r_0})$ for the function
    \begin{align}
        f_{\ground}(x,y;B_{r_0}(A)^c) := \bra{\psi_0(x,y;B_{r_0}(A)^c)} O_A \ket{\psi_0(x,y;B_{r_0}(A)^c)} \ ,
    \end{align}
    where $\ket{\psi_0(x,y;B_{r_0}(A)^c)}$ is defined, as in Assumption~\ref{assumption:non criticality}\eqref{it:ground assumption}, as the unique ground state state corresponding to the Hamiltonian
    \begin{align}
        H(x,y;B_{r_0}(A)^c) := \sum_{A'\subseteq B_{r_0}(A)^c} H_{A'}(x) + \sum_{\substack{A' \subseteq \Gamma \\ A'\nsubseteq B_{r_0}(A)^c}}H_{A'}(y)\ .
    \end{align}
    By the GALI property for Ground states (Lemma~\ref{lem:GALI ground}), $f_{\ground}(x,y;B_{r_0}(A)^c)$ only depends very slightly on its first argument, in the sense that we can bound
    \begin{align}
        \left| f_{\ground}(x,x;B_{r_0}(A)^c) - f_{\ground}(0,x;B_{r_0}(A)^c)\right| &\leq |x| c_2 \|O_A\| |A|^2 \sum_{r = r_0}^\infty r^D e^{-b_r a(r-1) / \log^2(a(r-1))} \\
        &\leq c'|x| e^{-a' r_0 / \log^2(r_0)}\ ,\label{eq:f ground truncation bound}
    \end{align}
    for some constants $b_2,c_2 > 0$, and where on the second line we have absorbed all the constants in the expression into some $a',c' > 0$. As in the proof of Theorem~\ref{thm:gibbs state extrapolation}, we can bound
    \begin{align}
        \left\|\partial_x H(0,x;B_{r_0}(A)^c) \right\| &\leq \sum_{\substack{A' \subseteq \Gamma \\ A'\nsubseteq B_{r_0}(A)^c}} \|\partial_x H_{A'}(x)\| \\
        &\leq \sum_{i \in B_{r_0}(A)} \sum_{\substack{A'\subseteq\Gamma \\ i \in A'}} \|\partial_x H_{A'}(x)\| \\
        &\leq |B_{r_0}(A)| \|\partial_x H(x)\|_{F_g} \|F_g\| \\
        &\leq b' r_0^D\ ,
    \end{align}
    for some constant $b' > 0$. Thus, by Theorem~\ref{thm:ground state small pert}, there exists a $(\delta',M,x_{\ast})$ approximation $\tilde{f}_{\ground}(x)$ for $f_{\ground}(0,x;B_{r_0}(A)^c)$, where
    \begin{align}
        \delta' = 2\|O_A\| c'' b' r_0^D R \exp\left(-a'' \frac{T}{\log^2 T} \right)\ ,\quad M = \|O_A\| \exp\left(2x_{\ast} T b' r_0^D e^{2RT b' r_0^D} \right)\ .
    \end{align}
    Hence, by Eq.~\eqref{eq:f ground truncation bound} $\tilde{f}_{\ground}(x)$ is a $(\delta,M,x_{\ast})$-approximation for $f_{\ground}(x)$, where
    \begin{align}
        \delta = c'x_{\ast} e^{-a'r_0 / \log^2(r_0)} + 2\|O_A\| c'' b' r_0^D x_{\ast} \exp\left(-a'' \frac{T}{\log^2 T} \right)\ ,\quad M = \|O_A\| \exp\left(2x_{\ast}T b' r_0^D e^{2RT b' r_0^D} \right)\ .
    \end{align}
    Note that we have the freedom to tune the parameters $r_0,T>0$, as well as the ability to restrict to a smaller $R\leq x_{\ast}$. We thus set $T = r_0$, and choose $R = \min\{x_{\ast},1/(2Tb' r_0^D)\}$. This leads to
    \begin{align}
        \delta = r_0^{-1} \exp\left(-\Omega(r_0/\log^2(r_0) \right)\ ,\quad M \leq \|O_A\| e^e \leq 16 \|O_A\|\ .
    \end{align}
    Hence any arbitrarily small $\delta > 0$ can be achieved by setting $r_0 \sim \log(\delta^{-1})$, giving
    \begin{align}
        R = 1/\bigO( r_0^{D+1}) = 1/\bigO(\log^{D+1}(\delta^{-1}))\ ,
    \end{align}
    as required.
\end{proof}

\section{Perturbation theory and simulation gadgets}\label{sec:perturbation theory}

\subsection{Effective Hamiltonians from local gaps}

In the analogue quantum simulation literature (for example, Refs.~\cite{oliveira2005complexity,bravyi2008quantum,cubitt2016complexity,cubitt2018universal}), perturbative ``gadgets'' are used to simulate complicated local Hamiltonian terms using simpler building blocks. These work by introducing ancillary qubits to mediate interactions, with large on-site energy penalty terms forcing the ancillary qubits into the correct subspace described by some projector $\id \otimes P_0$ (where $P_0$ acts only on the ancillary qubits). Provided that the local energy penalties are high enough (i.e., scaling linearly with the size of the system), this defines a low-energy subspace of the global Hamiltonian in which an effective Hamiltonian acting on the non-ancillary qubits can be derived. 

In this section we establish the perturbation theory tools required to deal with gadget Hamiltonians in the weak-interaction regime, where the strength of the local penalty terms is constant and hence a global energy gap is not present, but where we can still meaningfully define a perturbative effective Hamiltonian induced by the gadgets. We will make use of the local Schrieffer-Wolff transformation, for which the basic idea is as follows. Given a Hamiltonian $H'$ consisting of many gadgets with relatively weak interactions, we cannot guarantee that the global low-energy subspace will coincide with the ``all gadget ancillas in ground state'' subspace, which is described by some projector $\id\otimes P_0$ acting on the ancillary qubits, even though the gadgets may be locally gapped. However, the subspace described by $\id\otimes P_0$ can be slightly rotated to give a subspace which \emph{is} invariant under $H'$ --- though this does not necessarily correspond to a subspace spanned only by its lowest-energy eigenvectors. The local Schrieffer-Wolff transformation provides a generator $T$, which $T/i$ is a local Hamiltonian, which generates this rotation from the invariant subspace into the $\id\otimes P_0$ space, so that we can define an effective Hamiltonian $H_{\eff} \otimes P_0 = (\id \otimes P_0) e^T H' e^{-T} (\id \otimes P_0)$. This effective Hamiltonian corresponds to the space in which the gadgets are behaving as designed, and in which the usual effective Hamiltonian is recovered in the limit of large interaction strengths.

The construction of the local Schrieffer-Wolff transformation, originally introduced by Ref.~\cite{datta1996low}, is explained in Section~\ref{sec:localsw}, and various basic properties are proven (with several proofs deferred to Appendix~\ref{app:sw bounds}). Our analysis closely follows the excellent review Ref.~\cite{bravyi2011schrieffer}, but differs in two respects. Firstly, we will derive bounds on the effective Hamiltonian in terms of the $\|\cdot\|_F$-norm for a normalised $F$ function $F$ --- this will be useful when we later analyse dynamics under such Hamiltonians, and in fact turns out to lead to slightly simpler analysis than the norm used in Ref.~\cite{bravyi2011schrieffer}, in particular leading to a convergent power series which does not need to be truncated at finite order. Secondly, we generalise to the case where the simulator Hamiltonian is a polynomial in the perturbative parameter, allowing interaction strengths over several energy scales --- this is the relevant situation for perturbative gadgets. We define such gadgets formally in Section~\ref{sec:gadgets} and prove several useful properties in Appendix~\ref{app:gadget properties}.

\subsection{Local Schrieffer-Wolff perturbation theory}\label{sec:localsw}

Consider a Hilbert space $\cH = \bigotimes_{i\in \Gamma} \cH_i$ consisting of $|\Gamma| = n$ sites. We will study Hamiltonians of the form
\begin{align}\label{eq:sw ham definitions start}
    H(x) = \sum_{\alpha = 0}^d x^\alpha H^{(\alpha)} \ ,
\end{align}
parametrised by some real number $x\geq 0$.\footnote{Ultimately we will be interested in the object $x^{-d}H(x)$, containing strong interactions which become singular as $x\rightarrow 0$.} The constant term $H^{(0)}$ is assumed to be $1$-local with local ground state projectors (which may be the identity) $P_{0,i}$ at every $i\in \Gamma$, such that when a locally excited space exists, it is separated by a gap of at least $\Delta > 0$ above the ground space which is assumed to have zero energy. The higher-order terms are assumed to have interactions bounded by
\begin{align}
    \|H^{(\alpha)}\|_F \leq J\ ,\quad 1\leq \alpha \leq d\ ,
\end{align}
for some normalised $F$-function $F$, and constant $J>0$. Our goal is to construct a power series
\begin{align}
    T(x) = \sum_{q=1}^\infty x^q T^{(q)}\ ,
\end{align}
and some series $\{V^{(q)}\}_{q\geq 1}$ of Hamiltonians such that
\begin{align}\label{eq:local sw basic}
    e^{T(x)} H(x) e^{-T(x)} = H^{(0)} + \sum_{q=1}^\infty x^q \left(V^{(q)} + [T^{(q)},H^{(0)}] \right) \ ,
\end{align}
and where, for $q\geq 1$, the term $V^{(q)} + [T^{(q)},H^{(0)}]$ is local and block-diagonal with respect to the projector $P_0 := \bigotimes_{i\in \Gamma} P_{0,i}$. In other words, we wish to locally diagonalise $H(x)$ with respect to this $P_0$ (for $x$ sufficiently small that the series converges). Following Ref.~\cite{bravyi2011schrieffer}, the strategy will be to inductively construct $\{V^{(q)}\}_{q\geq 1}$ and $\{T^{(q)}\}_{q\geq 1}$ as follows:
\begin{itemize}
    \item The first term $V^{(1)}$ is set to $V^{(1)} = H^{(1)}$.
    \item For $q\geq 1$, $T^{(q)}$ is defined in terms of $V^{(q)}$.
    \item For $q\geq 2$, $V^{(q)}$ is defined in terms of the lower-order $T^{(q-1)},T^{(q-2)},\dots$.
\end{itemize}

Note that Eq.~\eqref{eq:local sw basic} fully determines $V^{(q)}$ in terms of the lower-order $T^{(q)}$ by comparing powers of $x^q$: we can expand the left-hand side in powers of $x$ to find (defining $H^{(\alpha)} = 0$ for $\alpha > d$):
\begin{align}
    e^{T(x)} H(x) e^{-T(x)} &= \sum_{\alpha = 0}^\infty x^\alpha e^{T(x)} H^{(\alpha)} e^{-T(x)} \\
    &= \sum_{\alpha = 0}^\infty x^\alpha \left(H^{(\alpha)} + \sum_{r= 1}^\infty \frac{1}{r!}\sum_{q_1,\dots,q_r = 1}^\infty x^{q_1+\dots+q_r} [T^{(q_1)},[T^{(q_2)},\dots [T^{(q_r)},H^{(\alpha)}]\cdots]] \right) \\
    &= H^{(0)} + \sum_{q= 1}^\infty x^q \bigg( [T^{(q)},H^{(0)}] + H^{(q)}  + \notag \\
    &\quad + \sum_{r= 2}^q\frac{1}{r!} \sum_{\substack{q -1 \geq q_1,\dots,q_r \geq 1 \\ q_1+\dots+q_r = q}} [T^{(q_1)},[T^{(q_2)},\dots [T^{(q_r)},H^{(0)}] \cdots]] \notag \\
    &\quad + \sum_{\alpha =1}^q \sum_{r= 1}^{q-\alpha}\frac{1}{r!} \sum_{\substack{q-\alpha \geq q_1,\dots,q_r \geq 1 \\ q_1+\dots+q_r = q-\alpha}}[T^{(q_1)},[T^{(q_2)},\dots[T^{(q_r)},H^{(\alpha)}] \cdots]]  \bigg) \ ,
\end{align}
where in the third line we have separated the sum over $r$ into the cases $\alpha = 0$ and $\alpha \geq 1$, and further separated the former term into the $r=1$ case and the $r\geq 2$ case. Hence
\begin{align}
    V^{(1)} &= H^{(1)}\ , \\
    V^{(q)} &= H^{(q)} + \sum_{r= 2}^q \frac{1}{r!}\sum_{\substack{q-1\geq q_1,\dots,q_r \geq 1 \\ q_1+\dots+q_r = q}} [T^{(q_1)},[T^{(q_2)},\dots[T^{(q_r)},H^{(0)}] \cdots]] \notag \\
    &\quad + \sum_{\alpha =1}^q \sum_{r= 1}^{q-\alpha}\frac{1}{r!} \sum_{\substack{q-\alpha \geq q_1,\dots,q_r \geq 1 \\ q_1+\dots+q_r = q-\alpha}}[T^{(q_1)},[T^{(q_2)},\dots[T^{(q_r)},H^{(\alpha)}] \cdots]]\ , \quad q\geq 2\ .\label{eq:vq definition}
\end{align}

Note that the expression for $V^{(q)}$ only depends on $T^{(q')}$ for $q' < q$ (and ultimately $T^{(q')}$ will depend only on $V^{(q')}$), so this inductive definition is well-defined. Note here that we define $H^{(q)} = 0$ for $q > d$. We only have the freedom to choose $T^{(q)}$ to satisfy the desired block-diagonal condition, which we do in the following.

Given $A \subseteq \Gamma$, we define the projectors
\begin{align}\label{eq:pa definition}
    P_A := \bigotimes_{i\in A} P_{0,i}\ ,\quad Q_A := 1-P_A\ .
\end{align}
Note that here, and below, we abuse notation and suppress tensor product factors of the identity on the non-ancillary qubits (that is, identifying $P_A$ with $P_A \otimes \id_{\Gamma \setminus A}$). For any operator $X$, we define its off-diagonal and diagonal parts with respect to this pair of projectors by
\begin{align}
    \offdiag_A(X) &:= P_A X Q_A + Q_A X P_A\ , \\
    \ondiag_A(X) &:= P_A X P_A + Q_A X Q_A\ .
\end{align}
We now define the superoperator $\cL_A$ by
\begin{align}\label{eq:la definition}
    \cL_A(X) := \left( \int_0^\infty \diff t e^{-tH^{(0)}|_A / \Delta} \right) Q_A X P_A -P_A X Q_A\left(  \int_0^\infty \diff t  e^{-tH^{(0)}|_A/\Delta}\right)\ ,
\end{align}
where $H^{(0)}|_A$ consists of the 1-local terms in $H^{(0)}$ acting on the sites $i\in A$, i.e.
\begin{align}
    H^{(0)}|_A := \sum_{i\in A} H^{(0)}_i\ ,
\end{align}
and $\Delta >0$ is the size of the local energy gaps in $H^{(0)}$. This satisfies the following properties:
\begin{lemma}[Properties of $\cL_A$ --- see \cite{bravyi2011schrieffer} Section 4.4]\label{lem:properties of la}
    Let the superoperator $\cL_A$ be defined as in Eq.~\eqref{eq:la definition}.
    \begin{enumerate}[(I)]
        \item \label{it:la bound}For any $X$, we have
        \begin{align}
            \|\cL_A(X)\| \leq \|X\|\ .
        \end{align}
        \item \label{it:la locality}If $X$ acts only on the spins $A \subseteq \Gamma$, then so does $\cL_A(X)$.
        \item \label{it:la offdiag} If $X$ acts only on the spins $A$, then
        \begin{align}
            \cL_A\left( [H^{(0)},X]\right) = \left[ H^{(0)} , \cL_A(X)\right] = \Delta \offdiag_A(X)\ .
        \end{align}
    \end{enumerate}
\end{lemma}
\begin{proof}[*lem:properties of la]
    Point~\eqref{it:la bound} follows from the definition of $Q_A$, which projects onto the excited space of $H^{(0)}|_A$ (with energy at least $\Delta$) and thus guarantees that $\|e^{-tH^{(0)}|_A/\Delta}Q_A\| \leq e^{-t}$.
    
    Point~\eqref{it:la locality} follows immediately from the definition, as $X$, $Q_A$, and $P_A$ are all supported only on $A$.
    
    For point~\eqref{it:la offdiag}, note that we have
    \begin{align}
        Q_A [H^{(0)}, X] P_A &= Q_A [H^{(0)}|_A , X] P_A \\
        &= H^{(0)}|_A Q_A X P_A\ ,
    \end{align}
    where the first equality follows as $X$ is only supported on $A$, and the second equality follows as $H^{(0)}|_A P_A = 0$, and $H^{(0)}|_A$ commutes with $Q_A$. Thus 
    \begin{align}
        \left(\int_0^\infty \diff t e^{-tH^{(0)}|_A / \Delta } \right) Q_A [H^{(0)},X] P_A &= \left( \int_0^\infty \diff t e^{-tH^{(0)}|_A/\Delta} \right) H^{(0)}|_A Q_A X P_A \\
        &= \Delta Q_A X P_A\ .
    \end{align}
    Hence, via an identical calculation for the second term,
    \begin{align}
        \cL_A\left([H^{(0)},X] \right) &= Q_A X_A P_A + P_A X_A Q_A\ ,
    \end{align}
    as required. The fact that this is equivalent to $[H^{(0)}|_A,\cL_A(X)] = [H^{(0)},\cL_A(X)]$ follows from a similar computation, by commuting $H^{(0)}|_A$ to the outside of both terms.
\end{proof}

Writing $V^{(q)}$ as a sum of local terms
\begin{align}
    V^{(q)} = \sum_{A\subseteq \Gamma} V^{(q)}_A\ ,
\end{align}
we define $T^{(q)}$ by
\begin{align}\label{eq:tq definition}
    T^{(q)} = \sum_{A\subseteq \Gamma} \Delta^{-1} \cL_A(V_A^{(q)})\ .
\end{align}
Observe that by Lemma~\ref{lem:properties of la}, $T^{(q)}$ inherits the locality properties from $V^{(q)}$, and
\begin{align}
    [T^{(q)},H^{(0)}] + V^{(q)} &= \sum_{A\subseteq \Gamma} \left( \Delta^{-1} [\cL_A(V_A^{(q)}) , H^{(0)}]  + V^{(q)}_A\right)\\
    &= \sum_{A\subseteq \Gamma} \left(-\offdiag_A(V_A^{(q)}) + V_A^{(q)} \right) \\
    &= \sum_{A\subseteq \Gamma} \ondiag_A(V_A^{(q)})\ ,
\end{align}
where for the second equality we used Lemma~\ref{lem:properties of la}\eqref{it:la offdiag}. Hence the right-hand side of Eq.~\eqref{eq:local sw basic} is indeed a sum of local and block-diagonal interactions. We will generally denote $V(x) = \sum_{q\geq 1} x^q V^{(q)}$, and write $V_{\eff}(x)$ for the effective Hamiltonian in the $\id\otimes P_0$ space, that is
\begin{align}\label{eq:sw veff definition}
    V_{\eff}(x) \otimes P_0 := \sum_{q\geq 1} x^q (\id \otimes P_0) V^{(q)} (\id \otimes P_0)\ .
\end{align}

\subsubsection*{Bounding the perturbative expansion}

Having obtained a recursive definition for the perturbative series $V^{(q)}$ and $T^{(q)}$ in Eq.~\eqref{eq:local sw basic}, we can prove bounds on the growth of these operators with $q$, ultimately establishing that the effective Hamiltonian is itself bounded in the $\|\cdot \|_F$-norm. This main result is stated below, as Lemma~\ref{lem:vq bound}, and is proved in Appendix~\ref{app:sw bounds}.

\begin{lemma}[Compare to Ref.~\cite{bravyi2011schrieffer}, Lemma 4.2 --- see Lemma~\ref{lem:vq bound app}]\label{lem:vq bound}\noproofref
    For $q\geq 1$, $V^{(q)}$ is bounded as
    \begin{align}
        \|V^{(q)}\|_F \leq \frac{\Delta \theta^q}{16}\ ,
    \end{align}
    where $\theta > 0$ is a constant depending on the ratio $J/\Delta$. In particular, if $x \leq 1/(2\theta)$, then the Hamiltonian $V(x) = \sum_{q\geq 1} x^q V^{(q)}$ is bounded as
    \begin{align}
        \|V_{\eff}(x)\|_F \leq \|V(x)\|_F \leq \frac{\Delta \theta x}{8}\ .
    \end{align}
\end{lemma}

From Lemma~\ref{lem:vq bound}, we can establish the following bound on the generator $T(x)$.

\begin{corollary}\label{cor:tq bound}
    For $q\geq 1$, $T^{(q)}$ is bounded as
    \begin{align}
        \|T^{(q)}\|_F \leq \frac{\theta^q}{16}\ ,
    \end{align}
    where $\theta > 0$ is the constant from Lemma~\ref{lem:vq bound}. Moreover, if $x \leq 1/(2\theta)$, then the generator of the Schrieffer-Wolff transformation $T(x) = \sum_{q\geq 1} x^q T^{(q)}$ is bounded as
    \begin{align}
        \|T(x)\|_F \leq \frac{\theta x}{8}\ .
    \end{align}
\end{corollary}

\begin{proof}[*cor:tq bound]
    Follows immediately from Eq.~\eqref{eq:tq definition} and Lemma~\ref{lem:properties of la}\eqref{it:la bound}, using the bound on $\|V^{(q)}\|_F$ from Lemma~\ref{lem:vq bound}.
\end{proof}

\subsubsection*{Gibbs states and ground states of the perturbed Hamiltonian}

As mentioned, the subspace $\id \otimes P_0$ will not necessarily contain only the lowest-energy states of the rotated Hamiltonian $e^{T(x)} H(x) e^{-T(x)}$, unless $x$ is very small --- typically $x\sim n^{-1}$, where $n = |\Gamma|$, is necessary to produce the necessary gap in the spectrum. In this section we establish that, even for constant $x$, low-temperature Gibbs states have high overlap within the $\id \otimes P_0$ subspace and are thus well-approximated by the Gibbs states of the effective Hamiltonian. The approximation improves as the inverse-temperature increases, and becomes exact for ground states (see Corollary~\ref{cor:sw ground state}).

The following result, Lemma~\ref{lem:gibbs state restriction}, is the key tool for this purpose. This lemma shows that a Hamiltonian of the form $H^{(0)} + V$, where $H^{(0)}$ consists of gapped 1-local terms and $V$ is geometrically local (and which is locally block-diagonal with respect to the terms in $H^{(0)}$), has Gibbs states which are likely to be in the ground space of $H^{(0)}$. Intuitively, this is because any excitation of $l$ terms in $H^{(0)}$ will be suppressed as $e^{-\beta l}$ in the Gibbs state of inverse temperature $\beta$ compared to its non-excited counterpart. Via some combinatorics to account for the terms at different levels of excitation, we can thus show that the non-excited space is strongly favoured in the Gibbs state.

\begin{lemma}\label{lem:gibbs state restriction}
    Let $H^{(0)}$ and $V$ be Hamiltonians on $\cH = \bigotimes_{i\in \Gamma} \cH_i$, where $\cH_i \cong \CC^m$, with the following properties:
    \begin{enumerate}[(a)]
        \item There is a subset $\Gamma_{\anc}\subseteq \Gamma$ of ``ancillary sites'' such that $H^{(0)}$ takes the form
        \begin{align}\label{eq:simple h0 term}
            H^{(0)} = \Delta\sum_{i\in \Gamma_{\anc}} \left(\id - P_{0,i} \right)\ ,
        \end{align}
        where, for every $i\in \Gamma_{\anc}$, $P_{0,i}$ is a 1-local projector acting on site $i$, of rank $1$. We write $P_0 = \bigotimes_{i\in \Gamma_{\anc}} P_{0,i}$ for the projector onto the ground space of $H^{(0)}$.
        \item $V = \sum_{A\subseteq \Gamma} V_A$ is local on $\Gamma$, with bounded $\|\cdot\|_F$ norm for some $F$-function $F$. Moreover, $V_A$ is block-diagonal with respect to $P_{0,A} := \bigotimes_{i\in A\cap \Gamma_{\anc}} P_{0,i}$ for every $A\subseteq \Gamma$.
    \end{enumerate}
    For any $\beta \geq 0$, we write $\sigma_\beta$ for the Gibbs state of the Hamiltonian $H^{(0)} + V$, i.e.,
    \begin{align}
        \sigma_\beta := \frac{e^{-\beta(H^{(0)} + V)}}{\tr[e^{-\beta(H^{(0)} + V)}]} \ .
    \end{align}
    Then
    \begin{align}
        \tr[P_0 \sigma_\beta] \geq \exp\left( -(m-1) |\Gamma_{\anc}| e^{-\beta (\Delta - 3\|V\|_F \|F\|)}\right)\ .
    \end{align}
\end{lemma}
This result only gives a non-trivial bound in the very low-temperature regime $\beta \gtrsim \log|\Gamma_{\anc}| \sim \log(n)$, however in our analysis of gadget Hamiltonians we will ultimately absorb a factor of $x^{-d}$ into $\beta$ before applying this theorem to show that the Gibbs state approximately lies in the effective subspace $P_0$ for $x\sim 1/\log n$.

\begin{proof}[*lem:gibbs state restriction]
    Since $[V,P_0] = 0$, it follows that $V$ is block-diagonal with respect to the projectors $P_0, 1-P_0$. We write $\{\ket{\psi_a}\}_a$ for the eigenvectors of $P_0V P_0$, i.e. an orthornormal set spanning the image of $P_0$, such that $P_0 \ket{\psi_a} = \ket{\psi_a}$, $V \ket{\psi_a} = E_a \ket{\psi_a}$ for some energies $E_a \in \RR$.

    Since each $P_{0,i}$ is a rank 1 projector, we can without loss of generality choose a basis $\{\ket{b}\}_{b=0}^{m-1}$ of each site $\cH_i$, such that $P_{0,i} = \proj{0}_i$. We define the 1-local unitary $X$ to be such that $X\ket{b} = \ket{b+1\mod m}$. For $\mathbf{b} \in \{0,\dots,m\}^{\Gamma_{\anc}}$, we write $X^{\mathbf{b}}$ for the tensor product
    \begin{align}
        X^{\mathbf{b}} := \bigotimes_{i\in \Gamma_{\anc}} X_i^{b_i}\ ,
    \end{align}
    which shifts the basis element at each ancilla site $i$ by $b_i$. Since the $\{\ket{\psi_a}\}_a$ span the entire space in which all the ancillas are in the $\ket{0}$ state, it follows that $\{\ket{\psi_{a,\mathbf{b}}} := X^{\mathbf{b}} \ket{\psi_a}\}_{a,\mathbf{b}}$ is a basis for the entire space $\cH$. Hence we can write
    \begin{align}
        \tr\left[ e^{-\beta(H^{(0)} + V)}\right] &= \sum_{a,\mathbf{b}} \bra{\psi_a}(X^{\mathbf{b}})^\dagger e^{-\beta(H^{(0)} + V)} X^{\mathbf{b}} \ket{\psi_a} \\
        &= \sum_{a,\mathbf{b}} \bra{\psi_a} e^{-\beta (X^{\mathbf{b}})^\dagger (H^{(0)} + V) X^{\mathbf{b}}} \ket{\psi_a}\ .
    \end{align}
    We now decompose the exponentiated Hamiltonian into its blocks with respect to the projectors $P_0$ and $(1-P_0)$:
    \begin{align}
        (X^{\mathbf{b}})^{\dagger} (H^{(0)} + V) X^{\mathbf{b}} &= P_0 (X^{\mathbf{b}})^{\dagger} ( H^{(0)} + V) X^{\mathbf{b}} P_0 \\
        &\quad + (1-P_0) (X^{\mathbf{b}})^{\dagger} (H^{(0)} + V) X^{\mathbf{b}} (1-P_0) \\
        &\quad P_0 (X^{\mathbf{b}})^{\dagger} (H^{(0)} + V) X^{\mathbf{b}} (1-P_0) + \text{h.c.}\ .
    \end{align}
    Notice that, since $X^{\mathbf{b}}$ maps a zero-energy state $\ket{0}_i$ to an excited state $\ket{b_i}_i$ for every $b_i\neq 0$, we have
    \begin{align}
        P_0 (X^{\mathbf{b}})^\dagger H^{(0)} X^{\mathbf{b}} P_0 = \Delta |\mathbf{b}| P_0\ ,
    \end{align}
    where $|\mathbf{b}|$ denotes the Hamming weight of $\mathbf{b}$. Moreover, we can bound
    \begin{align}
        \left\|P_0 (X^{\mathbf{b}})^\dagger V X^{\mathbf{b}} P_0 - P_0 V P_0 \right\| &\leq \sum_{A\subseteq \Gamma} \left\|P_0 (X^{\mathbf{b}})^\dagger V_A X^{\mathbf{b}} P_0 - P_0 V_A P_0 \right\| \\
        &\leq \sum_{\substack{i \in \Gamma_{\anc} \\ b_i \neq 0}} \sum_{\substack{A\subseteq \Gamma \\ i \in A}} \left\|P_0 (X^{\mathbf{b}})^\dagger V_A X^{\mathbf{b}} P_0 - P_0 V_A P_0 \right\| \\
        &\leq \sum_{\substack{i \in \Gamma_{\anc} \\ b_i \neq 0}} \sum_{\substack{A\subseteq \Gamma \\ i \in A}}  2\|V_A\| \\
        &\leq 2|\mathbf{b}| \|F\| \|V\|_F\ ,
    \end{align}
    where in the first line we used the triangle inequality, in the second line we used the fact that $X^{\mathbf{b}}$ and $V_A$ commute if $b_i = 0$ for all $i\in A$. Since $(X^{\mathbf{b}})^\dagger H^{(0)} X^{\mathbf{b}}$ is block-diagonal with respect to $P_0$, we have
    \begin{align}
        P_0 (X^{\mathbf{b}})^\dagger H^{(0)} X^{\mathbf{b}} (1-P_0) = 0\ .
    \end{align}
    Finally, using that $V$ is locally block-diagonal with respect to the $P_{0,A}$, and again that $X^{\mathbf{b}}$ is only supported on $i$ with $b_i \neq 0$, we have
    \begin{align}
        P_0(X^{\mathbf{b}})^\dagger V X^{\mathbf{b}} (1-P_0) &\leq \sum_{A\subseteq \Gamma} \left\|P_0(X^{\mathbf{b}})^\dagger V_A X^{\mathbf{b}} (1-P_0)  \right\| \\
        &\leq \sum_{\substack{i\in \Gamma \\ b_i \neq 0}} \sum_{\substack{A\subseteq \Gamma \\ i\in A}} \|V_A\| \\
        &\leq |\mathbf{b}| \|F\| \|V\|_F\ .
    \end{align}
    Putting these together, we can deduce that
    \begin{align}
        (X^{\mathbf{b}})^{\dagger} ( H^{(0)} + V) X^{\mathbf{b}} &\geq \Delta|\mathbf{b}| P_0 + P_0 V P_0 - 3|\mathbf{b}| \|F\| \|V\|_F P_0  \\
        &\quad + (1-P_0) (X^{\mathbf{b}})^{\dagger} ( H^{(0)} + V) X^{\mathbf{b}} (1-P_0)\ ,
    \end{align}
    and hence
    \begin{align}
        \tr\left[ e^{-\beta( H^{(0)} + V)} \right] &\leq \sum_{a,\mathbf{b}} \bra{\psi_a} e^{-\beta(\Delta |\mathbf{b}| P_0 + P_0 V P_0 - 3|\mathbf{b}| \|F\| \|V\|_F P_0 )} \ket{\psi_a} \\
        &=\sum_{a,\mathbf{b}} e^{-\beta(\Delta|\mathbf{b}| + E_a - 3|\mathbf{b}| \|F\| \|V\|_F)}\ ,
    \end{align}
    where we ignored the $(1-P_0)$-block, since the $\ket{\psi_a}$ are in the span of $P_0$. Notice that
    \begin{align}
        \tr\left[ P_0 e^{-\beta ( H^{(0)} + V)} \right] &= \sum_a e^{-\beta E_a}\ ,
    \end{align}
    and hence we can write
    \begin{align}
        \tr\left[e^{-\beta( H^{(0)} + V)} \right] &\leq \tr\left[ P_0 e^{-\beta ( H^{(0)} + V)} \right] \sum_{\mathbf{b} \in \{0,\dots,m-1\}^{\Gamma_{\anc}}} e^{-\beta |\mathbf{b}| (\Delta - 3\|F\| \|V\|_F)} \ ,
    \end{align}
    which can be rearranged to give
    \begin{align}
        \frac{1}{\tr\left[P_0 \sigma_\beta\right]} &\leq \sum_{l = 0}^{|\Gamma_{\anc}|} \sum_{\substack{\mathbf{b}\in \{0,\dots,m-1\}^{\Gamma_{\anc}}\\ |\mathbf{b}| = l}} e^{-\beta l(\Delta - 3\|F\| \|V\|_F)} \\
        &= \sum_{l=0}^{|\Gamma_{\anc}|}\binom{|\Gamma_{\anc}|}{l} (m-1)^l e^{-\beta l(\Delta - 3\|F\| \|V\|_F)} \\
        &= \left(1 + (m-1)e^{-\beta (\Delta - 3\|F\|\|V\|_F)} \right)^{|\Gamma_{\anc}|} \\
        &\leq \exp\left( (m-1) |\Gamma_{\anc}| e^{-\beta(\Delta - 3\|F\|\|V\|_F)}\right)\ ,
    \end{align}
    from which the result follows.
\end{proof}

We now specialise this result to the setting of local Schrieffer-Wolff perturbation theory, showing that low-temperature Gibbs states of the perturbed Hamiltonian approximately lie within the subspace of the effective Hamiltonian.

\begin{corollary}\label{cor:sw gibbs state}
    Let $H(x)$ be a Hamiltonian as defined in Eqs.~\eqref{eq:sw ham definitions start}-\eqref{eq:local sw basic}, where the constant term $H^{(0)}$ takes the form given in Eq.~\eqref{eq:simple h0 term}. Let $T(x)$ and $V(x)$ denote the local Schrieffer-Wolff transformation and associated perturbative series. Let $\rho_\beta(x)$ be the Gibbs state associated to $H(x)$, and let $\rho_{\beta,\eff}(x)$ be the Gibbs state associated to $V_{\eff}(x)$ on the subspace $\id \otimes P_0$ as defined in Eq.~\eqref{eq:sw veff definition}, that is
    \begin{align}
        \rho_\beta(x) := \frac{e^{-\beta H(x)}}{\tr[e^{-\beta H(x)}]}\ ,\quad \rho_{\beta,\eff}(x) := \frac{e^{-\beta V_{\eff}(x)}}{\tr[e^{-\beta V_{\eff}(x)}]}\ .
    \end{align}
    Then, provided $x\leq 1/(2\theta)$ (where $\theta$ is the constant from Lemma~\ref{lem:vq bound}),
    \begin{align}
        \left\| e^{T(x)} \rho_\beta(x) e^{-T(x)} - \rho_{\beta,\eff}(x) \otimes P_0\right\|_1 \leq 2 - 2\exp\left(-(m-1) |\Gamma_{\anc}| e^{-\beta\Delta(1 - 3\theta\|F\| x/8)} \right) \ .
    \end{align}
\end{corollary}

\begin{proof}[*cor:sw gibbs state]
    By the definition of the local Schrieffer-Wolff transformation, we have
    \begin{align}
        e^{T(x)} H(x) e^{-T(x)} = H^{(0)} + V(x)\ ,
    \end{align}
    where $H^{(0)}$ and $V(x)$ satisfy the conditions of Lemma~\ref{lem:gibbs state restriction}, and by Lemma~\ref{lem:vq bound} we have
    \begin{align}
        \|V(x)\|_F \leq \frac{\Delta \theta x}{8}\ .
    \end{align}
    Hence, by Lemma~\ref{lem:gibbs state restriction}, we have
    \begin{align}
        \tr\left[(\id \otimes P_0) e^{T(x)} \rho_\beta(x) e^{-T(x)} \right] \geq \exp\left(-(m-1) |\Gamma_{\anc}| e^{-\beta \Delta(1 - 3 \theta \|F\| x / 8)} \right) =: p_0\ .
    \end{align}
    By definition,
    \begin{align}
        \left\|e^{T(x)} \rho_\beta(x) e^{-T(x)} - \rho_{\beta,\eff}(x) \otimes P_0 \right\|_1 &= \left\|e^{T(x)} \rho_\beta(x) e^{-T(x)} - p_0^{-1} (\id \otimes P_0 )e^{T(x)} \rho_\beta(x) e^{-T(x)}  (\id \otimes P_0) \right\|_1 \\
        &= \Big\|(\id -\id \otimes P_0)e^{T(x)} \rho_\beta(x) e^{-T(x)} (\id -\id \otimes P_0)\notag \\
        &\quad + (1-p_0^{-1}) (\id \otimes P_0 )e^{T(x)} \rho_\beta(x) e^{-T(x)}(\id \otimes P_0) \Big\|_1 \\
        &\leq (1-p_0) + |1-p_0^{-1}|p_0 = 2(1-p_0)\ ,
    \end{align}
    where the final line follows from the triangle inequality.
\end{proof}

Taking the limit $\beta\rightarrow\infty$ in Corollary~\ref{cor:sw gibbs state}, we can obtain an analogous result for ground states. In particular this implies that, when the parameter $x$ is below a constant, all ground states of the perturbed Hamiltonian $H(x)$ lie in the subspace $\id\otimes P_0$.

\begin{corollary}[Compare with Ref.~\cite{bravyi2011schrieffer}, Lemma 4.1]\label{cor:sw ground state}
    Let $H(x)$ be a Hamiltonian as defined in Eqs.~\eqref{eq:sw ham definitions start}-\eqref{eq:local sw basic}, where the constant term $H^{(0)}$ takes the form given in Eq.~\eqref{eq:simple h0 term}. Assume that 
    \begin{align}\label{eq:xast definition}
        x < x_{\ast} := \frac{1}{\theta(2+\|F\|)}\ ,
    \end{align}
    where $\theta$ is the constant from Lemma~\ref{lem:vq bound}. Then the ground space of $e^{T(x)} H(x) e^{-T(x)}$ lies in the image of the projector $\id \otimes P_0$.
\end{corollary}

\begin{proof}[*cor:sw ground state]
    This follows by taking $\beta \rightarrow \infty$ in Corollary~\ref{cor:sw gibbs state}. In particular, note that
    \begin{align}
        \lim_{\beta \rightarrow \infty} e^{T(x)} \rho_\beta(x) e^{-T(x)} &= \frac{P(x)}{\tr[P(x)]} \\
        \lim_{\beta \rightarrow \infty} \rho_{\beta,\eff}(x) \otimes P_0 &= \frac{(\id \otimes P_0) P(x) (\id \otimes P_0) }{\tr[(\id \otimes P_0) P(x)]}\ ,
    \end{align}
    where $P(x)$ is the ground state projector of $e^{T(x)} H(x) e^{-T(x)}$, whilst Corollary~\ref{cor:sw gibbs state} implies that
    \begin{align}
        \lim_{\beta \rightarrow \infty} \left\| e^{T(x)} \rho_\beta(x) e^{-T(x)} - \rho_{\beta,\eff}(x) \otimes P_0 \right\| = 0\ .
    \end{align}
    Hence we can conclude that $(\id \otimes  P_0) P(x) (\id \otimes P_0) = P(x)$, as required.
\end{proof}

\subsection{Hamiltonian gadgets}\label{sec:gadgets}

In this section, we specialise the local Schrieffer-Wolff transformation constructed in Section~\ref{sec:localsw} to the case of perturbative gadgets. Generally, these are designed such that the target Hamiltonian $H$ only appears at $d$th order in the perturbative expansion for $H_{\eff}(x)$ --- and hence, the entire Hamiltonian must be rescaled by $x^{-d}$ to make this term constant. This leads to a singular Hamiltonian with interactions polynomial in $x^{-1}$, but such that $H_{\eff}(x)$ is polynomial in $x$. We state this explicitly below.

\begin{definition}[Gadget of degree $d$]\label{def:gadget}
    Let $\Gamma = \Gamma_{\anc}\cup \Gamma_{\eff}$ be a disjoint partition of the sites $\Gamma$, $|\Gamma| = n$.
    For $d\geq 1$, let $H'(x)$ be a family of Hamiltonians of the form
    \begin{align}
        H'(x) = x^{-d} \sum_{\alpha = 0}^d x^\alpha H^{(\alpha)}\ ,
    \end{align}
    where $H^{(0)} = \Delta \sum_{i\in \Gamma_{\anc}} (\id - \proj{0_i})$, where $\ket{0_i}$ is a state on $\cH_i$, and for $1\leq \alpha \leq d$ we have $\|H^{(\alpha)}\|_F \leq J$ for some normalised $F$-function $F$. Define $T(x)$ to be the local Schrieffer-Wolff transformation as constructed in Section~\ref{sec:localsw}, so that $e^{T(x)} H'(x) e^{-T(x)}$ is block-diagonal with respect to the projector $P_0 := \otimes_{i \in \Gamma_{\anc}} \proj{0_i}$. Define the Hamiltonian $H_{\eff}(x)$ on $\cH_{\eff} := \otimes_{i\in \Gamma_{\eff}}\cH_i$ by
    \begin{align}
        H_{\eff}(x)\otimes P_0:= P_0 e^{T(x)} H'(x) e^{-T(x)} P_0\ .
    \end{align}
    If $H_{\eff}(x)$ is analytic in $x$, then we say that $H'(x)$ is a gadget of degree $d$ for the target Hamiltonian $H_{\tar} := H_{\eff}(0)$.
\end{definition}

Note that the condition that $H_{\eff}(x)$ is analytic is equivalent to requiring that, for $\{V^{(q)}\}_{q\geq 1}$ defined via Eq.~\eqref{eq:vq definition}, we have $P_0 V^{(q)} P_0 = 0$ for $1\leq q \leq d-1$. In Appendix~\ref{app:gadget properties}, we prove several general results about gadgets as defined in Definition~\ref{def:gadget}. In particular, we show that the definition is compatible with previous notions of perturbative simulation \cite{bravyi2017complexity,cubitt2018universal,harley2024going} (Lemma~\ref{lem:equivalent gadgets}), we show that gadgets can be combined in parallel provided that $d\leq 3$ (Theorem~\ref{thm:gadget comb}), and we give a strict bound on their locality (Lemma~\ref{lem:gadget loc}).

A remark on notation and terminology: typically, when $H'(x)$ is an extensive Hamiltonian (i.e. when $\Gamma_{\anc} = \bigO(n)$) simulating many interactions, it is referred to as a \emph{simulator} rather than a gadget. Usually such $H'(x)$ are built by the parallel application of many gadgets each simulating local terms (as is made precise in Theorem~\ref{thm:gadget comb}), which are denoted with lowercase $h_i'(x)$. Definition~\ref{def:gadget} (and some of our subsequent results) slightly misuse this convention, and we use uppercase $H$ to emphasise that the statements are also valid for extensive simulator Hamiltonians.

\subsubsection*{Properties of simulator Hamiltonians}

For convenience, we collect our main results about the local Schrieffer-Wolff transformation, adapted to simulator Hamiltonians, into a Theorem~\ref{thm:simulation properties}. We separately state the setup required (essentially, that $H'(x)$ is a simulator Hamiltonian for some $H_{\tar}$) as Assumption~\ref{assumption:gadget ham}, as these conditions will be re-used in Section~\ref{sec:simulator extrapolation}.

\begin{assumption}\label{assumption:gadget ham}
    Let $H'(x)$ be a family of Hamiltonians on the system $\cH' = \otimes_{i \in \Gamma_{\eff} \cup \Gamma_{\anc}} \cH_i$, of the form
    \begin{align}
        H'(x) = x^{-d} \sum_{\alpha = 0}^d x^\alpha H^{(\alpha)}\ ,
    \end{align}
    where $H^{(0)} = \Delta \sum_{i\in \Gamma_{\anc}} (\id - \proj{0_i})$ for some states $\ket{0_i}\in \cH_i$, and where for $1\leq \alpha \leq d$ we have $\|H^{(\alpha)}\|_F \leq J$ for some normalised $F$-function $F$ on $\Gamma = \Gamma_{\eff} \cup \Gamma_{\anc}$. Let $P_0 := \otimes_{i\in \Gamma_{\anc}} \proj{0_i}$, and assume that $H'(x)$ is a gadget of degree $d$ for some Hamiltonian $H_{\tar}$ on $\cH_{\eff} := \otimes_{i\in \Gamma_{\eff}} \cH_i$, with associated Schrieffer-Wolff transformation $T(x)$. Assume that $x \leq x_{\ast} := 1/(\theta(2+\|F\|))$, where $\theta = 64(1+4J/\Delta)^2$ perturbative series converges and we can define
    \begin{align}
        H_{\eff}(x) \otimes P_0 = P_0 e^{T(x)} H'(x) e^{-T(x)} P_0\ .
    \end{align}
\end{assumption}

\begin{theorem}[Properties of simulator Hamiltonians]\label{thm:simulation properties}
    Let $H'(x)$ be a family of Hamiltonians satisfying the conditions of Assumption~\ref{assumption:gadget ham}. Then, for $x \leq x_{\ast}$, the following holds:
    \begin{enumerate}[(I)]
        \item\label{it:heff bound} Both $H_{\eff}(x)$ and $T(x)$ are analytic functions of $x$, and we have the locality bounds
        \begin{align}
            \|H_{\eff}(x)\|_F \leq \frac{\Delta \theta^d}{8} \ ,\quad \|T(x)\|_F \leq \frac{\theta x}{8} \ .
        \end{align}
        \item\label{it:gibbs bound} Let $\rho_\beta'(x)$ and $\rho_{\beta,\eff}(x)$ be the Gibbs states of $H'(x)$ and $H_{\eff}(x)$ respectively at inverse temperature $\beta > 0$. Then
        \begin{align}
            \left\| \rho_\beta'(x) - e^{-T(x)} (\rho_{\beta,\eff}(x) \otimes P_0) e^{T(x)} \right\|_1 \leq 2(m-1)|\Gamma_{\anc}| e^{-\beta\Delta x^{-d} / 2}\ ,
        \end{align}
        where $m = \max_{x\in \Gamma_{\anc}} \dim\cH_x$.
        \item\label{it:ground state} Let $P_{\ground}'(x)$ and $P_{\ground,\eff}(x)$ be the projectors onto the ground spaces of $H'(x)$ and $H_{\eff}(x)$ respectively. Then
        \begin{align}
            P_{\ground}'(x) = e^{-T(x)} (P_{\ground,\eff}(x) \otimes P_0) e^{T(x)}\ .
        \end{align}
    \end{enumerate}
\end{theorem}

\begin{proof}[*thm:simulation properties]
    For consistency with Section~\ref{sec:localsw}, we write $H(x)$ for the normalisation of $H'(x)$ by a factor of $x^d$, leading to the non-singular Hamiltonian
    \begin{align}
        H(x) := x^d H'(x) = \sum_{\alpha = 0}^d x^\alpha H^{(\alpha)}\ .
    \end{align}
    We may then construct the series $\{T^{(q)}\}_{q\geq 1}$, $\{V^{(q)}\}_{q\geq 1}$ as in Eqs.~\eqref{eq:vq definition} and \eqref{eq:tq definition}. For $x$ small enough so that the series $T(x) := \sum_{q\geq 1} x^q T^{(q)}$ and $V(x) := \sum_{q\geq 1} x^q V^{(q)}$ both converge, we have that $H(x)$ (and hence $H'(x)$) is block-diagonal with respect to $e^{-T(x)}P_0 e^{T(x)}$, and we write
    \begin{align}
        P_0 e^{T(x)} H'(x) e^{-T(x)} P_0 = x^{-d} \sum_{q\geq 1} x^q P_0 V^{(q)} P_0\ .
    \end{align}
    By the assumption that $H'(x)$ is a gadget Hamiltonian, we are guaranteed that the right-hand side of this expression is analytic in $x$ and hence $P_0V^{(q)} P_0 = 0$ for $q < d$. Therefore we can write
    \begin{align}
        P_0 e^{T(x)} H'(x) e^{-T(x)} P_0 &= \sum_{q\geq 0} x^q P_0 V^{(q+d)} P_0 = H_{\eff}(x) \otimes P_0\ .
    \end{align}
    Expanding into powers of $x$, we write
    \begin{align}
        H_{\eff}(x) = \sum_{q\geq 0} x^q H_{\eff}^{(q)}\ ,
    \end{align}
    where $H_{\eff}^{(q)}\otimes P_0 = P_0 V^{(q+d)} P_0$.
    \begin{enumerate}[(I)]
        \item The analyticity of $H_{\eff}(x)$ and $T(x)$ for $x \leq x_{\ast}$ is immediate from their construction as power series in $x$. In order to bound $\|H_{\eff}(x)\|_F$, we must relate the $\|\cdot\|_F$-norm of $H_{\eff}^{(q)}$ to that of $V^{(q+d)}$. Expanding $V^{(q+d)}$ into a sum of local terms,
        \begin{align}
            V^{(q+d)} = \sum_{A\subseteq \Gamma} V^{(q+d)}_A = \sum_{A_{\eff}\subseteq \Gamma_{\eff}} \sum_{A_{\anc} \subseteq \Gamma_{\anc}} V^{(q+d)}_{A_{\eff} \cup A_{\anc}}\ ,
        \end{align}
        we can write
        \begin{align}
            H_{\eff}^{(q)} = \sum_{A_{\eff} \subseteq \Gamma_{\eff}} H^{(q)}_{\eff,A_{\eff}}\ ,
        \end{align}
        where
        \begin{align}
            H^{(q)}_{\eff,A_{\eff}} := \sum_{A_{\anc} \subseteq \Gamma_{\anc}} (\id \otimes \bra{\mathbf{0}_{A_{\anc}}} ) V^{(q+d)}_{A_{\eff} \cup A_{\anc}} (\id \otimes \ket{\mathbf{0}_{A_{\anc}}})\ ,
        \end{align}
        and where $\ket{\mathbf{0}_{A_{\anc}}} = \otimes_{i\in A_{\anc}} \ket{0_i}$ is the unique state on $\bigotimes_{i\in \Gamma_{\anc}} \cH_i$ in the image of $P_{A_{\anc}} := \otimes_{i\in A_{\anc}} \proj{0_i}$. In particular, by the triangle inequality this implies that
        \begin{align}
            \|H^{(q)}_{\eff,A_{\eff}}\| \leq \sum_{A_{\anc}\subseteq \Gamma_{\anc}} \| V^{(q+d)}_{A_{\eff}\cup A_{\anc}}\|\ .
        \end{align}
        Hence
        \begin{align}
            \|H^{(q)}_{\eff}\|_F &= \sup_{i,j \in \Gamma_{\eff}} \frac{1}{F(\dist(i,j))} \sum_{\substack{A_{\eff}\subseteq \Gamma_{\eff} \\ i,j\in A_{\eff}}} \|H^{(q)}_{\eff,A_{\eff}}\| \\
            &\leq \sup_{i,j\in \Gamma_{\eff}} \frac{1}{F(\dist(i,j))} \sum_{\substack{A_{\eff} \subseteq \Gamma_{\eff} \\ i,j\in A_{\eff}}} \sum_{A_{\anc \subseteq \Gamma_{\anc}}} \|V^{(q+d)}_{A_{\eff} \cup A_{\anc}}\| \\
            &\leq \sup_{i,j\in \Gamma_{\eff\cup \Gamma_{\anc}}} \frac{1}{F(\dist(i,j))} \sum_{\substack{A\subseteq \Gamma_{\eff} \cup \Gamma_{\anc} \\i,j\in A}} \|V^{(q+d)}_A\| \\
            &= \|V^{(q+d)}\|_F\ .
        \end{align}
        By Lemma~\ref{lem:vq bound}, we can bound $\|V^{(q)}\|_F \leq \Delta \theta^q / 16$. Hence, since $x\leq 1/(2\theta) \leq x_{\ast}$, 
        \begin{align*}
            \|H_{\eff}(x)\|_F &\leq \sum_{q\geq 0} x^q \|H_{\eff}^{(q)}\|_F \\
            &
            \leq \sum_{q\geq 0} x^q \|V^{(q+d)}\|_F \\
            &\leq \frac{\Delta \theta^d}{16}\sum_{q\geq 0} (x\theta)^q \\
            &\leq \frac{\Delta \theta^d}{8}\ ,
        \end{align*}
        as required. The bound $\|T(x)\|_F \leq \theta x/8$ is immediate from Corollary~\ref{cor:tq bound}.
        \item This follows by applying Corollary~\ref{cor:sw gibbs state} to the Hamiltonian $H(x)$, and absorbing the factor of $x^{-d}$ into the inverse temperature, yielding
        \begin{align}
            \left\|\rho_\beta'(x) - e^{-T(x)} (\rho_{\beta_{\eff}(x)}\otimes P_0) e^{T(x)} \right\|_1 &\leq 2(m-1)|\Gamma_{\anc}| \exp\left(-\beta \Delta x^{-d} (1 - 3 \theta \|F\| x/8) \right)\ .
        \end{align}
        Since we assume that $x \leq x_{\ast} \leq 1/(\theta \|F\|) \leq 4/(3\theta\|F\|)$, this furhter simplifies to
        \begin{align}
            \left\|\rho_\beta'(x) - e^{-T(x)} (\rho_{\beta_{\eff}(x)}\otimes P_0) e^{T(x)} \right\|_1 \leq 2(m-1)|\Gamma_{\anc}| e^{-\beta\Delta x^{-d} / 2}\ ,
        \end{align}
        as required.
        \item This is an immediate consequence of Corollary~\ref{cor:sw ground state}.
    \end{enumerate}
    
\end{proof}

\section{Extrapolation with simulator Hamiltonians}\label{sec:simulator extrapolation}

\subsection{Main statement}

In this section, we combine the local Schrieffer-Wolff perturbation theory tools from Section~\ref{sec:perturbation theory} with the extrapolation results from Section~\ref{sec:basic extrapolation}. The main result, Theorem~\ref{thm:extrapolation with gadgets} below, establishes analytic approximations for observables on simulator Hamiltonians. In Corollary~\ref{cor:extrapolation with gadgets}, we explicitly rephrase this in terms of the simulation overhead required for Richardson extrapolation.
\begin{theorem}\label{thm:extrapolation with gadgets}
    Let $H'(x)$ be a family of degree-$d$ gadget Hamiltonians on $\cH'=\bigotimes_{i\in \Gamma_{\eff}\cup \Gamma_{\anc}}\cH_i$ satisfying the conditions of Assumption~\ref{assumption:gadget ham} with exponentially decaying interactions (i.e. with $\|\cdot\|_{F_g}$ in place of $\|\cdot\|_F$, for linear $g$). Define the target Hamiltonian $H_{\tar} := H_{\eff}(0)$, and let $O_A$ be an observable supported on $A\subseteq \Gamma_{\eff}$, where $|A|,\|O_A\| = \bigO(1)$. Then the following holds:
    \begin{enumerate}[(I)]
        \item\label{it:gibbs extrapolation with gadgets} Let $\beta > 0$, and let $\rho_\beta'(x)$, $\rho_{\beta,\eff}(x)$ and $\rho_{\beta,\tar}$ be the Gibbs states corresponding to the Hamiltonians $H'(x)$, $H_{\eff}(x)$, and $H_{\tar}$ respectively. Assume that the family $\rho_{\beta,\eff}(x)$ satisfies the exponential correlation decay of Assumption~\ref{assumption:non criticality}\eqref{it:gibbs assumption} for $x \in [0,x_{\ast}]$. Then the function
        \begin{align}
            f'_\beta(x) := \left\{\begin{array}{ll}
                \tr[O_A \rho_\beta'(x)] & \quad\text{for}\quad x \in (0,x_\ast) \\
                \tr[O_A \rho_{\beta,\tar}] & \quad\text{for}\quad x=0
            \end{array}  \right.\ ,
        \end{align}
        has a $(\delta,M,R)$-analytic approximation for any $\delta > 0$, where
        \begin{align}
            M = \exp\left(  \bigO(\log^D(\delta^{-1}))\right)\ ,\quad R = \min\left\{ 1/\bigO(\log^D(\delta^{-1})), 1/\bigO(\log^{1/d}(n\delta^{-1}))\right\}\ .
        \end{align}
        
        \item\label{it:gs extrapolation with gadgets} Assume that $H_{\eff}(x)$ has a spectral gap as in Assumption~\ref{assumption:non criticality}\eqref{it:ground assumption} for $x\in [0,x_{\ast}]$, and let $\ket{\psi_0'(x)}$, $\ket{\psi_{0,\eff}(x)}$ and $\ket{\psi_{0,\tar}}$ be the ground states corresponding to the Hamiltonians $H'(x)$, $H_{\eff}(x)$, and $H_{\tar}$ respectively\footnote{Note that the uniqueness of the ground state for $H'(x)$ follows immediately from our assumption on $H_{\eff}(x)$ and Theorem~\ref{thm:simulation properties}\eqref{it:ground state}.}. Then the function
        \begin{align}
            f'_{\ground}(x) := \left\{\begin{array}{ll}
                \bra{\psi_0'(x)} O_A \ket{\psi_0'(x)} & \quad\text{for}\quad x\in (0,x_\ast) \\
                \bra{\psi_{0,\tar}} O_A \ket{\psi_{0,\tar}} & \quad \text{for}\quad x=0
            \end{array} \right.\ ,
        \end{align}
        has a $(\delta,M,R)$-analytic approximation for any $\delta > 0$, where
        \begin{align}
            M = \exp\left(\bigO(\log^{D}(\delta^{-1})) \right)\ ,\quad R =  1/\bigO(\log^{D+1}(\delta^{-1}))\ ,
        \end{align}
    \end{enumerate}
\end{theorem}

Note that the ground state part is $n$-independent, as in Theorem~\ref{thm:extrapolation without gadgets}, however for Gibbs states we obtain a $1/\bigO(\log^{1/d}(n\Delta^{-1}))$ upper bound on $R$ --- this is necessary to ensure that the Gibbs state lies in the correct effective subspace (which does not require $n$-dependent scaling for ground states due to Corollary~\ref{cor:sw ground state}).

\begin{corollary}[Extrapolating local observables with simulator Hamiltonians]\label{cor:extrapolation with gadgets}
    The value of $\tr[O_A\rho_{\beta,\tar}]$ (respectively $\bra{\psi_{0,\tar}} O_A \ket{\psi_{0,\tar}}$) can be calculated up to any desired accuracy $\epsilon > 0$, given the values of $f'_\beta(x_k)$ (respectively $f'_{\ground}(x_k)$) at $m$ Chebyshev sample points $\{x_k\}_{k=1}^m$, where each $x_k$ is bounded above zero by $x_{\min} := \min_k x_k$, where:
    \begin{enumerate}[(I)]
        \item For $f_{\beta}'$,
        \begin{align}
            m = \bigO(\log^D (\epsilon^{-1}))\ ,\quad x_{\min} = \frac{\min\{1/\log^D(\epsilon^{-1}),1/\log^{1/d}(n\epsilon^{-1}) \}}{\log^{2D}(\epsilon^{-1})}\ .
        \end{align}
        \item For $f_{\ground}'$,
        \begin{align}
            m = \bigO(\log^D(\epsilon^{-1})) \ ,\quad x_{\min} = \frac{1}{\bigO(\log^{3D+1}(\epsilon^{-1}))} \ .
        \end{align}
    \end{enumerate}
    In particular, this process only requires simulator Hamiltonians with interaction strengths scaling as $\sim x_{\min}^{-d}$, which is $\poly\log(n\epsilon^{-1})$ for $f_\beta'$ and $\poly (\log (\epsilon^{-1}))$ for $f_{\ground}'$. These conclusions also hold using noisy estimates of $f_\beta'(x_k)$ (respectively $f_{\ground}'(x_k)$) each with additive error $\delta = \Theta(\epsilon/\log\log(\epsilon^{-1}))$.
\end{corollary}

\begin{proof}[*cor:extrapolation with gadgets]
    This follows from Theorem~\ref{thm:extrapolation with gadgets} and Corollary~\ref{cor:richardson extrapolation approx}, via a similar argument to the proof of Corollary~\ref{cor:extrapolating without gadgets explicit}. Concretely, we aim to choose parameters $M$, $m$, and $\delta$, such that
    \begin{align}
        \epsilon = (\delta + 2^{-m} M) \bigO(
        \log m)\ ,
    \end{align}
    where Theorem~\ref{thm:extrapolation with gadgets} fixes $M = \exp(\bigO(\log^D(\delta^{-1}))$. This is achievable by taking $\delta = \epsilon / \bigO(\log\log(\epsilon^{-1}))$ and $m = \bigO(\log^D(\epsilon^{-1}))$. Using the relationship $x_{\min} \sim R/m^2$ with the values of $R$ given by Theorem~\ref{thm:extrapolation with gadgets} gives the required result.
\end{proof}

Our proof of Theorem~\ref{thm:extrapolation with gadgets} will rely on the observation that Theorem~\ref{thm:simulation properties}\eqref{it:gibbs bound}-\eqref{it:ground state} ensures that 
\begin{align}
    f'_\beta(x) \approx \tr[O_{A,\eff}(x) \rho_{\beta,\eff}(x)]\quad\text{and}\quad f'_{\ground}(x) = \bra{\psi_{0,\eff}(x)} O_{A,\eff}(x) \ket{\psi_{0,\eff}(x)}\ ,
\end{align}
where $O_{A,\eff}(x)$ is defined as
\begin{align}\label{eq:oaeffdef}
    O_{A,\eff}(x) := (\id_{\eff} \otimes \bra{\mathbf{0}_{\anc}}) e^{T(x)} O_A e^{-T(x)} (\id_{\eff} \otimes \ket{\mathbf{0}_{\anc}})\ ,
\end{align}
where $\ket{\mathbf{0}_{\anc}}$ is the unique state on $\cH_{\anc}$ such that $P_0 = \proj{\mathbf{0}_{\anc}}$. In this way, we can reduce the extrapolation task in Theorem~\ref{thm:extrapolation with gadgets} --- which involves the singular Hamiltonian $H'(x)$ --- to an extrapolation task on the well-behaved (analytic) Hamiltonian $H_{\eff}(x)$. From here, the proof follows exactly as in Theorems~\ref{thm:gibbs state extrapolation} and \ref{thm:ground state extrapolation}. The only additional subtlety is that we must account for the fact that $O_{A,\eff}(x)$ no longer has constant-sized support, but we can use the locality of $T(x)$ given by Theorem~\ref{thm:simulation properties}\eqref{it:heff bound} with Lieb-Robinson bounds to argue that its support is approximately localised around $A$, made precise in the following lemma, which is proved in Appendix~\ref{app:gadget properties}:

\subsection{Quasi-locality of effective observables}

\begin{lemma}[See Lemma~\ref{lem:oaeff truncation app}]\label{lem:oaeff truncation}\noproofref
    Let $O_{A,\eff}(x)$ be the observable on $\Gamma_{\eff}$ defined in Eq.~\eqref{eq:oaeffdef}, and assume $x \leq 1/(2\theta)$. Then, for any $r\geq 0$, there exists an observable $O_{A,\eff}^{[r]}(x)$ on $\Gamma_{\eff}$, such that $O_{A,\eff}^{[r]}(x)$ is supported on $B_r(A)$, we have
    \begin{align}
        \left\|O_{A,\eff}(x) - O_{A,\eff}^{[r]}(x) \right\| &\leq a e^{-bg(r)}\ ,
    \end{align}
    for some constants $a,b > 0$. Moreover, $O_{A,\eff}^{[r]}(x)$ extends to a complex function $O_{A,\eff}^{[r]}(z)$ which is analytic on the disc $|z| \leq 1/(2\theta)$, and which is bounded as
    \begin{align}\label{eq:oaeff truncated bound}
        \sup_{|z| \leq 1/(2\theta)} \|O_{A,\eff}^{[r]}(z)\| &\leq a' e^{b' r^D}\ ,
    \end{align}
    for constants $a',b' > 0$.
\end{lemma}
In fact, Lemma~\ref{lem:oaeff truncation} is sufficient to guarantee that $O_{A,\eff}(x)$ remains (approximately) local even when transformed by the quantum belief propagation and spectral flow operators. This is stated in the following lemma, which can be viewed as a generalisation of Lemma~\ref{lem:qbp and spectral flow truncation}, and which is proved in Appendix~\ref{app:extensive properties}.

\begin{lemma}[See Lemma~\ref{lem:truncations with decaying observable app}]\label{lem:truncations with decaying observable}\noproofref
    Let $H$ be a Hamiltonian with bounded $\|\cdot\|_{F_g}$-norm for $g$ linear, and let $O$ be an observable with $\|O\| = \bigO(1)$ localised around $i\in \Gamma$ in the following sense: for every $r \geq 0$, there exists an observable $O^{[r]}$ with support contained within $B_r(\{i\})$ such that
    \begin{align}
        \|O - O^{[r]}\| \leq a e^{-bg(r)}\ ,
    \end{align}
    for some constants $a,b > 0$. Let $\Phi_H(O)$ and $\Psi_H(O)$ be the quantum belief propagation and spectral flow operators as defined in Eqs.~\eqref{eq:qbp operator} and \eqref{eq:spectral flow operator}. Then there exist constants $a_1',a_2',b_1',b_2'>0$ such that, for every $r \geq 0$, there exist operators $\Phi_H^{[r]}(O)$ and $\Psi_H^{[r]}(O)$ with support contained within $B_r(\{i\})$ such that
    \begin{align}
        \|\Phi_H(O) - \Phi_H^{[r]}(O)\| &\leq a_1' \|O\| e^{-b_1' g(r)} \ , \\
        \|\Psi_H(O) - \Psi_H^{[r]}(O)\| &\leq a_2' \|O\| e^{-b_2' g(r) / \log^2 g(r)}\ .
    \end{align}
\end{lemma}

We are now ready to prove Theorem~\ref{thm:extrapolation with gadgets} below.

\begin{proof}[*thm:extrapolation with gadgets]
    We define the functions $f_{\beta,\eff}(x)$ and $f_{\ground,\eff}(x)$ for $x\in [0,x_{\ast}]$ by
    \begin{align}
        f_{\beta,\eff}(x) &:= \tr[O_{A,\eff}(x)\rho_{\beta,\eff}(x)] \\
        f_{\ground,\eff}(x) &:= \bra{\psi_{0,\eff} (x)} O_{A,\eff}(x) \ket{\psi_{0,\eff}(x)}\ ,
    \end{align}
    where $O_{A,\eff}(x)$ is defined as in Eq.~\eqref{eq:oaeffdef}. By Theorem~\ref{thm:simulation properties}\eqref{it:gibbs bound}, we can bound
    \begin{align}
        |f_\beta'(x) - f_{\beta,\eff}(x)| &= \tr\left[O_A \left(\rho_\beta'(x) - e^{-T(x)} (\rho_{\beta,\eff}(x) \otimes P_0 ) e^{T(x)} \right) \right] \\
        &\leq \|O_A\| \left\| \rho_\beta ' (x) - e^{-T(x)} (\rho_{\beta,\eff}(x) \otimes P_0) e^{T(x)} \right\|_1 \\
        &\leq 2(m-1) \|O_A\| |\Gamma_{\anc}| e^{-\beta \Delta x^{-d} / 2}\ ,\label{eq:fbeta effective error}
    \end{align}
    whilst $f'_{\ground}(x) = f_{\ground,\eff}(x) $ is guaranteed by Theorem~\ref{thm:simulation properties}\eqref{it:gibbs bound}. Our problem thus reduces to finding analytic approximations for $f_{\beta,\eff}(x)$ and $f_{\ground,\eff}(x)$. Such analytic approximations exist by Theorem~\ref{thm:extrapolation without gadgets}, however there are a couple of additional subtleties in this case:
    \begin{itemize}
        \item Firstly, Theorem~\ref{thm:extrapolation without gadgets} is stated in terms of strictly local $O_A$ supported on $|A| = \bigO(1)$ sites. This is not satisfied by $O_{A,\eff}(x)$, however Lemma~\ref{lem:oaeff truncation} ensures that $O_{A,\eff}(x)$ can be truncated at finite radius up to exponentially small errors, and thus Lemma~\ref{lem:truncations with decaying observable} ensures that the corresponding quantum belief propagation and spectral flow operators can be truncated with the same asymptotic behaviour as for strictly local $O_A$, as given by Lemma~\ref{lem:qbp and spectral flow truncation}. Since the strict locality of $O_A$ is only necessary in order to obtain these truncation bounds, the proof of Theorem~\ref{thm:extrapolation without gadgets} goes through unchanged, except using Lemma~\ref{lem:truncations with decaying observable} in place of Lemma~\ref{lem:qbp and spectral flow truncation}.
        \item Secondly, in the statement of Theorem~\ref{thm:extrapolation without gadgets}, it is assumed that $\|O_A\| = \bigO(1)$, and this results in $M = \bigO(\|O_A\|) = \bigO(1)$ for both the Gibbs state and ground state parts. In our case, $\|O_{A,\eff}(x)\|$ may grow for complex $x \in \CC$, however the truncated observables $\|O^{[r]}_{A,\eff}(x)\|$ can be bounded following Lemma~\ref{lem:oaeff truncation} as $\sim e^{\bigO(r^D)}$. We ultimately take $r = \bigO(\log\delta^{-1})$ in the proof of Theorem~\ref{thm:extrapolation without gadgets}, and hence an additional factor of $\exp(\bigO(\log^D(\delta^{-1})))$ will appear in the $M$ expressions.
    \end{itemize}
    We deal with the quantitative conclusions of this for $f_{\beta,\eff}(x)$ and $f_{\ground,\eff}(x)$ separately:
    \begin{enumerate}[(I)]
        \item For $\delta > 0$, Theorem~\ref{thm:extrapolation without gadgets}\eqref{it:gibbs extrapolation without gadgets} gives a $(\delta/2,M,R')$-analytic approximation $\tilde{f}_{\beta,\eff}(z)$ for $f_{\beta,\eff}(x)$, where
        \begin{align}
            M = \exp\left( \bigO(\log^D(\delta^{-1})) \right)\ ,\quad R' = 1/\bigO(\log^D(\delta^{-1}))\ .
        \end{align}
        As mentioned, the $e^{\bigO(\log^D(\delta^{-1}))}$ factor in $M$ arises from the corresponding bound on $\|O_{A,\eff}^{[r]}(z)\|$ given by Eq.~\eqref{eq:oaeff truncated bound}, as we ultimately take $r \sim \log (\delta^{-1})$ in the proof of Theorem~\ref{thm:extrapolation without gadgets}. If we take $x \sim 1/\log^{1/d}(n\delta^{-1})$, the right-hand side of Eq.~\eqref{eq:fbeta effective error} can be upper bounded by $\delta / 2$, implying that $\tilde{f}_{\beta,\eff}(z)$ is a $(\delta,M,R)$-analytic approximation for $f'_\beta(x)$, where
        \begin{align}
            M = \exp\left(\bigO(\log^D(\delta^{-1})) \right)\ ,\quad R = \min\left\{1/\bigO(\log^D(\delta^{-1})), 1/\bigO(\log^{1/d}(n\delta^{-1})) \right\}\ .
        \end{align}
        \item For $\delta > 0$, Theorem~\ref{thm:extrapolation without gadgets}\eqref{it:gs extrapolation without gadgets} gives a $(\delta,M,R)$-analytic approximation $\tilde{f}_{\ground,\eff}(z)$ for $f_{\ground,\eff}(x)$ (and hence $f'_{\ground}(x)$), where
        \begin{align}
            M = \exp\left(\bigO(\log^{D}(\delta^{-1})) \right)\ ,\quad R =  1/\bigO(\log^{D+1}(\delta^{-1}))\ ,
        \end{align}
        by similar arguments to the previous case.
    \end{enumerate}
\end{proof}

\subsection{Example: locality reduction}\label{sec:example}

As a concrete example of the types of simulations made possible through the sorts of perturbative gadgets described by this work, we mention here the specific application of reducing $3$-local Hamiltonians to $2$-local Hamiltonians \cite{kempe2006complexity} (more generally, $k$-local to $2$-local reduction is possible for any $k\geq 3$ \cite{oliveira2005complexity} by recursive simulation). Such reductions are made possible through the use of a $3$-to-$2$ gadget, for example as introduced in Ref.~\cite{oliveira2005complexity}. This gives a degree-3 gadget $h'(x)$ on $\cH_{\eff}\otimes\cH_{\anc}$, where $\cH_{\eff} \simeq (\CC^2)^{\otimes 3}$, $\cH_{\anc} \simeq \CC^2$, which simulates an arbitrary interaction of the form $h := h_1\otimes h_2\otimes h_3$ on $\cH_{\eff}$, whilst $h'(x)$ is itself only $2$-local. Explicitly, $h'(x)$ is defined as
\begin{align}
    h'(x) &= x^{-3} \id \otimes \id \otimes \id \otimes \proj{1}\notag  \\
    &+ x^{-2} \frac{1}{\sqrt 2}\left( -h_1 \otimes \id + \id \otimes h_2 \right)  \otimes \id \otimes X - \id \otimes \id \otimes h_3 \otimes \proj{1}\notag \\
    &+ x^{-1} \frac{1}{2} \left(-h_1 \otimes \id + \id \otimes  h_2\right)^2 \otimes \id \otimes \id \notag \\
    &+ \frac{1}{2} \left(h_1^2 \otimes \id  + \id \otimes h_2^2 \right) \otimes h_3 \otimes \id\ .
\end{align}
The fact that $h'(x)$ is a gadget for $h$ follows from direct calculation of the local Schrieffer-Wolff perturbative series (or otherwise can be deduced by the previous analysis of Ref.~\cite{oliveira2005complexity} along with our result Lemma~\ref{lem:equivalent gadgets}). Since $h'(x)$ is a gadget of degree $d\leq 3$, many such gadgets can be combined in parallel (by our gadget combination result Lemma~\ref{thm:gadget comb}) to simulate an arbitrary 3-local Hamiltonian $H_{\tar}$ with a 2-local simulator Hamiltonian $H'(x)$. According to Corollary~\ref{cor:extrapolation with gadgets}, the local properties of the ground and Gibbs states of $H$ can be extrapolated from the analogous properties of $H'(x)$ for $x\sim 1/\poly\log (\epsilon^{-1})$, compared to the $1/\poly(n,\epsilon^{-1})$ scalings necessary for exact full-spectrum simulation.

The important caveat for this comparison is that the quasi-local effective Hamiltonian $H_{\eff}(x)$ obtained from $H'(x)$ via the local Schrieffer-Wolff transformation must remain uniformly non-critical along the path of extrapolation, as described in Assumption~\ref{assumption:non criticality}. This is in general difficult to prove rigorously, but we may heuristically expect this to hold for many disordered systems (see e.g. Ref.~\cite{sims2013many}); the numerics of Ref.~\cite{huang2022provably} show that their similar techniques work well empirically in many cases when non-criticality cannot be rigorously proved.

\section*{Acknowledgements}

We acknowledge financial support from the Novo Nordisk Foundation (Grant No. NNF20OC0059939 ‘Quantum for Life’ and Grant No. NNF25OC0105181 `Molecular Recognition from Quantum Computing'), and VILLUM FONDEN via the QMATH Centre of Excellence (Grant No.10059).

\bibliographystyle{alpha}
\bibliography{bib}

\newpage

\appendix

\section{Bounding the condition number}\label{app:condition number}

Below we give an elementary proof of the claim from Section~\ref{sec:richardsonextrapolation}, that the condition number $\alpha(\mathbf{x})$ arising from Chebyshev nodes $\{x_k\}_{k=1}^m$ is bounded as $\bigO(\log(m))$. 

\begin{lemma}\label{lem:condition number}
    Let $m \geq 2$, and define $\{x_k\}_{k=1}^m$ by
    \begin{align}
        x_k = \sin^2\left( \frac{(2k-1)\pi}{4m} \right)\ .
    \end{align}
    and let
    \begin{align}
        \alpha(\mathbf{x}) := \sum_{k=1}^m \prod_{j\neq k} \left|\frac{x_j}{x_j - x_k} \right|\ .
    \end{align}
    Then we can bound
    \begin{align}
        \alpha(\mathbf{x}) &\leq 3\log(m)\ .
    \end{align}
\end{lemma}

\begin{proof}[*lem:condition number]
    Using the identity $\sin^2\theta - \sin^2\phi = \sin(\theta + \phi) \sin(\theta - \phi)$, we have
    \begin{align}
        x_j - x_k &= \sin\left(\frac{(2j-1) \pi + (2k-1)\pi}{4m}\right)\sin\left( \frac{(2j-1)\pi - (2k-1)\pi}{4m}\right) \\
        &= \sin\left(\frac{(j + k - 1)\pi}{2m} \right)\sin\left(\frac{(j-k)\pi}{2m} \right)\ ,
    \end{align}
    and hence
    \begin{align}
        \prod_{j\neq k} \left| \frac{x_j}{x_j - x_k} \right| &= \left| \frac{\left(\prod_{j\neq k} \sin\left(\frac{(2j-1)\pi}{4m} \right) \right)^2}{\left(\prod_{j\neq k} \sin\left(\frac{(j+k-1) \pi}{2m}\right) \right) \left(\prod_{j\neq k} \sin\left(\frac{(j-k) \pi}{2m}\right) \right)} \right|\ .
    \end{align}
    For the evaluation of these products we will make repeated use of the following identity:
    \begin{align}\label{eq:trig product identity}
        \prod_{k = 1}^{m-1} \sin\left(\frac{k\pi}{m}\right) = \frac{m}{2^{m-1}}\ .
    \end{align}
    For the product in the numerator we have
    \begin{align}
        \left(\prod_{j\neq k} \sin\left( \frac{(2j-1)\pi}{4m}\right)\right)^2 &= \frac{\prod_{j=1}^{2m} \sin\left(\frac{(2j-1)\pi}{4m} \right)}{\sin^2 \left(\frac{(2k-1)\pi}{4m} \right)} \\
        &= \frac{\prod_{j=1}^{4m - 1} \sin\left(\frac{j\pi}{4m} \right)}{\prod_{j=1}^{2m-1} \sin\left(\frac{j\pi}{2m} \right)} \cdot \frac{1}{\sin^2 \left(\frac{(2k-1)\pi}{4m}\right)} \\
        &= \frac{1}{2^{2m-1} \sin^2 \left(\frac{(2k-1)\pi}{4m} \right)}\ .
    \end{align}
    In the first line we used that the terms corresponding to $j$ and $m-j$ are identical, and in the third line we used Eq.~\eqref{eq:trig product identity} twice.

    Meanwhile, for the denominator we have
    \begin{align}
        \left| \left(\prod_{j\neq k} \sin\left(\frac{(j+k-1)\pi}{2m} \right) \right) \left( \prod_{j\neq k} \sin\left(  \frac{(j-k)\pi}{2m}\right)\right)\right| &= \left| \frac{1}{\sin\left( \frac{(2k-1) \pi}{2m}\right)} \prod_{\substack{j=-m \\ j\neq k}}^{m-1} \sin\left( \frac{(j+k)\pi}{2m}\right) \right| \\
        &= \frac{1}{\sin\left( \frac{(2k-1)\pi}{2m}\right)} \prod_{j=1}^{2m-1} \sin\left( \frac{j\pi}{2m}\right) \\
        &= \frac{2m}{2^{2m-1} \sin\left(\frac{(2k-1)\pi}{2m} \right)}\ ,
    \end{align}
    where in the last line we used Eq.~\eqref{eq:trig product identity}. Hence
    \begin{align}
        \prod_{j\neq k} \left| \frac{x_j}{x_j - x_k}\right| &= \frac{\sin\left(\frac{(2k-1)\pi}{2m} \right)}{2m\sin^2 \left(\frac{(2k-1)\pi}{4m} \right)} \\
        &= \frac{1}{m} \cot\left(\frac{(2k-1)\pi}{4m} \right)\ .
    \end{align}
    We can then write
    \begin{align}
        \alpha(\mathbf{x}) &= \frac{1}{m}\sum_{k=1}^m \cot\left( \frac{(2k-1) \pi}{4m}\right)\ .
    \end{align}
    Using that $\cot\theta \leq \theta^{-1}$ for $\theta \in (0,\pi/2)$, we then have
    \begin{align}
        \alpha(\mathbf{x}) &\leq \frac{4}{\pi} \sum_{k=1}^m \frac{1}{2k-1} \\
        &\leq 3\log (m)\ ,
    \end{align}
    where the last inequality holds for $m\geq 2$.
\end{proof}

\section{Localising quantum belief propagation and the spectral flow}\label{app:truncation results}

In this section, we will prove some basic results involving local truncations of the quantum belief propagation and spectral flow operators defined in Section~\ref{sec:many body}. The qualitative conclusions of this section are not novel and are stated and proved here for completeness and consistency with our conventions; the locality of the quantum belief propagation operator has been informally proved as early as Ref.~\cite{kim2012perturbative}, Corollary 3, and more recently applied to prove the GALI property for Gibbs states in Ref.~\cite{rouze2024efficient}. Meanwhile, a similar result concerning the locality of the spectral flow operator appears in Ref.~\cite{nachtergaele2019quasi} Section 6.5, which this work follows closely.

\subsection{Alternative representation of the spectral flow}

Firstly, we prove the following Lemma, stated in the main text as Lemma~\ref{lem:spectral flow alternative}, which expresses the spectral flow operator $\Psi_H(X)$ in a form which will be helpful for our subsequent analysis.

\begin{lemma}[Restatement of Lemma~\ref{lem:spectral flow alternative}]\label{lem:spectral flow alternative app}
    The spectral flow operator $\Psi_H(O)$ as defined in \cref{lem:spectral flow} can be written as
    \begin{align}
        \Psi_H(X) = \int_{-\infty}^\infty \tilde{w}_\gamma(t) e^{itH} X e^{-itH}\ ,
    \end{align}
    where $\tilde{w}_\gamma : \RR \rightarrow \RR$ is an odd $L_1$ function with $|\tilde{w}_\gamma(t)| \leq 1/2$ for all $t$. Moreover, for $|t|\geq e^3 \gamma^{-1}$,
    \begin{align}
        |\tilde{w}_\gamma(t)| &\leq W_1 \left(\frac{\gamma t}{\log^2(\gamma t)} \right)^2  \exp\left( -\frac{2\gamma t}{7\log^2(\gamma t)}\right) \ ,\\
        \int_t^\infty \diff s \tilde{w}_\gamma(s) &\leq W_2 \gamma^{-1} \left( \frac{\gamma t}{\log^2(\gamma t)}\right)^3 \exp\left(-\frac{2\gamma t}{7\log^2(\gamma t)}\right)\ ,
    \end{align}
    for some constants $W_1,W_2 > 0$. Moreover,
    \begin{align}
        \int_{-\infty}^\infty \diff t | \tilde{w}_\gamma(t) |\leq W_3 \gamma^{-1} \ ,
    \end{align}
    for some constant $W_3 > 0$.
\end{lemma}

\begin{proof}[*lem:spectral flow alternative app]
    By choosing the function $w_\gamma$ given by \cite{bachmann2012automorphic} (which in particular is even and non-negative) and reordering the integrals, we rewrite Eq.~\eqref{eq:spectral flow operator} to give 
    \begin{align}
        \Psi_H(X) &=\int_{-\infty}^\infty \diff t w_\gamma(t) \int_0^t \diff u e^{iuH} X e^{-iuH} \\
        &= \int_0^\infty \diff t w_\gamma (t) \int_0^t \diff u e^{iuH} X e^{-iuH} - \int_{-\infty}^0 \diff t w_\gamma(t) \int_t^0\diff u e^{iuH} X e^{-iuH}\\
        &= \int_0^\infty \diff u \int_u^\infty \diff t w_\gamma(t) e^{iuH} X e^{-iuH} - \int_{-\infty}^0 \diff u \int_{-\infty}^u \diff t w_\gamma(t) e^{iuH} X e^{-iuH} \\
        &= \int_{-\infty}^\infty \diff t \left( \sign(t) \int_{|t|}^\infty \diff u w_\gamma(u)\right) e^{itH} X e^{-itH}\ ,
    \end{align}
    where in the last line we have used that $w_\gamma$ is an even function, and relabelled $t\leftrightarrow u$. It remains to prove the stated bounds on the function
    \begin{align}
        \tilde{w}_\gamma(t) := \int_{|t|}^\infty \diff u w_\gamma(u)\ .
    \end{align}
    The fact that $|\tilde{w}_\gamma(t)| \leq 1/2$ for all $t$ follows by the normalisation of $w_\gamma(t)$.
    
    Using the bound given by Eq.~\eqref{eq:wgamma inequality} for $t\geq e^{-1/\sqrt{2}} \gamma^{-1}$, we have
    \begin{align}
        \tilde{w}_\gamma(t) &\leq 2(e\gamma)^2 \int_t^\infty \diff u \cdot u \exp\left(-\frac{2\gamma u}{7\log^2 (\gamma u)}\right) \\
        &= 2e^2 \int_t^\infty \diff u \left(\frac{\log^7(\gamma u)}{\gamma u(\log(\gamma u) - 2)} \right) \left( \gamma^3 u^2 \cdot \frac{\log(\gamma u) - 2}{\log^7(\gamma u)}\right) \exp\left(-\frac{2\gamma u}{7\log^2 (\gamma u)}\right)\ .
    \end{align}
    We may bound
    \begin{align}
        \frac{\log^7(x)}{x(\log(x) - 2)} \leq 200\quad \text{for $x \geq e^3$\ ,}
    \end{align}
    and hence assuming that $t \geq e^3 \gamma^{-1}$ we have
    \begin{align}
        \tilde{w}_\gamma(t) &\leq 400 e^2 \int_t^\infty \diff u \left(\gamma^3 u^2 \cdot \frac{\log(\gamma u) - 2}{\log^7(\gamma u)} \right)\exp\left(-\frac{2\gamma u}{7\log^2 (\gamma u)}\right) \\
        &= 400e^2 \int_{\frac{\gamma t}{\log^2(\gamma t)}}^\infty \diff v\cdot v^2 e^{-\frac{2v}{7}}\ ,
    \end{align}
    where in the second line we performed the change of variable
    \begin{align}
        v = \frac{\gamma u}{\log^2(\gamma u)} \Rightarrow \diff v = \frac{\gamma(\log(\gamma u) - 2)}{\log^3(\gamma u)} \diff u\ .
    \end{align}
    Evaluating the integral, this gives
    \begin{align}
        \tilde{w}_\gamma(t) &\leq 700e^2 \left[ e^{-2v/7} (2v^2 + 14v + 49)\right]_{v = \frac{\gamma t}{\log^2(\gamma t)}} \\ 
        &\leq 2100e^2 \left[ e^{-2v/7} v^2\right]_{v = \frac{\gamma t}{\log^2(\gamma t)}} \\
        &= W_1 \left(\frac{\gamma t}{\log^2(\gamma t)} \right)^2  \exp\left( \frac{7\gamma t}{\log^2(\gamma t)}\right)\ ,
    \end{align}
    where the second line holds for $\gamma t \geq e^3$, and $W_1 = 2100e^2 \approx 1.55\times 10^4$.

    Using this result, we can bound the integral
    \begin{align}
        \int_t^\infty \diff s \tilde{w}_\gamma(s) &\leq 2100e^2 \int_t^\infty \diff s \left(\frac{\gamma s}{\log^2(\gamma s)} \right)^2 \exp\left( -\frac{2\gamma s}{7\log^2(\gamma s)}\right) \\
        &= 2100e^2 \int_t^\infty \diff s \left(\frac{\gamma^4 s^3(\log(\gamma s) - 2)}{\log^9(\gamma s)} \right) \left(\frac{\log^5(\gamma s)}{\gamma^2 s(\log(\gamma s) - 2)} \right) \exp \left( -\frac{2\gamma s}{7\log^2(\gamma s)}\right)\ .
    \end{align}
    This time, we may bound
    \begin{align}
        \frac{\log^5(x)}{x(\log(x) - 2)} \leq 25\quad\text{for $x \geq e^3$}\ ,
    \end{align}
    and hence
    \begin{align}
        \int_t^\infty \diff s \tilde{w}_\gamma(s) &\leq 52500e^2 \gamma^{-1}\int_t^\infty \diff s \left( \frac{\gamma^4 s^3 (\log(\gamma s) - 2)}{\log^9(\gamma s)} \right) \exp\left( -\frac{2\gamma s}{7\log^2(\gamma s)} \right) \\
        &= 52500e^2\gamma^{-1} \int_{\frac{\gamma t}{\log^2(\gamma t)}}^\infty \diff v \cdot v^3 e^{-2v/7}\ ,
    \end{align}
    where we have once again used the change of variables $v = \gamma s / \log^2(\gamma s)$. We can evaluate the integral to give
    \begin{align}
        \int_t^\infty \diff s \tilde{w}_\gamma(s) &\leq 46000 e^2 \gamma^{-1} \left[ e^{-2v / 7} (4v^3 + 42v^2 + 294v + 1029)\right]_{v = \frac{\gamma t}{\log^2(\gamma t)}} \\
        &\leq 322000e^2\gamma^{-1} \left[e^{-2v/7} v^3 \right]_{v = \frac{\gamma t}{\log^2(\gamma t)}} \\
        &= W_2 \gamma^{-1} \left( \frac{\gamma t}{\log^2(\gamma t)}\right)^3 \exp\left(\frac{2\gamma t}{7\log^2(\gamma t)}\right)\ ,
    \end{align}
    where the second inequality holds for $\gamma t \geq e^3$, and $W_2 = 322000e^2 \approx 2.38\times 10^6$.

    Finally, applying the previous bound, we see that
    \begin{align}
        \int_{8500\gamma^{-1}}^\infty \diff s \tilde{w}_\gamma(s) \leq 2^{-1}\gamma^{-1}\ .
    \end{align}
    So we can write
    \begin{align}
        \int_{-\infty}^\infty\diff t \tilde{w}_\gamma(t) &= 2\int_0^{8500 \gamma^{-1}} \diff t \tilde{w}_\gamma(t) + 2\int_{8500\gamma^{-1}}^\infty \diff t \tilde{w}_\gamma(t) \\
        &\leq 8501\gamma^{-1} = W_3 \gamma^{-1}\ ,
    \end{align}
    where we defined $W_3 = 8501$.
\end{proof}

\subsection{Local truncations}

In this section, we prove Lemma~\ref{lem:qbp and spectral flow truncation}. This establishes that, for a local observable $O_A$ supported on $A\subseteq \Gamma$, the operators $\Phi_H(O_A)$ and $\Psi_H(O_A)$ can both be well-approximated by operators which act only within a radius $r\geq 0$ of $A$, up to an error which decays quickly with $r$.

We first prove the following lemma, which shows that operators $\Lambda(O_A)$, of a form which generalises the quantum belief propagation and spectral flow operators, can be locally truncated. This will then be specialised to give the proof of Lemma~\ref{lem:qbp and spectral flow truncation}.

\begin{lemma}\label{lem:operator integral truncation}
    Let $H$ be a Hamiltonian on $\Gamma$ with bounded interactions with respect to the $\|\cdot\|_{F_g}$ norm as in \cref{lem:lr bounds}. Let $O_A$ be a local observable supported on $A\subseteq \Gamma$, and let $w : \RR \rightarrow \RR$ be a non-negative $L_1$ function. We define the operator $\Lambda(O_A)$ (as a generalisation of the quantum belief propagation and spectral flow operators) by
    \begin{align}
        \Lambda(O_A) = \int_{-\infty}^\infty \diff t w(t) e^{itH} O_A e^{-itH}\ ,
    \end{align}
    for constants $c,\nu>0$ depending on $C_{F_g}$ and $\|F_g\|$. Then for any $r\geq0$ there exists an operator $\Lambda^{[r]}(O_A)$ which acts as the identity outside of $B_r(A)$ such that
    \begin{align}
        \|\Lambda(O_A) - \Lambda^{[r]}(O_A)\| \leq \|O_A\| \int_{-\infty}^\infty \diff t w(t) \min\left\{2,c|A|\left(e^{\nu\|H\|_{F_g} |t|} - 1\right) e^{-g(r)}\right\}\ .
    \end{align}
    Moreover, assuming that $w(t)$ is even and normalised by $\int_{\RR} \diff t w(t) \leq 1$, the following bound holds:
    \begin{align}\label{eq:operator integral truncation}
        \|\Lambda(O_A) - \Lambda^{[r]}(O_A)\| \leq \|O_A\| \inf_{\delta \in [0,\infty)} \left\{c|A| e^{\nu \|H\|_{F_g} \delta - g(r)} + 4\int_{\delta}^\infty \diff t w(t)\right\}\ .
    \end{align}
\end{lemma}

To prove Lemma~\ref{lem:operator integral truncation}, we will define $\Lambda^{[r]}(O_A)$ simply by taking a partial trace on the sites outside $B_r(A)$, that is
\begin{align}
    \Lambda^{[r]}(O_A) := \frac{1}{\dim \cH_{\Gamma \setminus B_r(A)}} \id_{\Gamma \setminus B_r(A)} \otimes \tr_{\Gamma \setminus B_r(A)} [\Lambda(O_A) ]\ .
\end{align}
This is a slight abuse of notation; note that $\Lambda^{[r]}(O_A)$ depends on the choice of $A$, not just $O_A$.

\begin{proof}[*lem:operator integral truncation]
    Letting $\cU_{\Gamma\setminus B_r(A)}$ denote the Haar measure over unitaries on $\cH_{\Gamma\setminus B_r(A)}$ (normalised as $\int_U \diff U = 1$), we can write
    \begin{align}
        \Lambda^{[r]}(O_A) = \int_{U \sim \cU_{\Gamma\setminus B_r(A)}} \diff U \left(U \otimes \id_{B_r(A)}\right) \Lambda(O_A) \left(U^\dagger \otimes \id_{B_r(A)}\right) \ .
    \end{align}
    We can then bound
    \begin{align}
        \| \Lambda^{[r]}(O_A) - \Lambda(O_A)\| &\leq \left\| \int_{U\sim \cU_{\Gamma\setminus B_r(A)}} \diff U \left[\left(U \otimes \id_{B_r(A)}\right) \Lambda(O_A) \left(U^\dagger \otimes \id_{B_r(A)}\right) - \Lambda(O_A) \right] \right\| \\
        &\leq \sup_U \left\|[U\otimes \id_{B_r(A)},\Lambda(O_A)]\right\|\ .
    \end{align}
    By the definition of $\Lambda(O_A)$ we have
    \begin{align}
        \left\|[U \otimes \id_{B_r(A)},\Lambda(O_A)] \right\| \leq \int_{-\infty}^\infty \diff t w(t) \left\|[U\otimes \id_{B_r(A)},e^{itH} O_A e^{-itH}] \right\|\ .
    \end{align}
    Applying \cref{lem:lr bounds}, we can bound the commutator by
    \begin{align}
        \left\|[U\otimes \id_{B_r(A)},e^{itH} O_A e^{-itH}] \right\| &\leq c\|O_A\| |A| (e^{\nu \|H\|_{F_g} |t|} - 1) e^{-g(r)}\ .
    \end{align}
    Since $U$ is unitary, we can also upper bound this by $2\|O_A\|$. Putting this together, we arrive at the bound
    \begin{align}
        \|\Lambda(O_A) - \Lambda^{[r]}(O_A)\| \leq \|O_A\| \int_{-\infty}^\infty \diff t w(t) \min\left\{2,c|A|\left(e^{\nu\|H\|_{F_g} |t|} - 1\right) e^{-g(r)}\right\}\ .
    \end{align}
    Assuming that $w(t)$ is an even function, we can take the integral only over the positive real line, picking up a factor of two. Moreover, for any $\delta \geq 0$, we can split this integral into two parts: the $[0,\delta)$ interval, and the $[\delta,\infty)$ interval. This leads to the following bound:
    \begin{align}
        \|\Lambda(O_A) - \Lambda^{[r]}(O_A)\| \leq 2 \|O_A\| \left( \int_0^\delta \diff t w(t) c|A| \left( e^{\nu \|H\|_{F_g} t} - 1\right)e^{-g(r)} + \int_\delta^\infty \diff t\cdot  2 w(t) \right) \ .
    \end{align}
    Using the normalisation of $w(t)$, we can upper bound the first integral to obtain
    \begin{align}
        \|\Lambda(O_A) - \Lambda^{[r]}(O_A)\| \leq \|O_A\| \left( c|A| e^{\nu \|H\|_{F_g} \delta - g(r)} + 4\int_\delta^\infty \diff t w(t) \right)\ ,
    \end{align}
    from which the result follows.
\end{proof}

We now prove Lemma~\ref{lem:qbp and spectral flow truncation}, restated below.

\begin{lemma}[Restatement of Lemma~\ref{lem:qbp and spectral flow truncation}]\label{lem:qbp and spectral flow truncation app}
    Let $H$ be a Hamiltonian with bounded $\|\cdot\|_{F_g}$-norm, and let $O_A$ be an observable supported on $A \subseteq \Gamma$. Let $\Phi_H(O_A)$ be the quantum belief propagation operator defined by Eq.~\eqref{eq:qbp operator}, and $\Psi_H(O_A)$ be the spectral flow operator defined by Eq.~\eqref{eq:spectral flow operator}. Then for every $r\geq0$ there exist operators $\Phi_H^{[r]}(O_A)$ and $\Psi_H^{[r]}(O_A)$ which act as the identity outside of $B_r(A)$, and constants $a_1,b_1> 0$ (depending on $c,\beta,\nu,\|H\|_{F_g}$) and $a_2,b_2 > 0$ (depending on $c,\gamma,\nu,\|H\|_{F_g}$), such that
    \begin{align}
        \|\Phi_H(O_A) - \Phi_H^{[r]}(O_A) \| &\leq a_1 |A| \|O_A\| e^{-b_1 g(r)}\ , \label{eq:truncation w constants qbp app}\\
        \|\Psi_H(O_A) - \Psi_H^{[r]}(O_A) \| &\leq a_2 |A| \|O_A\| e^{-b_2 g(r) /\log^2 g(r)}\ ,\label{eq:truncation w constants sf app}
    \end{align}
    where $\omega(x)$ is defined by
    \begin{align}
        \omega(x) := \left( \frac{x}{\log^2(x)}\right)^3 \exp\left( -\frac{2x}{7\log^2 x}\right)\ .
    \end{align}
\end{lemma}

\begin{proof}[*lem:qbp and spectral flow truncation app]
    Application of \cref{lem:operator integral truncation} immediately yields the operator $\Phi_H^{[r]}(O_A)$, with the bound
    \begin{align}
        \| \Phi_H(O_A) - \Phi_H^{[r]}(O_A) \| \leq \inf_{\delta \in [0,\infty)} \|O_A\| \left( c|A| e^{\nu \|H\|_{F_g} \delta - g(r)} + 4\int_\delta^\infty \diff t \kappa_\beta (t)\right)
    \end{align}
    where $\kappa_\beta(t)$ is the function defined in Eq.~\eqref{eq:kappa function}. By the bound given in Eq.~\eqref{eq:kappa function} (see Ref.~\cite{rouze2024efficient}), we have for $t\geq 0$ that
    \begin{align}
        \kappa_\beta(t) \leq \frac{1}{e^{\frac{\pi t}{\beta}} - 1} = e^{-\frac{\pi t}{2\beta}} \cdot\frac{1}{2\sinh\left(\frac{\pi t}{2\beta} \right)}\ .
    \end{align}
    In particular, choosing $\delta \geq \delta_\ast$ where
    \begin{align}
        \delta_\ast = \frac{2\beta}{\pi}\sinh^{-1}(1) \ ,
    \end{align}
    we have
    \begin{align}
        \int_\delta^\infty \diff t\kappa_\beta(t) &\leq \int_\delta^\infty \diff t \frac{1}{2}e^{-\frac{\pi t}{2\beta}} = \frac{\beta}{\pi} e^{-\frac{\pi \delta}{2\beta}}\ .
    \end{align}
    Hence we can write
    \begin{align}
        \|\Phi_H(O_A) - \Phi_H^{[r]}(O_A)\| &\leq \inf_{\delta \in [\delta_\ast,\infty)} \|O_A\| \left(c|A| e^{\nu \|H\|_{F_g}\delta - g(r)} + \frac{4\beta}{\pi} e^{-\frac{\pi \delta}{2\beta}} \right) \\
        &\leq \|O_A\| \left( c|A| e^{\frac{2\beta}{\pi} \nu \|H\|_{F_g} \sinh^{-1}(1)} + \frac{4\beta}{\pi}\right) \exp\left[ -\frac{\pi g(r)}{2\beta \nu\|H\|_{F_g} + \pi}\right]\ ,
    \end{align}
    where in the last inequality we have set
    \begin{align}
        \delta = \delta_\ast + \frac{g(r)}{\nu \|H\|_{F_g} + \frac{\pi}{2\beta}}\ 
    \end{align}
    for an upper bound.

    For the spectral flow, we may apply \cref{lem:spectral flow alternative} and \cref{lem:operator integral truncation} to obtain $\Psi_H^{[r]}(O_A)$ which satisfies
    \begin{align}
        \|\Psi_H(O_A) - \Psi_H^{[r]}(O_A)\| &\leq \|O_A\| \inf_{\delta \in [e^3 \gamma^{-1},\infty)} \bigg\{c|A| W_3 \gamma^{-1} e^{\nu \|H\|_{F_g}\delta - g(r)} \\ 
        &\quad + 4W_2 \gamma^{-1}\left(\frac{\gamma \delta}{\log^2(\gamma \delta)} \right)^3 \exp\left(- \frac{2\gamma \delta}{7\log^2(\gamma \delta)}\right) \bigg\}\ ,
    \end{align}
    where $W_2,W_3 > 0$ are the constants from \cref{lem:spectral flow alternative}. We define the shorthand
    \begin{align}
        \omega(x) := \left( \frac{x}{\log^2 x} \right)^3 \exp\left(-\frac{2x}{7\log^2x} \right) \geq e^{-2x/7}\ ,
    \end{align}
    where the inequality holds for $x \geq e^3$. Choosing
    \begin{align}
        \delta = e^3 \gamma^{-1 }+ \frac{g(r)}{\nu \|H\|_{F_g} + \frac{2}{7}\gamma}\ ,
    \end{align}
    we therefore obtain the upper bound
    \begin{align}
        \|\Psi_H(O_A) - \Psi_H^{[r]}(O_A)\| &\leq \|O_A\|\gamma^{-1}\left(c|A| W_3 e^{\nu \|H\|_{F_g} e^3 \gamma^{-1}} + 4W_2 \right) \omega\left( \frac{g(r)}{\nu \gamma^{-1} \|H\|_{F_g} + \frac{2}{7}}\right)\ ,
    \end{align}
    as required.
\end{proof}

\subsection{Generalised approximate local indistinguishability}

We now apply the results of the previous section to establish the GALI property for parametrised Gibbs states and ground states, stated in the main text as Lemmas~\ref{lem:GALI gibbs} and \ref{lem:GALI ground}. These arguments are essentially the same as those in Ref.~\cite{rouze2024efficient}.

\begin{lemma}[Restatement of Lemma~\ref{lem:GALI gibbs}]\label{lem:GALI gibbs app}
    Let $\{H(x)\}_{x\in \cX}$ be a family of Hamiltonians with continuous first derivative with respect to $x$ on the interval $\cX$, such that $\|H\|_{F_g}$ and $\|\partial_x H\|_{F_g}$ are both bounded. Let $O_A$ be an observable supported on $A\subseteq \Gamma$, and assume that $\partial_s H(s)$ contains no terms with support in $B_{r_0}(A)$ for some $r_0\geq 0$.

    Let $\rho_\beta(x)$ be the associated family of Gibbs states at temperature $\beta = \bigO(1)$, and assume that these satisfy a uniform exponential decay of correlations with parameters $K,\xi > 0$. Then, for all $x$, it holds that
    \begin{align}
        \left|\Cov_{\rho_\beta(x)} (\Phi_{H(x)}(O_A),\partial_x H(x)) \right| \leq c_1\|O_A\| |A|^3 \sum_{r = r_0}^\infty \left( r^{3D} e^{-r/2\xi} + r^D e^{-b_1 g(r/2)} \right)\ ,
    \end{align}
    for positive constants $b_1,c_1>0$.
    Hence, defining $f_\beta(x) := \tr[O_A\rho_\beta(x)]$, we have for all $x_0,x_1 \in \cX$ that
    \begin{align}\label{eq:GALI gibbs function app}
        |f_\beta(x_1) - f_\beta(x_0)| &\leq \beta |x_1 - x_0| c_1 \|O_A\| |A|^3 \sum_{r = r_0}^\infty \left( r^{3D} e^{-r/2\xi} + r^D e^{-b_1 g(r/2)} \right)\ .
    \end{align}
\end{lemma}

\begin{proof}[*lem:GALI gibbs app]
    In the below, we will suppress some of the $x$-dependences for convenience. We write
    \begin{align}\label{eq:partial x H decomposition}
        \partial_x H = \sum_{r=r_0}^\infty \partial_x H^{[r]}\ ,
    \end{align}
    where each $H^{[r]}$ contains the terms in $H$ whose support intersects with $B_r(A)$, but not with $B_{r'}(A)$ for any $r'<r$. Using the triangle inequality, we can thus write
    \begin{align}\label{eq:qbp covariance sum triangle ineq}
        \left|\Cov_{\rho_\beta} (\Phi_{H}(O_A),\partial_x H )\right| &\leq \sum_{r = r_0}^\infty \left| \Cov_{\rho_\beta} (\Phi_H(O_A), \partial_x H^{[r]} )\right|\ .
    \end{align}
    For each $r \geq r_0$, we can decompose
    \begin{align}
        \Phi_H(O_A) = \Phi_H^{[\lfloor r/2 \rfloor]}(O_A) + (\Phi_H(O_A) -  \Phi_H^{[\lfloor r/2 \rfloor]}(O_A) )\ ,
    \end{align}
    where by \cref{lem:qbp and spectral flow truncation} the operator $\Phi_H^{[\lfloor r/2 \rfloor]}(O_A)$ is supported only within radius $\lfloor r/2\rfloor$ of $A$, and
    \begin{align}
        \|\Phi_H(O_A) -  \Phi_H^{[\lfloor r/2 \rfloor]}(O_A)\| \leq a_1 |A| \|O_A\| e^{-b_1 g(r/2)}\ .
    \end{align}
    Hence, by the triangle inequality
    \begin{align}
        \left| \Cov_{\rho_\beta} (\Phi_H(O_A), \partial_x H^{[r]} )\right| &\leq \left| \Cov_{\rho_\beta} (\Phi^{[\lfloor r/2\rfloor ]}(O_A),\partial_x H^{[r]}) \right| + 2a_1 |A| \|O_A\| e^{-b_1 g(r/2)} \|\partial_x H^{[r]}\| \\
        &\leq 2\|O_A\| \|\partial_x H^{[r]}\| |B_r(A)|^2 K e^{-r/2\xi}  + 2a_1 |A| \|O_A\| e^{-b_1 g(r/2)} \|\partial_x H^{[r]}\| \\
        &\leq 2\|O_A\| \|\partial_x H^{[r]}\| |A| \left( |A| K k_D^2 r^{2D} e^{-r/2\xi} + a_1 e^{-b_1 g(r/2)}\right) \ ,
    \end{align}
    where in the second line we used the assumption of exponential correlation decay. Using \cref{lem:hamiltonian restriction}, we can bound
    \begin{align}\label{eq:partial x H r bound}
        \|\partial_x H^{[r]}\| \leq |B_r(A)| \|\partial_x H\|_{F_g} \|F_g\| \leq k_D r^D |A| \|\partial_x H\|_{F_g} \|F_g\|\ ,
    \end{align}
    and hence we can write
    \begin{align} 
        \left| \Cov_{\rho_\beta} (\Phi_H(O_A), \partial_x H^{[r]} )\right| &\leq 2\|O_A\| k_D |A|^3 \|\partial_x H\|_{F_g} \|F_g\| (k_D^2 K + a_1)\left(r^{3D} e^{-r/2\xi}  + r^D e^{-b_1 g(r/2)} \right) \\
        &= c_1 \|O_A\| |A|^3 \left( r^{3D} e^{-r/2\xi} + r^D e^{-b_1 g(r/2)} \right)\ ,
    \end{align}
    where we have defined the constant
    \begin{align}
        c_1 := 2 k_D \|\partial_x H\|_{F_g} \|F_g\| (k_D^2 K + a_1) = \bigO(1) \ .
    \end{align}
    Applying Eq.~\eqref{eq:qbp covariance sum triangle ineq}, we can bound the overall covariance by
    \begin{align}
        \left|\Cov_{\rho_\beta}(\Phi_H(O_A),\partial_x H) \right| &\leq c_1\|O_A\| |A|^3 \sum_{r = r_0}^\infty \left( r^{3D} e^{-r/2\xi} + r^D e^{-b_1 g(r/2)} \right)\ .
    \end{align}
    Defining $f_\beta(x) := \tr[O_A \rho_\beta(x)]$, we can apply Eq.~\eqref{eq:qbp observable} to immediately deduce that
    \begin{align}
        |f_\beta(x_1) - f_\beta(x_0)| &\leq \int_{x_0}^{x_1} \diff x \left| \frac{\diff}{\diff x} f_\beta(x)\right| \\
        &\leq \beta \int_{x_0}^{x_1} \diff x \left| \Cov_{\rho_\beta(x)} (\Phi_{H(x)}(O_A),\partial_x H(x))\right| \\
        &\leq \beta |x_1 - x_0| c_1 \|O_A\| |A|^3 \sum_{r = r_0}^\infty \left( r^{3D} e^{-r/2\xi} + r^D e^{-b_1 g(r/2)} \right)\ ,
    \end{align}
    as required.
\end{proof}

\begin{lemma}[Restatement of Lemma~\ref{lem:GALI ground}]\label{lem:GALI ground app}
    Let $\{H(x)\}_{x\in \cX}$ and $O_A$ be as in \cref{lem:GALI gibbs app}, and assume that $\{H(x)\}_{x\in \cX}$ has a uniform gap $\gamma > 0$ above its ground state $\ket{\psi_0(x)}$. Then
    \begin{align}
        \left\| [\Psi_{H(x)}(O_A), \partial_x H(x)]\right\| \leq c_2 \|O_A\| |A|^2 \sum_{r=r_0}^\infty r^D e^{-b_2g(r-1)/\log^2 g(r-1)}\ ,
    \end{align}
    for positive constants $b_2,c_2 >0$. Hence, defining $f_{\ground}(x) := \bra{\psi_0(x)} O_A \ket{\psi_0(x)}$, we have for all $x_0,x_1 \in \cX$ that
    \begin{align}\label{eq:GALI ground function app}
        |f_{\ground}(x_1) - f_{\ground}(x_0)| \leq |x_1 - x_0| c_2 \|O_A\| |A|^2 \sum_{r=r_0}^\infty r^D e^{-b_2g(r-1)/\log^2 g(r-1)}\ .
    \end{align}
\end{lemma}

\begin{proof}[*lem:GALI ground app]
    We begin by defining $H^{[r]}$ as in Eq.~\eqref{eq:partial x H decomposition}, and use the triangle inequality to bound
    \begin{align}\label{eq:spectral flow truncation triangle sum}
        \left\|[\Psi_H(O_A),\partial_x H] \right\| \leq \sum_{r = r_0}^\infty \left\|[\Psi_H(O_A) , \partial_x H^{[r]}] \right\|\ .
    \end{align}
    For each $r\geq r_0$, we decompose
    \begin{align}
        \Psi_H(O_A) = \Psi_H^{[r]}(O_A) + (\Psi_H(O_A) - \Psi_H^{[r]}(O_A))\ ,
    \end{align}
    where each $\Psi_H^{[r]}(O_A)$ is supported only within a radius $r$ of $A$, and according to \cref{lem:qbp and spectral flow truncation} we have
    \begin{align}\label{eq:spectral flow truncation bound thing}
        \|\Psi_H(O_A) - \Psi_H^{[r]}(O_A)\| \leq a_2 |A| \|O_A\| e^{-b_2 g(r)/\log^2 g(r)}\ .
    \end{align}
    Note that the operators $\Psi_H^{[r-1]}(O_A)$ and $\partial_s H^{[r]}$ commute as they have disjoint support, so we can bound
    \begin{align}
        \left\|[\Psi_H(O_A),\partial_x H^{[r]}]  \right\| &\leq 2 \|\Psi_H(O_A) -  \Psi_H^{[r-1]}(O_A)\|\cdot \|\partial_x H^{[r]}\| \\
        &\leq 2|A|\|O_A\|e^{-b_2g(r-1) / \log^2 g(r-1)} \cdot k_D r^D |A| \|\partial_x H\|_{F_g} \|F_g\| \\
        &= c_2 \|O_A\| |A|^2 r^De^{-b_2 g(r-1)/\log^2g(r-1)}\ ,
    \end{align}
    where in the second line we have used Eqs.~\eqref{eq:partial x H r bound} and \eqref{eq:spectral flow truncation bound thing}, and where
    \begin{align}
        c_2 := 2k_D \|\partial_x H\|_{F_g} \|F_g\|\ .
    \end{align}
    Using Eq.~\eqref{eq:spectral flow truncation triangle sum}, we arrive at the required bound:
    \begin{align}
        \left\|[\Psi_H(O_A) , \partial_x H] \right\| \leq c_2 \|O_A\| |A|^2 \sum_{r=r_0}^\infty r^D e^{-b_2 g(r-1) / \log^2 g(r-1)}\ .
    \end{align}
    Defining $f_{\ground}(x) := \bra{\psi_0(x)} O_A \ket{\psi_0(x)}$, we can apply Eq.~\eqref{eq:spectral flow observable derivative} to deduce that
    \begin{align}
        |f_{\ground}(x_1) - f_{\ground}(x_0)| &\leq \int_{x_0}^{x_1} \diff x \left|\frac{\diff}{\diff x} f_{\ground}(x) \right| \\
        &\leq \int_{x_0}^{x_1} \left\|[\Psi_{H}(O_A),\partial_x H] \right\| \\
        &\leq |x_1 - x_0| c_2 \|O_A\| |A|^2 \sum_{r=r_0}^\infty r^D e^{-b_2g(r-1)/\log^2 g(r-1)}\ ,
    \end{align}
    as required.
\end{proof}

\section{Bounding the local Schrieffer-Wolff expansion}\label{app:sw bounds}

\subsection{Preliminary results}

In this section, we prove a couple of lemmas controlling the growth of the $\|\cdot\|_F$-norm under commutator, which will ultimately be important to bound the terms of the perturbative expansion $V^{(q)}$ obtained via the local Schrieffer-Wolff transformation. 

\begin{lemma}\label{lem:Fnormcommutator}
    Let $F$ be a normalised $F$-function and let $X = \sum_{A\subseteq \Gamma}X_A$ and $Y = \sum_{A\subseteq \Gamma} Y_A$ be Hamiltonians on $\Gamma$. Then
    \begin{align}
        \|[X,Y]\|_F \leq 4 \|X\|_F \|Y\|_F\ .
    \end{align}
\end{lemma}

\begin{proof}[*lem:Fnormcommutator]
    Writing $[X,Y]$ as a sum of local terms, $[X,Y] = \sum_{A\subseteq \Gamma}[X,Y]_A$, by definition we have
    \begin{align}
        \|[X,Y]\|_F &\leq \sup_{i,j\in \Gamma} \frac{1}{F(\dist(i,j))} \sum_{\substack{A\subseteq \Gamma \\ i,j\in A}} \|[X,Y]_A\|\ .
    \end{align}
    But for fixed $i,j$, note that nonzero terms $[X,Y]_A$ can only arise from as commutators between terms $X_B$ and $Y_C$ with overlapping support (containing at least one $k \in B\cap C$), i.e.
    \begin{align}
        \sum_{\substack{A\subseteq \Gamma \\ i,j\in A}} \|[X,Y]_A\| &\leq \sum_{k\in \Gamma} \left( \sum_{\substack{B\ni i,k \\ C\ni j,k}} \left(\|[X_B,Y_C]\| + \|[X_C,Y_B]\|\right) \right)\\
        &\leq 2 \sum_{k\in \Gamma} \left( \left(\sum_{B\ni i,k} \|X_B\| \right) \left(\sum_{C\ni j,k} \|Y_C\| \right) + \left(\sum_{B\ni i,k} \|Y_B\| \right) \left(\sum_{C\ni j,k} \|X_C\| \right) \right)\ .
    \end{align}
    But we have $\sum_{A\ni i,k}\|X_A\| \leq F(\dist(i,k))\|X\|_F$, and similarly for the other terms, so the above can be bounded by
    \begin{align}
        \sum_{\substack{A\subseteq \Gamma \\ i,j\in A}} \|[X,Y]_A\| &\leq 4\|X\|_F\|Y\|_F \sum_{k\in \Gamma} F(\dist(i,k)) F(\dist(j,k))\ .
    \end{align}
    Thus
    \begin{align}
        \|[X,Y]\|_F &\leq 4\|X\|_F \|Y\|_F \sup_{i,j \in \Gamma} \sum_{k\in \Gamma} \frac{F(\dist(i,k)) F(\dist(j,k))}{F(\dist(i,j))} \\
        &\leq 4\|X\|_F \|Y\|_F\ ,
    \end{align}
    as required, where we have used the normalisation condition (the constant $C_F$ defined in Eq.~\eqref{eq:cf constant} is equal to 1).
\end{proof}

\begin{lemma}\label{lem:tq commutator}
    Let $W$ be a Hamiltonian, and $q\geq 1$. Then
    \begin{align}
        \|[T^{(q)},W]\|_F \leq 4 \Delta^{-1} \|V^{(q)}\|_F \|W\|_F \ .
    \end{align}
\end{lemma}

\begin{proof}[*lem:tq commutator]
    Using Lemma~\ref{lem:Fnormcommutator}, we have
    \begin{align}
        \|[T^{(q)},W]\|_F \leq 4 \|T^{(q)}\|_F \|W\|_F\ .
    \end{align}
    Moreover, using the definition Eq.~\eqref{eq:tq definition} with Lemma~\ref{lem:properties of la}\eqref{it:la bound}, we have
    \begin{align}
        \|T^{(q)}\|_F \leq \Delta^{-1} \|V^{(q)}\|_F\ ,
    \end{align}
    and the result follows.
\end{proof}

\subsection{Bounding series terms}

Below we restate and prove Lemma~\ref{lem:vq bound}, which places bounds on the $V^{(q)}$ constructed via the local Schrieffer-Wolff transformation. The proof is very similar to Ref.~\cite{bravyi2011schrieffer}, Lemma 4.2, but differs in the respects that it accounts for general polynomial perturbations in $x$, and that Lemmas~\ref{lem:Fnormcommutator} and \ref{lem:tq commutator} provide tighter bounds for geometrically local Hamiltonians which leads to a provably convergent series.

\begin{lemma}[Restatement of Lemma~\ref{lem:vq bound}]\label{lem:vq bound app}
    For $q\geq 1$, $V^{(q)}$ is bounded as
    \begin{align}
        \|V^{(q)}\|_F \leq \frac{\Delta \theta^q}{16}\ ,
    \end{align}
    where $\theta > 0$ is a constant depending on the ratio $J/\Delta$. In particular, if $x \leq 1/(2\theta)$, then the Hamiltonian $V(x) = \sum_{q\geq 1} x^q V^{(q)}$ is bounded as
    \begin{align}
        \|V_{\eff}(x)\|_F \leq \|V(x)\|_F \leq \frac{\Delta \theta x}{8}\ .
    \end{align}
\end{lemma}

\begin{proof}[*lem:vq bound app]
    By definition, we have $\|V^{(1)}\|_F = \|H^{(1)}\|_F \leq J$. Using Eq.~\eqref{eq:vq definition} and the triangle inequality, we have the following recursive bound for $q\geq 2$:
    \begin{align}
        \|V^{(q)}\|_F &\leq \|H^{(q)}\|_F + \sum_{r= 2}^q \frac{1}{r!}\sum_{\substack{q-1\geq q_1,\dots,q_r \geq 1 \\ q_1+\dots+q_r = q}} \|[T^{(q_1)},[T^{(q_2)},\dots[T^{(q_r)},H^{(0)}]\cdots]] \| \notag\\
        &\quad + \sum_{\alpha = 1}^q \sum_{r= 1}^{q-\alpha} \frac{1}{r!}\sum_{\substack{q-\alpha \geq q_1,\dots,q_r \geq 1 \\ q_1+\dots+q_r = q-\alpha}} \|[T^{(q_1)},[T^{(q_2)},\dots [T^{(q_r)},H^{(\alpha)}]\cdots ]] \|\ .
    \end{align}
    By assumption, $\|H^{(q)}\|_F \leq J$. Moreover, using Lemma~\ref{lem:tq commutator} $r$ times we can bound
    \begin{align}
        \|[T^{(q_1)},\dots[T^{(q_r)},H^{(\alpha)}]\dots] \| &\leq (4\Delta^{-1})^r \|V^{(q_1)}\|_F\dots\|V^{(q_r)}\|_F J \ ,\quad \alpha \geq 0\ .
    \end{align}
    Hence we have
    \begin{align}
        \|V^{(q)}\|_F &\leq J + \sum_{r= 2}^q \frac{1}{r!} \sum_{\substack{q-1\geq q_1,\dots,q_r \geq 1 \\ q_1+\dots+q_r  = q}} (4\Delta^{-1})^r \|V^{(q_1)}\|_F\dots \|V^{(q_r)}\|_F J \notag\\
        &\quad + \sum_{\alpha = 1}^q \sum_{r= 1}^{q-\alpha} \frac{1}{r!} \sum_{\substack{q-\alpha \geq q_1,\dots,q_r \geq 1 \\ q_1+\dots+q_r = q-\alpha}} (4\Delta^{-1})^r \|V^{(q_1)}\|_F \dots \|V^{(q_r)}\|_F J \\
        \Rightarrow \frac{\|V^{(q)}\|_F}{J} &\leq 1+ \sum_{r= 2}^q c^r \sum_{\substack{q-1\geq q_1,\dots,q_r \geq 1 \\q_1+\dots+q_r = q}} \prod_{i=1}^r \left(\frac{\|V^{(q_i)}\|_F}{J} \right) + \sum_{\alpha = 1}^q \sum_{r= 1}^{q-\alpha} c^r \sum_{\substack{q-\alpha \geq q_1,\dots,q_r \geq 1\\ q_1+\dots+q_r = q-\alpha}} \prod_{i=1}^r \left(\frac{\|V^{(q_i)}\|_F}{J} \right)\ .
    \end{align}
    where we have defined the constant $c = 4J/\Delta$ and generously lower bounded $r! \geq 1$. 
    
    We can use this expression to recursively define a sequence $\{\mu_q\}_{q=1}^\infty$ such that $\|V^{(q)}\|_F/J \leq \mu_q$ for all $q\geq 1$. We set $\mu_1 = 1$ and
    \begin{align}
        \mu_q &= 1 + \sum_{r= 2}^q c^r \sum_{\substack{q-1\geq q_1,\dots,q_r \geq 1 \\ q_1 + \dots + q_r = q}} \mu_{q_1}\dots \mu_{q_r} + \sum_{\alpha = 1}^{q-1} \sum_{r = 1}^{q-\alpha} c^r\sum_{\substack{q-\alpha \geq q_1,\dots,q_r\geq 1 \\ q_1+\dots+q_r = q-\alpha}} \mu_{q_1} \dots \mu_{q_r}\ ,\quad q\geq 2\ ,
    \end{align}
    which, after some algebra, can be rearranged to
    \begin{align}\label{eq:muq recursion}
        (1+c)\mu_q = 1 + \sum_{\alpha = 0}^{q-1} \sum_{r= 1}^{q-\alpha} c^r\sum_{\substack{q-\alpha \geq q_1,\dots,q_r \geq 1 \\ q_1+\dots+q_r = q-\alpha}} \mu_{q_1} \dots \mu_{q_r}\ ,\quad q\geq 2\ .
    \end{align}
    To make further progress, we define the function $\mu(z)$ as a formal power series
    \begin{align}
        \mu(z) := \sum_{q = 1}^\infty \mu_q z^q\ .
    \end{align}
    Notice that $\sum_{q\geq 1} z^q = z/(1-z)$ and 
    \begin{align}
        \sum_{q= 2}^\infty z^q \sum_{\alpha = 0}^{q-1} \sum_{r= 1}^{q-\alpha} c^r \sum_{\substack{q-\alpha \geq q_1,\dots,q_r \geq 1 \\ q_1+\dots+q_r = q-\alpha}} \mu_{q_1} \dots \mu_{q_r} &= \sum_{q= 0}^\infty z^q \sum_{\alpha = 0}^{q-1}\sum_{r= 1}^{q-\alpha} c^r \sum_{\substack{q-\alpha \geq q_1,\dots,q_r \geq 1 \\ q_1+\dots+q_r = q-\alpha}} \mu_{q_1} \dots \mu_{q_r} - cz \\
        &= \sum_{\alpha = 0}^\infty z^\alpha \sum_{r= 1}^\infty c^r \sum_{q_1,\dots,q_r = 1}^\infty z^{q_1+\dots+q_r} \mu_{q_1} \dots \mu_{q_r} - cz \\
        &= \frac{1}{1-z} \cdot \frac{c\mu(z)}{1-c\mu(z)} - cz\ .
    \end{align}
    Thus, the condition that $\mu_1 = 1$ and the recursion relation Eq.~\eqref{eq:muq recursion} are equivalent to
    \begin{align}
        (1+c)\mu(z) = \frac{z}{1-z} + \frac{1}{1-z} \cdot \frac{c\mu(z)}{1-c\mu(z)} \ ,
    \end{align}
    which can be rearranged to 
    \begin{align}
        c(1+c)\mu(z)^2 - \mu(z) + \frac{z}{1-z} = 0\ .
    \end{align}
    Solving this quadratic, and choosing the branch with $\mu(0) = 0$, we obtain
    \begin{align}
        \mu(z) = \frac{1}{2c(1+c)}\left( 1 - \sqrt{1-4c(1+c) \frac{z}{1-z}}\right)\ .
    \end{align}
    This function is analytic in the disc
    \begin{align}
        |z| < z_0 = \frac{1}{4c(1+c) + 1}\ ,
    \end{align}
    and moreover, defining $z_1 = 1/(8c(1+c) + 1)$, we have
    \begin{align}
        |\mu(z)| \leq \frac{1}{2c(1+c)} \left( 1 - \frac{1}{\sqrt{2}}\right) \leq \frac{1}{4c(1+c)}\ , \quad \text{for $|z|\leq z_1$}\ .
    \end{align}
    Hence, using Cauchy's integral formula, we can bound the coefficients $\mu_q$ as
    \begin{align}
        \mu_q &= \frac{1}{2\pi i} \oint_{|z| = z_1} \frac{\mu(z)}{z^{q+1}} \diff z \\
        &\leq z_1^{-q} \sup_{|z| = z_1} |\mu(z)| \\
        &\leq (8c(1+c)+1)^q\cdot\frac{1}{4c(c+1)} \\
        &\leq \frac{(8(1+c))^{2q}}{4c} \ .
    \end{align}
    Hence
    \begin{align}
        \|V^{(q)}\|_F &\leq J\mu_q \leq J\cdot \frac{(8(1+c))^{2q}}{4c} \\
        &= \frac{\Delta}{16} \left(8\left( 1 + \frac{4J}{\Delta}\right) \right)^{2q} \\
        &= \frac{\Delta \theta^q}{16}\ ,
    \end{align}
    where 
    \begin{align}\label{eq:gamma and theta definition}
        \theta = 64( 1+4 J/\Delta)^2\ .
    \end{align}
    It follows by the triangle inequality that
    \begin{align}
        \|V(x)\|_F &\leq \sum_{q\geq 1} x^q \|V^{(q)}\|_F \\
        & \leq \frac{\Delta}{16} \sum_{q\geq 1} (\theta x)^q \\
        &= \frac{\Delta\theta x}{16}\sum_{q\geq 0} (\theta x)^q \leq \frac{\Delta \theta x
        }{8}\ ,
    \end{align}
    where the last inequality holds if $x\leq 1/ (2\theta)$.
\end{proof}

\section{Properties of Hamiltonian gadgets}\label{app:gadget properties}

\subsection{Equivalent characterisation}\label{app:gadget equiv}

In this section, we show that Definition~\ref{def:gadget} is equivalent in the high-energy regime to an alternative general formulation for gadgets introduced in a previous work (Ref.~\cite{harley2024going}), restated below. The formulation of Ref.~\cite{harley2024going} is itself based on notions of simulation from prior works Refs.~\cite{bravyi2017complexity,cubitt2018universal}, so this result can be viewed as a consistency check with these also.

\begin{definition}[$(\Delta,\eta,\epsilon)$-gadget --- \cite{harley2024going}, Definition 18]\label{def:deltaetaepsgadget}
    Let $H'$ be a Hamiltonian acting on $\cH'=\bigotimes_{i\in\Gamma_{\eff} \cup \Gamma_{\anc}} \cH_{i}$, and let $H$ be a Hamiltonian on $\cH_{\eff} = \bigotimes_{i\in \Gamma_{\eff}} \cH_i$. For constants $\Delta,\eta,\epsilon\geq 0$, we say that $H'$ is a $(\Delta,\eta,\epsilon)$-gadget for $H$ if there exists a rank-1 projector $P\in \Proj(\cH_{\anc})$ and a unitary $U \in \U(\cH')$ such that the projector onto subspace of states below energy $\Delta$ (with respect to $H'$), $P'\in \Proj(\cH')$, can be written as $P' = U(\id\otimes P) U^\dagger$, and
    \begin{align}
        \|U - \id\| \leq \eta\ ,\quad \|P' H' P' - U(H\otimes P)U^\dagger \| \leq \epsilon \ .
    \end{align}
\end{definition}

This definition defines gadgets entirely in terms of their behaviour in a low-energy subspace, and in particular only has meaning when $H'$ applies energy penalties to the ancillary sites of the order of $\|H\|$. Definition~\ref{def:gadget} only requires $H'$ to be locally gapped, however in the limit $x\rightarrow 0$ these definitions coincide as we formalise in the following lemma:

\begin{lemma}[Equivalent characterisation of gadgets]\label{lem:equivalent gadgets}
    Let $\{H'(x)\}_{x\in(0,\infty)}$ be a family of Hamiltonians on $\cH' = \otimes_{i\in\Gamma_{\anc}\cup\Gamma_{\eff}} \cH_i$ which is a degree-$d$ polynomial in $x^{-1}$ written as $H'(x) = x^{-d} \sum_{\alpha = 0}^d x^\alpha H^{(\alpha)}$. As in Definition~\ref{def:gadget}, we assume that $H^{(0)} = \Delta' \sum_{i\in \Gamma_{\anc}} (\id - \proj{0_i})$, where $\Delta' > 0$ and $\ket{0_i}$ is a local state on site $i$. Assume $\|H^{(\alpha)}\|_F\leq J$ for $1\leq \alpha \leq d$, where $F$ is a normalised $F$-function on $\Gamma$. The following are equivalent:
    \begin{enumerate}[(I)]
        \item For sufficiently small $x$, $H'(x)$ is a $(\Delta(x),\eta(x),\epsilon(x)$)-gadget for $H$, for some functions $\Delta,\eta,\epsilon : (0,\infty) \rightarrow [0,\infty)$ such that
        \begin{align}
            \lim_{x\rightarrow 0} \Delta(x) = \infty\ ,\quad \lim_{x\rightarrow 0} \eta(x) = 0\ ,\quad \lim_{x\rightarrow 0} \epsilon(x) = 0\ .
        \end{align}
        \item $H'(x)$ is a gadget of degree $d$ for $H$.
    \end{enumerate}
\end{lemma}

This result essentially follows by noticing that, for sufficiently small $x$, the projection obtained through Schrieffer-Wolff perturbation theory will project onto the low-energy subspace of $H'(x)$. In this regime, Definitions~\ref{def:gadget} and \ref{def:deltaetaepsgadget} are in correspondence.

\begin{proof}[*lem:equivalent gadgets]
    The local Schrieffer-Wolff transformation produces series $\{V^{(q)}\}_{q\geq 1}$ and $\{T^{(q)}\}_{q\geq 1}$ such that
    \begin{align}
        e^{T(x)} H'(x) e^{-T(x)} = x^{-d} H^{(0)} + x^{-d} \sum_{q\geq 1} x^q \left( V^{(q)} + [T^{(q)}, H^{(0)}]\right)\ ,
    \end{align}
    where the sum is convergent for sufficiently small $x$ (as given explicitly by Lemma~\ref{lem:vq bound}), and all summands are block-diagonal with respect to the projector $P_0 = \otimes_{i\in \Gamma_{\anc}} \proj{0_i}$. In particular, all states in the image of $P_0$ have energy at most $\bigO(x^{1-d})$, whilst all states in the image of $(1-P_0)$ have energy at least $\Omega(x^{-d})$ due to the penalty imposed by the $H^{(0)}$ term. It follows that, for sufficiently small $x$, the projector $P' = e^{-T(x)} P_0 e^{T(x)}$ projects onto the low-energy space of $H'(x)$ with cutoff $\Delta = \Theta(\Delta' x^{-d})$.

    Suppose that $H'(x)$ is a gadget of degree $d$ for $H$. This implies that the part of $e^{T(x)} H'(x) e^{-T(x)}$ lying in the $P_0$ space is a polynomial in $x$ with constant term given by $H\otimes P_0$, and in particular we can write
    \begin{align}
        \left\| P' H'(x) P' - e^{-T(x)} (H \otimes P_0 )e^{T(x)} \right\| \leq \bigO(x)\ .
    \end{align}
    Hence, using that $\|T(x)\| = \bigO(x)$ and identifying $U = e^{-T(x)}$, we immediately recover Definition~\ref{def:deltaetaepsgadget} as required.

    Now suppose that $H'(x)$ is a $(\Delta(x),\eta(x),\epsilon(x))$-gadget for $H$, where $\Delta \rightarrow \infty$, $\eta,\epsilon \rightarrow 0$ as $x\rightarrow 0$. This guarantees that, for small $x$,
    \begin{align}
        \|P' H'(x) P'\| \leq \|H\| + \bigO(x)\ ,
    \end{align}
    which in particular implies that all negative powers of $x$ must disappear in the $P_0$ space of $e^{T(x)} H'(x) e^{-T(x)}$. This, along with the observation that $P' H'(x) P' \rightarrow H\otimes P_0$ as $x\rightarrow 0$ (using that $\epsilon,\eta \rightarrow 0$), ensures that the constant term in the expansion is $H$, recovering Definition~\ref{def:gadget}.
\end{proof}

\subsection{Combination of gadgets}\label{app:gadget comb}

In this section, we prove some results concerning the combination of gadgets in parallel. Firstly, in the below lemma, we observe that a gadget can trivially be combined with a background Hamiltonian independent of the parameter $x$.
\begin{lemma}\label{lem:gadget plus const}
    Suppose that $H'(x)$ is a gadget of degree $d$ for $H$. Then, for any local Hamiltonian $\tilde{H}$ acting on $\cH_{\eff}$, $H'(x) + \tilde{H}$ is a gadget of degree $d$ for $H + \tilde{H}$.
\end{lemma}

\begin{proof}[*lem:gadget plus const]
    This can be seen immediately from the construction of the perturbative series $\{V^{(q)}\}_{q\geq 1}$ as in Eq.~\eqref{eq:vq definition}. Notice that the constant term $\tilde{H}$ may be absorbed into $H^{(d)}$, which only appears at $d$th order in the perturbative expansion.
\end{proof}

It turns out that gadgets of degree $d\leq 3$ can immediately be used in parallel as guaranteed by the below theorem. Intuitively, this is because the interactions are not strong enough for cross-gadget contributions in the effective Hamiltonian, since this requires terms which excite and de-excite at least two different gadgets --- such terms only appear at fourth order and above. Parallel combination results for gadgets have been proved in many specialised settings to prove results in Hamiltonian complexity theory, see e.g. Refs.~\cite{bravyi2008quantum,cubitt2018universal}. The below theorem can be viewed as a somewhat cleaner version of the previous general result of Ref.~\cite{harley2024going}, Proposition 24.

\begin{theorem}[Gadget combination]\label{thm:gadget comb}
    Let $\{h'_j(x)\}_{j\in \Gamma_{\anc}}$ be such that, for every $j\in \Gamma_{\anc}$, $h'_j(x)$ is a family of Hamiltonians acting on $\cH_{\eff}\otimes \cH_j$ which is a gadget of degree $d\leq 3$ for $h_j$ on $\cH_{\eff}$. Then the Hamiltonian
    \begin{align}
        H'(x) := \sum_{j\in \Gamma_{\anc}} h'_j(x)\ ,
    \end{align}
    which acts on $\cH' = \otimes_{x\in \Gamma_{\eff} \cup \Gamma_{\anc}}\cH_x$, is a gadget of degree $d$ for the combined Hamiltonian
    \begin{align}
        H(x) := \sum_{j\in \Gamma_{\anc}} h_j\ ,
    \end{align}
    on $\cH_{\eff}$.
\end{theorem}

Since a $H'(x)$ constructed in this way is typically extensive, it is more natural to refer to this as a simulator Hamiltonian rather than a gadget as mentioned above. The proof of Theorem~\ref{thm:gadget comb} follows by explicitly computing the first three terms $V^{(1)},V^{(2)},V^{(3)}$ of the local Schrieffer-Wolff expansion. We will require the following identity:

\begin{lemma}\label{lem:commutator thing}
    Let $A,B,C\subseteq \Gamma$ be disjoint sets of sites. Assume that $X_{AB}$ and $Y_{BC}$ are Hermitian operators supported only on the sites $A\cup B$ and $B\cup C$ respectively. Then
    \begin{align}\label{eq:lcommutatorthing}
        P_{ABC} \left[\cL_{AB}(X_{AB}),Y_{BC} \right] P_{ABC} = P_{ABC} \left[ \cL_{BC}(Y_{BC}) , X_{AB} \right] P_{ABC}\ ,
    \end{align}
    where the projector $P_{ABC}$ and the superoperators $\cL_{AB},\cL_{BC}$ are defined as in Eqs.~\eqref{eq:pa definition} and \eqref{eq:la definition} respectively.
\end{lemma}

\begin{proof}[*lem:commutator thing]
    Using \eqref{eq:la definition}, we can write the left-hand side of the above expression as
    \begin{align}
        & \int_0^\infty \diff t P_{ABC} \big[ e^{-t H^{(0)}|_{AB}/\Delta}Q_{AB} X_{AB} P_{AB} - P_{AB} X_{AB} Q_{AB} e^{-tH^{(0)}|_{AB}/\Delta}, Y_{BC} \big] P_{ABC} \\
        &\quad = -\int_0^\infty \diff t \bigg( P_{ABC} X_{AB} Q_{AB} e^{-tH^{(0)}|_{AB}/ \Delta} Y_{BC} P_{ABC}  + P_{ABC} Y_{BC} e^{-tH^{(0)}|_{AB}/\Delta} Q_{AB} X_{AB} P_{ABC} \bigg) \\
        &\quad = -\int_0^\infty \diff t \bigg( P_{ABC} X_{AB} Q_B e^{-tH^{(0)}|_{B}/\Delta} Y_{BC} P_{ABC} + P_{ABC} Y_{BC} e^{-tH^{(0)}|_B/\Delta} Q_B X_{AB} P_{ABC}\bigg)\ .
    \end{align}
    In the last line, we have used that $H^{(0)}$ is 1-local to decompose $e^{-tH^{(0)}|_{AB}/\Delta} = e^{-tH^{(0)}|_A/\Delta} e^{-tH^{(0)}|_B/\Delta}$. In both terms, the part acting on $A$ is projected into the $P_A$ space on which it acts as the identity. Additionally, we have used the fact that $P_A Q_{AB} = P_AQ_B$. An identical calculation on the right-hand side of \eqref{eq:lcommutatorthing} arrives at the same expression.
\end{proof}

\begin{proof}[*thm:gadget comb]
    Writing each gadget $h_j'(x)$ as a sum of terms $h_j'(x) = x^{-d} \sum_{\alpha = 0}^d x^{\alpha} h_j^{(\alpha)}$, we can write $H'(x)$ as
    \begin{align}
        H^\prime(x) = x^{-d}\sum_{\alpha = 0}^d x^\alpha H^{(\alpha)}\ ,\quad\text{where}\quad H^{(\alpha)} := \sum_{j\in [N]} h_j^{(\alpha)}\ .
    \end{align}
    We will explicitly compute the first three terms $V^{(q)}$ of the local Schrieffer-Wolff transformation for this Hamiltonian, and show that their restriction to the $(\id\otimes P_0)$ space give to the correct target Hamiltonian.

    For a local Hamiltonian $X = \sum_{A \subseteq \Gamma} X_A$, we will adopt the shorthand
    \begin{align}
        \mathcal{L}(X) := \sum_{A\subseteq \Gamma} \mathcal{L}_A(X_A)\, \quad \offdiag_A(X) := \sum_{A\subseteq\Gamma} \offdiag_A(X_A)\ .
    \end{align}
    Each of the Hamiltonians $V^{(1)}$, $V^{(2)}$, and $V^{(3)}$ will be a polynomial in the $h_j^{(\alpha)}$. Note that --- for the restriction of $V^{(q)}$ to the $(\id\otimes P_0)$ subspace --- we can ignore all of these polynomial terms which only contain one distinct $j$-index, because these terms will cancel due to our assumption that each of the $h_j^\prime(x)$ are gadgets.

    For the first-order term, we have
    \begin{align}
        V^{(1)} = H^{(1)} = \sum_j h_j^{(1)}\ .
    \end{align}
    By assumption, $(\id\otimes P_0) h_j^{(1)}(\id\otimes P_0) = 0$ for all $j$, so this term will vanish. From \eqref{eq:tq definition} we have
    \begin{align}
        T^{(1)} = \mathcal{L}(V^{(1)}) = \mathcal{L}(H^{(1)})\ .
    \end{align}

    For the second-order term, from \eqref{eq:vq definition} we have
    \begin{align}
        V^{(2)} &= H^{(2)} + \big[ T^{(1)}, V^{(1)} - \frac{1}{2}\offdiag(V^{(1)})\big] \\
        &= H^{(2)} + \big[\mathcal{L}(H^{(1)}), H^{(1)} - \frac{1}{2}\offdiag(H^{(1)}) \big] \\
        &= \sum_j\big( h_j^{(2)} + \big[\mathcal{L}(h_j^{(2)}),h_j^{(1)} - \frac{1}{2}\offdiag(h_j^{(1)})\big] + \sum_{i\neq j} \big[ \mathcal{L}(h_i^{(2)}) , h_j^{(1)} - \frac{1}{2}\offdiag(h_j^{(1)})\big]\ .
    \end{align}
    As mentioned, the first term can be ignored; by the assumption that each $h_j^\prime(x)$ is a gadget, its contribution to the $P_0$ space will be zero. For the second term, notice that $\mathcal{L}(h_i^{(2)})$ is off-diagonal with respect to $P_{0,i} := \proj{0_i}$ by construction. However, by the fact that $[h_i^{(0)},h_j^{(1)}] = 0$ (since $h_i^{(0)}$ is supported only on $\cH_i$, on which $h_j^{(1)}$ does not act), we are guaranteed that the second term in the commutator is diagonal with respect to $P_{0,i}$. Hence the entire commutator is off-diagonal with respect to $P_{0,i}$, hence its contribution to the $P_0$ space will also be zero.

    For notational convenience, we write
    \begin{align}
        G := H^{(1)} - \frac{1}{2} \offdiag(H^{(1)})=\sum_j g_j\ ,\quad g_j := h_j^{(1)} - \frac{1}{2} \offdiag(h_j^{(1)})\ .
    \end{align}
    Then from \eqref{eq:tq definition} we have
    \begin{align}
        V^{(2)} = H^{(2)} + [\mathcal{L}(H^{(1)}),G], \quad T^{(2)} = \mathcal{L}(H^{(2)}) + \mathcal{L}\big([\mathcal{L}(H^{(1)}),G]\big)\ .
    \end{align}

    For the third-order term, from \eqref{eq:vq definition} we then have
    \begin{subequations}
    \begin{align}
        V^{(3)} &= H^{(3)} + \frac{1}{2}[T^{(1)} , -\offdiag(V^{(2)})] + \frac{1}{2}[T^{(2)},-\offdiag(V^{(1)})] + \frac{1}{6} \big[T^{(1)},[T^{(1)},-\offdiag(V^{(1)})]\big] \notag \\
        &\quad + [T^{(2)},H^{(1)}] + \frac{1}{2}\big[ T^{(1)},[T^{(1)},H^{(1)}]\big] + [T^{(1)},H^{(2)}] \notag \\
        &= H^{(3)} + [\mathcal{L}(H^{(1)}),H^{(2)} - \frac{1}{2}\offdiag(H^{(2)})] + [\mathcal{L}(H^{(2)}),G] \label{eq:v3a}\\
        &\quad +\big[\mathcal{L}([\mathcal{L}(H^{(1)}),G]),G\big] - \frac{1}{2}\big[ \mathcal{L}(H^{(1)}),\offdiag([\mathcal{L}(H^{(1)}),G])\big] \label{eq:v3b}\\
        &\quad + \frac{1}{2} \big[\mathcal{L}(H^{(1)}),[\mathcal{L}(H^{(1)}),H^{(1)}] \big] \label{eq:v3c}\\
        &\quad - \frac{1}{6}\big[ \mathcal{L}(H^{(1)}),[\mathcal{L}(H^{(1)}),\offdiag(H^{(1)})] \big]\label{eq:v3d}
    \end{align}
    \end{subequations}
    We will now deal with the lines \eqref{eq:v3a}-\eqref{eq:v3d} separately.
    \begin{itemize}
        \item \eqref{eq:v3a}:

        As mentioned above, we only need to worry about the terms consisting of contributions from two distinct gadgets; that is:
        \begin{align}
            \sum_{i\neq j} \big( [\mathcal{L}(h_i^{(1)}),h_j^{(2)} - \frac{1}{2}\offdiag(h_j^{(2)})] + [\mathcal{L}(h_i^{(2)}),g_j]\big)\ .
        \end{align}
        For each term in this sum, notice that $\mathcal{L}(h_i^{(1)})$ and $\mathcal{L}(h_i^{(2)})$ are block-off-diagonal with respect to $P_{0,i}$ by definition, whilst $h_j^{(2)} - \frac{1}{2}\offdiag(h_j^{(2)})$ and $g_j$ commute with $h_i^{(0)}$, and are hence block-diagonal with respect to $P_{0,i}$. So the entire term is block-off-diagonal and has no contribution to the $P_0$ space.

        \item \eqref{eq:v3b}:

        Notice that $\mathcal{L}$ only depends on the local off-diagonal parts of all its input terms (in other words, $\mathcal{L} = \mathcal{L}\circ \offdiag$), and in particular $\mathcal{L}(H^{(1)}) = \frac{1}{2} \mathcal{L}(G)$. Hence we can write the terms in \eqref{eq:v3b} as
        \begin{align}
            2\big[ \mathcal{L}(\offdiag([\mathcal{L}(G),G])) ,G\big] - 2\big[\mathcal{L}(G), \offdiag([\mathcal{L}(G),G]) \big]\ .
        \end{align}
        However, separating $G$ into its local terms and applying \cref{lem:commutator thing}, we see that these terms cancel in the $P_0$ space, giving zero contribution overall.

        \item \eqref{eq:v3c}:

        Separating into individual gadgets, we are concerned with terms of the form
        \begin{align}
            \big[ \mathcal{L}(h_i^{(1)}),[\mathcal{L}(h_j^{(1)}),h_k^{(1)}] \big]\ ,
        \end{align}
        where $i,j,k$ are not all equal. Note that $\mathcal{L}(h_i^{(1)})$ and $\mathcal{L}(h_j^{(1)})$ are block off-diagonal with respect to $P_{0,i}$ and $P_{0,j}$ respectively. Moreover, due to the condition that the first-order term vanishes in the $P_0$ space, we know that $(\id \otimes P_{0,k}) h_k^{(1)} (\id \otimes P_{0,k}) = 0$. This means that, unless $i=k$ or $j=k$, the overall term will have no contribution to the $P_{0,k}$ space.  If $i = k$, then $i \neq j$ by assumption, and the overall term is off-diagonal in the $P_{0,j}$ space. Likewise if $j=k$, then the overall term is off-diagonal in the $P_{0,i}$ space. Hence in all cases except $i=j=k$, this term provides no contribution to the $P_0$ space.

        \item \eqref{eq:v3d}:

        For this term the argument is similar to the previous; this time we are interested in contributions of the form
        \begin{align}
            \big[ \mathcal{L}(h_i^{(1)}),[\mathcal{L}(h_j^{(1)}),\mathcal{L}(h_k^{(1)})]\big]\ ,
        \end{align}
        where $i,j,k$ are not all equal. Here, $\mathcal{L}(h_i^{(1)})$, $\mathcal{L}(h_j^{(1)})$, and $\offdiag(h_k^{(1)})$ are off-diagonal with respect to $P_{0,i}$, $P_{0,j}$ and $P_{0,k}$ respectively. At least one of $i,j,k$ is distinct from the other two, and in that case the overall term will be block off-diagonal with respect to the corresponding projector. Hence the overall term provides no contribution to the $P_0$ space unless $i=j=k$.
    \end{itemize}
    In conclusion, we have shown that the only contributions to $(\id \otimes P_0) V^{(3)} (\id \otimes P_0)$ are those which involve no cross-gadget terms. From the assumption that each $h_j^\prime(x)$ is a gadget for $h_j$, we are therefore guaranteed that
    \begin{align}
        (\id \otimes P_0) V^{(3)}(\id\otimes P_0) = \bigg( \sum_{j \in [n]} h_j\bigg) \otimes P_0\ .
    \end{align}
\end{proof}

\subsection{Locality of gadgets}\label{app:gadget loc}

Finally, we prove the following simple result, establishing that a strictly $k'$-local gadget Hamiltonian can only yield a simulated Hamiltonian of locality $\bigO(d k')$ at order $d$ in its perturbative expansion.

\begin{lemma}[Gadget locality]\label{lem:gadget loc}
    Let $H'(x)$ be a gadet of degree $d$ for $H$. Suppose $H$ is $k$-local, and $H'$ is $k'$-local. Then
    \begin{align}
        d \geq \frac{k-1}{k'-1}\ .
    \end{align}
\end{lemma}

\begin{proof}[*lem:gadget loc]
    We will show by induction that $V^{(q)}$ is $[q(k'-1)+1]$-local. Note that, by the construction of $T^{(q)}$ in Eq.~\eqref{eq:tq definition} and using Lemma~\ref{lem:properties of la}\eqref{it:la locality}, this implies that $T^{(q)}$ is also $[q(k'-1)+1]$-local.

    For $q=1$, we have $V^{(1)} = H^{(1)}$ which is indeed $k'$-local by assumption. For $q\geq 2$, the result follows by inspection of Eq.~\eqref{eq:vq definition}, using the inductive hypothesis for $T^{(q')}$, $q'<q$: each commutator of the form $[T^{(q_1)},\dots[T^{(q_r)},H^{(0)}]\dots ]$ is at most $l$-local, where
    \begin{align}
        l = \sum_{j=1}^r \left( q_j(k'-1) + 1\right) - r + 1 = q(k'-1) + 1\ ,
    \end{align}
    using that $H^{(0)}$ is $1$-local and $\sum_j q_j = q$. Moreover, for $\alpha \geq 1$ each commutator of the form $[T^{(q_1)},\dots [T^{(q_r)},H^{(\alpha)}]\dots]$ is at most $l'$-local, where
    \begin{align}
        l' = \sum_{j=1}^r \left( q_j(k' - 1) + 1 \right) - r + k' \leq q(k' - 1) + 1\ ,
    \end{align}
    using that $H^{(\alpha)}$ is at most $k'$-local and $\sum_j q_j \leq q-1$ for terms of this form. With this in hand, we see that $V^{(d)}$ --- and hence $H$ --- is $[d(k'-1)+1]$-local. But then by assumption we have
    \begin{align}
        k\leq d(k'-1) + 1\ ,
    \end{align}
    from which the result follows.
\end{proof}

\subsection{Quasi-locality of observables on the effective space}

Below we prove Lemma~\ref{lem:oaeff truncation} which, informally, ensures that local measurements on the simulator system correspond to measurement of a quasi-local observable on the effective system $O_{A,\eff}(x)$ which can be truncated at finite radius up to small error.

\begin{lemma}[Restatement of Lemma~\ref{lem:oaeff truncation}]\label{lem:oaeff truncation app}
    Let $O_{A,\eff}(x)$ be the observable on $\Gamma_{\eff}$ defined in Eq.~\eqref{eq:oaeffdef}, and assume $x \leq 1/(2\theta)$. Then, for any $r\geq 0$, there exists an observable $O_{A,\eff}^{[r]}(x)$ on $\Gamma_{\eff}$, such that $O_{A,\eff}^{[r]}(x)$ is supported on $B_r(A)$, we have
    \begin{align}
        \left\|O_{A,\eff}(x) - O_{A,\eff}^{[r]}(x) \right\| &\leq a e^{-bg(r)}\ ,
    \end{align}
    for some constants $a,b > 0$. Moreover, $O_{A,\eff}^{[r]}(x)$ extends to a complex function $O_{A,\eff}^{[r]}(z)$ which is analytic on the disc $|z| \leq 1/(2\theta)$, and which is bounded as
    \begin{align}\label{eq:oaeff truncated bound app}
        \sup_{|z| \leq 1/(2\theta)} \|O_{A,\eff}^{[r]}(z)\| &\leq a' e^{b' r^D}\ ,
    \end{align}
    for constants $a',b' > 0$.
\end{lemma}

\begin{proof}[*lem:oaeff truncation app]
    We define $S(x)$ to be the Hamiltonian
    \begin{align}
        S(x) := \frac{T(x)}{ix}\ .
    \end{align}
    By the construction of $T(x)$, we are guaranteed that $S(x)$ is an analytic function of $x$, and by Theorem~\ref{thm:simulation properties}\eqref{it:heff bound} we know that
    \begin{align}
        \|S(x)\|_{F_g} \leq \theta /8 \ .
    \end{align}
    Writing $S(x)$ as a sum of local terms $S(x) = \sum_{A'\subseteq \Gamma} S_{A'}(x)$, we define $S^{[r]}(x) := \sum_{A'\subseteq B_r(A)}S_{A'}(x)$ to be the restriction of $S(x)$ which only contains terms supported within a radius $r$ of $A$. For $0\leq t \leq x$, we write
    \begin{align}
        O_A(x;t) &:= e^{itS(x)} O_A e^{-itS(x)}\ , \\
        O_A^{[r]}(x;t) &:= e^{itS^{[r]}(x)} O_A e^{-itS^{[r]}(x)}\ .
    \end{align}
    We will now use the locality of $S(x)$, with Lieb-Robinson bounds, to show that $O_A^{[r]}(x;t)$ is a good approximation for $O_A(x;t)$ up to an error which shrinks exponentially with $r$. Define $\epsilon(x;t) := O_{A}(x;t) - O_A^{[r]}(x;t)$. This satisfies $\epsilon(x;0) = 0$, and
    \begin{align}
        \frac{\partial}{\partial t} \epsilon(x;t) &= i[S(x),O_A(x;t)] - i[S^{[r]}(x),O_A^{[r]}(x;t)] \\
        &= i[S(x) - S^{[r]}(x),O_{A}(x;t)] + i[S^{[r]}(x),\epsilon(x;t)]\ .
    \end{align}
    This differential equation can be solved to give
    \begin{align}
        \epsilon(x;t) = ie^{itS^{[r]}(x)} \left( 
        \int_0^t \diff \tau e^{-i\tau S^{[r]}(x)} [S(x) - S^{[r]}(x),O_A(x;\tau)] e^{i\tau S^{[r]}(x)}\right) e^{-itS^{[r]}(x)}\ .
    \end{align}
    Using the triangle inequality we can thus bound
    \begin{align}
        \|\epsilon(x;t)\| \leq \int_0^t \diff \tau \left\|[S(x) - S^{[r]}(x),O_A(x;\tau)] \right\|\ .
    \end{align}
    To bound the integrand, we rewrite $S(x) - S^{[r]}(x)$ into a sum of local terms which do not act exclusively on $A$ and apply the triangle inequality, as follows:
    \begin{align}
        \left\| [S(x) - S^{[r]}(x) , O_A(x;\tau)] \right\| &\leq \sum_{\substack{A' \subseteq \Gamma \\ A' \nsubseteq B_r(A)}} \left\|[ S_{A'}(x) , O_A(x;\tau)]\right\| \\
        &\leq \sum_{A'\subseteq \Gamma \setminus B_{r/2}(A)} \left\|[ S_{A'}(x) , O_A(x;\tau)]\right\| + \sum_{\substack{A' \subseteq \Gamma \\ A' \cap B_{r/2}(A) \neq \emptyset \\ A' \cap (\Gamma \setminus B_r(A)) \neq \emptyset}} \left\|[ S_{A'}(x) , O_A(x;\tau)]\right\|\ .\label{eq:two sums to bound}
    \end{align}
    In the second line, we have split the sum into two parts: one containing the terms in $S(x)$ whose support lies a distance at least $r/2$ from $A$; and one containing the terms in $S(x)$ whose support contains at least one element within distance $r/2$ of $A$, and at least one element separated by distance $r$ from $A$. The former will be small due to Lieb-Robinson bounds localising the support of $O_A(x;\tau)$, whilst the latter will be small due to the locality of $S(x)$. In particular, we have:
    \begin{align}
        \sum_{A' \subseteq \Gamma \setminus B_{r/2}(A)}\left\|[ S_{A'}(x) , O_A(x;\tau)]\right\| &\leq \sum_{r' > r/2} \sum_{\substack{A'\subseteq \Gamma \\ \dist(A',A) = r'}} \left\|[ S_{A'}(x) , O_A(x;\tau)]\right\| \\
        &\leq \sum_{r' > r/2} \sum_{\substack{A'\subseteq \Gamma \\ \dist(A',A) = r'}} c \|S_A'(x)\| \|O_A\| |A| \left( e^{\nu \|S(x)\|_{F_g} \tau} - 1\right) e^{-g(r')} \\
        &= c\|O_A\| |A| \left( e^{\nu \|S(x)\|_{F_g} \tau} - 1 \right) \sum_{r' > r/2} e^{-g(r')} \sum_{\substack{A' \subseteq \Gamma \\ \dist(A',A) = r'}} \|S_{A'}(x)\|\ ,
    \end{align}
    where in the second line we have applied Lemma~\ref{lem:lr bounds}, introducing the constants $c,\nu > 0$. The sums can then be bounded by
    \begin{align}
        \sum_{r' > r/2} e^{-g(r')} \sum_{\substack{A'\subseteq \Gamma \\ \dist(A,A') = r'}} \|S_{A'}(x)\| &\leq \sum_{r' > r/2} e^{-g(r')} \sum_{i \in B_{r'}(A)} \sum_{\substack{A'\subseteq \Gamma \\ i \in A'}} \|S_{A'}(x)\| \\
        &\leq \sum_{r' > r/2} e^{-g(r')} |B_{r'}(A)| \|F_g\| \|S(x)\|_{F_g} \\
        &\leq k_D \theta \gamma \Delta^{-1} \|F_g\| \sum_{r' > r/2} (r')^D e^{-g(r')}\ ,
    \end{align}
    where in the last line we used the bounds $\|S(x)\|_{F_g} \leq \theta/8$ (from Corollary~\ref{cor:tq bound}) and $|B_{r'}(A)| \leq k_D (r')^D$ (from Eq.~\eqref{eq:ball volume bound}). Since we assume exponentially decaying interactions (i.e. $g$ is linear), the right-hand side of this expression is exponentially decaying in $r$.

    Meanwhile, the second term in Eq.~\eqref{eq:two sums to bound} can be bounded by
    \begin{align}
        \sum_{\substack{A' \subseteq \Gamma \\ A' \cap B_{r/2}(A) \neq \emptyset \\ A' \cap (\Gamma \setminus B_r(A)) \neq \emptyset}} \left\|[ S_{A'}(x) , O_A(x;\tau)]\right\| &\leq 2\|O_A\| \sum_{\substack{A' \subseteq \Gamma \\ A' \cap B_{r/2}(A) \neq \emptyset \\ A' \cap (\Gamma \setminus B_r(A)) \neq \emptyset}} \|S_{A'}(x)\| \\
        &\leq 2\|O_A\| \sum_{i \in B_{r/2}(A)} \sum_{j \in \Gamma \setminus B_r(A)} \sum_{\substack{A' \subseteq \Gamma \\ i,j\in A'}} \|S_{A'}(x)\| \\ 
        &\leq 2\|O_A\| \sum_{i \in B_{r/2}(A)} \sum_{j \in \Gamma \setminus B_r(A)} F_g(\dist(i,j)) \|S(x)\|_{F_g} \\
        &\leq \frac{\|O_A\| \theta}{4} \sum_{i\in B_{r/2}(A)} \sum_{r' > r} \sum_{\substack{j\in \Gamma \\ \dist(A,j) = r'}} F_g(r') \\
        &\leq \frac{\|O_A\| \theta F(0)}{4} \sum_{i\in B_{r/2}(A)} \sum_{r' > r} |B_{r'}(A)| e^{-g(r')} \\
        &\leq \frac{\|O_A\| \theta F(0)|B_{r/2}(A)|^2 }{4} \sum_{r' > r} e^{-g(r')} \\
        &\leq \frac{\|O_A\| \theta F(0)|A|^2 k_D^2 (r/2)^D }{4} \sum_{r'>r} (r')^D e^{-g(r')}\ .
    \end{align}
    Once again, using the assumption that $g$ is linear, we see that the right-hand side of this expression decays exponentially with $r$. We can therefore conclude that the bound Eq.~\eqref{eq:two sums to bound} decays exponentially in $r$, and thus
    \begin{align}
        \left\| e^{T(x)} O_A e^{-T(x)} - O_A^{[r]}(x;x) \right\| &= \|\epsilon(x;x)\| \\
        &\leq x\sup_{\tau \in [0,x]} \left\|[S(x) - S^{[r]}(x),O_A(x;\tau)]\right\| \\
        &\leq a e^{-bg(r)}\ ,
    \end{align}
    for some constants $a,b > 0$. Here we have used that $x \leq 1/(2\theta)$ is upper bounded by a constant, and treat $\|O_A\|$ and $|A|$ as constants.

    We define the observable $O_{A,\eff}^{[r]}(x)$ on $\Gamma_{\eff}$ by
    \begin{align}
        O_{A,\eff}^{[r]}(x) := (\id_{\eff} \otimes \bra{\mathbf{0}_{\anc}} ) e^{ixS^{[r]}(x)} O_A e^{-ix S^{[r]}(x)} (\id_{\eff} \otimes \ket{\mathbf{0}_{\anc}})\ .
    \end{align}
    By definition, $O_{A,\eff}^{[r]}(x)$ has support contained in $B_r(A)$, and we have established that
    \begin{align}
        \left\|O_{A,\eff}(x) - O_{A,\eff}^{[r]}(x) \right\| &\leq a e^{-bg(r)}\ ,
    \end{align}
    for $0\leq x \leq 1/(2\theta)$. Furthermore, since $S^{[r]}(x)$ is constructed as a power series in $x$, we are guaranteed that $O_{A,\eff}^{[r]}(x)$ extends to an analytic function for $|z| \leq 1/(2\theta)$, bounded by
    \begin{align}
        \| O^{[r]}_{A,\eff}(z)\| &\leq \|O_A\| \exp\left(2|z| \|S^{[r]}(z)\| \right)\ ,
    \end{align}
    where we have
    \begin{align}
        \|S^{[r]}(z)\| &\leq \sum_{A'\subseteq B_r(A)} \|S_{A'}(z)\| \\
        &\leq \sum_{i\in B_r(A)} \sum_{\substack{A' \subseteq \Gamma \\ i \in A'}} \|S_{A'}(z)\| \\
        &\leq |B_r(A)| \|S(z)\|_{F_g} \|F_g\| \\
        &\leq \frac{k_D r^D \theta \|F_g\|}{8}\ ,
    \end{align}
    so
    \begin{align}
        \sup_{|z| \leq 1/(2\theta)} \|O_{A,\eff}^{[r]}(z) \| &\leq \|O_A\| \exp\left(\frac{k_D r^D \|F_g\|}{4}\right) \\
        &= a' e^{b' r^D}\ ,
    \end{align}
    for $a',b'>0$ constants, as required.
\end{proof}

\section{Extrapolating extensive quantities}\label{app:extensive properties}

\subsection{Extrapolation of analytic Hamiltonians}

In this section, we will extend Theorem~\ref{thm:extrapolation without gadgets} to the case of extensive observables, i.e. of the form $O = \sum_{A\subseteq \Gamma} O_A$, where $O$ typically has support across the entire system $\Gamma$. The only requirement on $O$ is that its interactions decay exponentially over long distances --- in other words, it is bounded in the $\|\cdot\|_{F_g}$-norm like the Hamiltonian $H$.

\begin{theorem}\label{thm:extensive extrapolation}
    Let $\{H(x\}_{x \in \cX}$ be a family of Hamiltonians on $\Gamma$, where $|\Gamma|=n$. Let $O$ be a (possibly extensive) observable on $\Gamma$ such that $\|O\|_{F_g} = \bigO(1)$ (where $F_g$ is the same function describing the exponential decay of interactions in $H(x)$). Then the following holds:
    \begin{enumerate}[(I)]
        \item For constant $\beta > 0$, assume the Gibbs states $\rho_\beta(x)$ satisfy Assumption~\ref{assumption:non criticality}\eqref{it:gibbs assumption}. Then the function $f_\beta(x) := \tr[O \rho_\beta(x)]$ has a $(\delta,M,R)$-analytic approximation for any $\delta > 0$, where
        \begin{align}
            M = \bigO(n) \ ,\quad x_{\ast} \geq R = 1/\bigO(\log^D(n\delta^{-1}))\ .
        \end{align}
        \item Assume that $\{H(x)\}_{x\in \cS}$ satisfies Assumption~\ref{assumption:non criticality}\eqref{it:ground assumption}. Then the function $f_{\ground}(x) := \bra{\psi_0(x)} O\ket{\psi_0(x)}$ has a $(\delta,M,R)$-analytic approximation for any $\delta > 0$, where
        \begin{align}
            M = \bigO(n) \ ,\quad x_{\min} \geq R = 1/\bigO(\log^{D+1}(n\delta^{-1}))\ .
        \end{align}
    \end{enumerate}
\end{theorem}

Concretely, this yields the following scalings for performing Richardson extrapolation:

\begin{corollary}[Extrapolating extensive observables within phases of matter]\label{cor:extensive extrapolation}
    The value of $f_{\beta}(0)$ (respectively $f_{\ground}(0)$) can be calculated up to any desired accuracy $\epsilon > 0$, given the values of $f_\beta(x_k)$ (respectively $f_{\ground}(x_k)$) at $m$ Chebyshev sample points $\{x_k\}_{k=1}^m$, where each $x_k$ is bounded above zero by $x_{\min} := \min_k x_k$, where:
    \begin{enumerate}[(I)]
        \item For $f_\beta$,
        \begin{align}
            m = \bigO(\log(n\epsilon^{-1})) \ ,\quad x_{\min} = \frac{1}{\bigO(\log^{D+2}(n\epsilon^{-1}))} \ .
        \end{align}
        \item For $f_{\ground}$,
        \begin{align}
            m = \bigO(\log(n\epsilon^{-1}))\ ,\quad x_{\min} = \frac{1}{\bigO(\log^{D+3}(n\epsilon^{-1}))}\ .
        \end{align}
    \end{enumerate}
    The result still holds using noisy estimates of $f_\beta(x_k)$ (respectively $f_{\ground}(x_k)$ each with additive error $\delta = \Theta(\epsilon / \log\log(\epsilon^{-1}))$.
\end{corollary}

\begin{proof}[*cor:extensive extrapolation]
    Corollary~\ref{cor:richardson extrapolation approx} ensures we can extrapolate to within error $\epsilon > 0$, given a $(\delta,M,R)$-analytic approximation and $m$ Chebyshev samples, as long as
    \begin{align}
        \epsilon = (\delta + 2^{-m} M)\bigO(\log m)\ .
    \end{align}
    This time, we have $M = \bigO(n)$, so for the right-hand side to be sufficiently small we can choose $m = \bigO(\log(n\epsilon^{-1}))$ and $\delta = \epsilon / \log\log(n\epsilon^{-1})$. This leads to $R = 1/\bigO(\log^D(n\epsilon^{-1}))$ and $R = 1/\bigO(\log^{D+1}(n\epsilon^{-1}))$ for the Gibbs and ground state cases respectively. Using that $x_{\min} \sim R / m^2$ we arrive at the stated result.
\end{proof}

Ultimately, our approach will just be to split $O$ into a sum of $n=|\Gamma|$ observables each localised around a single site, show that each of these has an analytic approximation using Theorem~\ref{thm:extrapolation without gadgets}, and conclude by summing the analytic approximations. There is a small subtlety to applying this approach: the individual observables in the summand will in fact be quasi-local rather than strictly local as in the statement of Theorem~\ref{thm:extrapolation without gadgets}. To this end, we prove the following lemma (which will also be useful in Section~\ref{sec:simulator extrapolation}), generalising Lemma~\ref{lem:qbp and spectral flow truncation app}.

\begin{lemma}[Restatement of Lemma~\ref{lem:truncations with decaying observable}]\label{lem:truncations with decaying observable app}
    Let $H$ be a Hamiltonian with bounded $\|\cdot\|_{F_g}$-norm for $g$ linear, and let $O$ be an observable with $\|O\| = \bigO(1)$ localised around $i\in \Gamma$ in the following sense: for every $r \geq 0$, there exists an observable $O^{[r]}$ with support contained within $B_r(\{i\})$ such that
    \begin{align}
        \|O - O^{[r]}\| \leq a e^{-bg(r)}\ ,
    \end{align}
    for some constants $a,b > 0$. Let $\Phi_H(O)$ and $\Psi_H(O)$ be the quantum belief propagation and spectral flow operators as defined in Eqs.~\eqref{eq:qbp operator} and \eqref{eq:spectral flow operator}. Then there exist constants $a_1',a_2',b_1',b_2'>0$ such that, for every $r \geq 0$, there exist operators $\Phi_H^{[r]}(O)$ and $\Psi_H^{[r]}(O)$ with support contained within $B_r(\{i\})$ such that
    \begin{align}
        \|\Phi_H(O) - \Phi_H^{[r]}(O)\| &\leq a_1' \|O\| e^{-b_1' g(r)} \ , \\
        \|\Psi_H(O) - \Psi_H^{[r]}(O)\| &\leq a_2' \|O\| e^{-b_2' g(r) / \log^2 g(r)}\ .
    \end{align}
\end{lemma}

\begin{proof}[*lem:truncations with decaying observable app]
    Fix $r\geq 0$, and let $r' = \lfloor r/2\rfloor$. By assumption, the observable $O^{[r']}$ has support contained within $B_{r'}(\{i\})$, and
    \begin{align}\label{eq:O truncation error}
        \|O - O^{[r']} \| \leq a\|O\| e^{-bg(r')}\ .
    \end{align}
    Abusing notation slightly, we define $\Phi_H^{[r']}(O^{[r']})$ and $\Psi_H^{[r']}(O^{[r']})$ as prescribed by Lemma~\ref{lem:qbp and spectral flow truncation}, which have support contained within $B_{r'}(B_{r'}(\{i\})) \subseteq B_r(\{i\})$, and where
    \begin{align}
        \|\Phi_H(O^{[r']}) - \Phi_H^{[r']}(O^{[r']})\| &\leq a_1 |B_{r'}(\{i\}| \|O^{[r']}\| e^{-b_1 g(r')}\ , \\
        \|\Psi_H(O^{[r']}) - \Psi_H^{[r']}(O^{[r']})\| &\leq a_2 |B_{r'}(\{i\}| \|O^{[r']}\| e^{-b_2 g(r')/\log^2g(r')}\ ,
    \end{align}
    where $a_1,a_2,b_1,b_2>0$ are the constants from Eqs.~\eqref{eq:truncation w constants qbp}-\eqref{eq:truncation w constants sf}. By the triangle inequality, we can combine this inequality with Eq.~\eqref{eq:O truncation error} to give
    \begin{align}
        \| \Phi_H^{[r']}(O^{[r']}) - \Phi_H(O)\| &\leq \|\Phi_H(O^{[r']}) - \Phi_H^{[r']}(O^{[r']})\| + \|O - O^{[r']}\| \\
        &\leq \left( a_1 |B_{r'}(\{i\})| \|O^{[r']}\|+ a\|O\| \right) e^{-b_1 g(r')} \ .
    \end{align}
    Since $g$ is linear by assumption, the exponential decay dominates the bracketed term, which grows only polynomially with $r'$, and the result follows by setting $\Phi^{[r]}_H(O) := \Phi^{[r']}_H(O^{[r']})$. An analogous argument applies to $\Psi_H(O)$.
\end{proof}

\begin{proof}[*thm:extensive extrapolation]
    For every $i \in \Gamma$, define the observable $O_{[i]}$ by
    \begin{align}
        O_{[i]} := \sum_{\substack{A\subseteq \Gamma \\ i \in A}} \frac{1}{|A|} O_A\ .
    \end{align}
    Note that by construction, $O = \sum_{i\in \Gamma} O_{[i]}$. Also, for any $r\geq 0$, we can truncate $O_{[i]}$ to $B_r(\{i\})$ by setting
    \begin{align}
        O_{[i]}^{[r]} := \sum_{\substack{A\subseteq B_r(\{i\}) \\ i\in A}} \frac{1}{|A|}O_A\ .
    \end{align}
    Then we can bound
    \begin{align}
        \|O_{[i]} - O_{[i]}^{[r]} \| &\leq \sum_{j \in \Gamma \setminus B_r(i)} \sum_{\substack{A\subseteq \Gamma \\ i,j \in A}} \frac{1}{|A|} \|O_A\| \\
        &\leq \sum_{r' > r} \sum_{\substack{j\in \Gamma \\ \dist(i,j) = r'}} \sum_{\substack{A\subseteq \Gamma \\ i,j \in A}} \|O_A\| \\
        &\leq \sum_{r' > r} \sum_{\substack{j\in \Gamma \\ \dist(i,j) = r'}} \|O\|_{F_g} F_g(r') \\
        &\leq \|O\|_{F_g} \sum_{r' > r} |B_{r'}(i)| F_g(r') \\
        &\leq \|O\|_{F_g} F(0) k_D \sum_{r > r'}  (r')^D e^{-g(r')}\ .
    \end{align}
    Since $g$ is linear, the right-hand side of this expression is exponentially decaying in $r$ and in particular $O_{[i]}$ satisfies the assumptions of Lemma~\ref{lem:truncations with decaying observable app}. Now we consider the Gibbs state and ground state cases separately:
    \begin{enumerate}[(I)]
        \item For every $i \in \Gamma$, we define the function $f_{\beta,i}(s) := \tr[O_{[x]}\rho_\beta(s)]$. Note that $f_\beta(s) = \sum_{i\in \Gamma} f_{\beta,i}(s)$. We may now apply Theorem~\ref{thm:extrapolation without gadgets}\eqref{it:gibbs extrapolation without gadgets} to this function --- using Lemma~\ref{lem:truncations with decaying observable app} in place of Lemma~\ref{lem:qbp and spectral flow truncation app}, to account for the fact that $O_{[i]}$ is quasi-local rather than strictly local --- to deduce that for any $\delta > 0$ there exists a $(\delta / n,\tilde{M},\tilde{R})$-analytic approximation for $f_{\beta,i}$, where
        \begin{align}
            \tilde{M} = \bigO(1) \ ,\quad \tilde{R} = 1/\bigO(\log^D(n\delta^{-1}))\ .
        \end{align}
        We denote this analytic approximation $\tilde{f}_{\beta,i}(z)$, and define the complex function $\tilde{f}_\beta(z)$ as the sum
        \begin{align}
            \tilde{f}_\beta(z) := \sum_{i\in \Gamma} \tilde{f}_{\beta,i}(z)\ .
        \end{align}
        It follows immediately that $\tilde{f}_\beta$ is a $(\delta,M,R)$-analytic approximation for $f_\beta$, where
        \begin{align}
            M = n\tilde{M} = \bigO(n) \ ,\quad R = \tilde{R} = 1/\bigO(\log^D(n\delta^{-1}))\ ,
        \end{align}
        as required.
        \item Similarly, for every $i \in \Gamma$ we define the function $f_{\ground,i}(s) := \bra{\psi_0(s)} O_{[i]}\ket{\psi_0(s)}$, and apply Theorem~\ref{thm:extrapolation without gadgets}\eqref{it:gs extrapolation without gadgets} --- with Lemma~\ref{lem:truncations with decaying observable app} in place of Lemma~\ref{lem:qbp and spectral flow truncation app} --- to construct a $(\delta/n,\tilde{M},\tilde{R})$-analytic approximation $\tilde{f}_{\ground,i}(z)$ for each $f_{\ground,i}$, where
        \begin{align}
            \tilde{M} = \bigO(1) \ ,\quad \tilde{R} = 1/\bigO(\log^{D+1}(n\delta^{-1}))\ .
        \end{align}
        Summing these gives
        \begin{align}
            \tilde{f}_{\ground}(z) := \sum_{i\in \Gamma}\tilde{f}_{\ground,i}(z) \ ,
        \end{align}
        which is a $(\delta,M,R')$-approximation for $f_{\ground}$, where
        \begin{align}
            M = n\tilde{M} = \bigO(n)\ ,\quad R = \tilde{R} = 1/\bigO(\log^{D+1}(n\delta^{-1}))\ ,
        \end{align}
        as required.
    \end{enumerate}
\end{proof}

\subsection{Extrapolation of simulator Hamiltonians}\label{app:extensive simulator}

In this section, we extend Theorem~\ref{thm:extrapolation with gadgets} to the case of an extensive observable $O$.

\begin{theorem}\label{thm:extensive extrapolation simulator}
    Let $H'(x)$ be a family of gadget Hamiltonians on $\cH' = \bigotimes_{i\in \Gamma_{\eff} \cup \Gamma_{\anc}} \cH_x$ satisfying the conditions of Assumption~\ref{assumption:gadget ham} with exponentially decaying interactions, and let $H_{\tar} := H_{\eff}(0)$. Let $O = \sum_{A\subseteq \Gamma_{\eff}} O_A$ be an observable on $\Gamma_{\eff}$ such that $\|O\|_{F_g} = \bigO(1)$ (where $F_g$ is the same function describing the exponential decay of interactions in $H'(x)$). Then the following holds:
    \begin{enumerate}[(I)]
        \item Let $\beta > 0$, and let $\rho_\beta(x)$, $\rho_{\beta,\eff}(x)$, and $\rho_{\beta,\tar}$ be the Gibbs states corresponding to the Hamiltonians $H'(x)$, $H_{\eff}(x)$, and $H_{\tar}$ respectively. Assume the family $\rho_{\beta,\eff}(x)$ satisfies Assumption~\ref{assumption:non criticality}\eqref{it:gibbs assumption} for $x\in [0,x_{\ast}]$. Then the function
        \begin{align}
            f_\beta'(x) := \left\{\begin{array}{ll}
                \tr[O \rho_\beta'(x)] & \quad\text{for}\quad x\in (0,x_{\ast}) \\
                \tr[O \rho_{\beta,\tar}] & \quad\text{for}\quad x= 0
            \end{array} \right.\ ,
        \end{align}
        has a $(\delta,M,R)$-analytic approximation for $\delta > 0$, where
        \begin{align}
            M = n\exp\left(\bigO( \log^D(n\delta^{-1})) \right) \ ,\quad x_{\ast} \geq R = 1/\bigO(\log^D(n\delta^{-1}))\ .
        \end{align}
        \item Assume the $H_{\eff}(x)$ satisfies Assumption~\ref{assumption:non criticality}\eqref{it:ground assumption} for $x\in [0,x_{\ast}]$, and let $\ket{\psi_0'(x)}$, $\ket{\psi_{0,\eff}(x)}$, and $\ket{\psi_{0,\tar}}$ be the ground states corresponding to the Hamiltonians $H'(x)$, $H_{\eff}(x)$, and $H_{\tar}$ respectively. Then the function
        \begin{align}
            f'_{\ground}(x) := \left\{\begin{array}{ll}
                \bra{\psi_0'(x)} O \ket{\psi_0'(x)} & \quad\text{for}\quad x\in (0,x_{\ast}) \\
                \bra{\psi_{0,\tar}}O\ket{\psi_{0,\tar}} & \quad\text{for}\quad x = 0
            \end{array} \right.\ ,
        \end{align}
        has a $(\delta,M,R)$-analytic approximation for any $\delta > 0$, where
        \begin{align}
            M = n\exp\left(\bigO(\log^{D}(n\delta^{-1})) \right) \ ,\quad x_{\ast} \geq R = 1/\bigO(\log^{D+1}(n\delta^{-1}))\ .
        \end{align}
    \end{enumerate}
\end{theorem}

\begin{proof}[*thm:extensive extrapolation simulator]
    We proceed exactly as in Theorem~\ref{thm:extensive extrapolation}, and for every $i\in \Gamma_{\eff}$ we define
    \begin{align}
        O_{[i]} := \sum_{\substack{A\subseteq \Gamma \\ i \in A}} \frac{1}{|A|} O_A\ .
    \end{align}
    Then $O = \sum_{i\in \Gamma_{\eff}} O_{[i]}$, and by the same argument as Theorem~\ref{thm:extensive extrapolation} we are guaranteed that each $O_{[i]}$ can be truncated to any radius $r$ up to exponentially small error; the same thus holds for the corresponding effective observable
    \begin{align}
        O_{[i],\eff}(x) := (\id_{\eff} \otimes \bra{\mathbf{0}_{\anc}}) e^{T(x)} O_{[i]} e^{-T(x)} (\id_{\eff} \otimes \ket{\mathbf{0}_{\anc}})\ ,
    \end{align}
    as in Eq.~\eqref{eq:oaeffdef}, by Lemma~\ref{lem:oaeff truncation}. 
    \begin{enumerate}[(I)]
        \item For every $i\in \Gamma_{\eff}$, we define $f_{\beta,i}'(x) := \tr[O_{[i]}\rho_\beta'(x)]$. For $\delta > 0$, Theorem~\ref{thm:extrapolation with gadgets}\eqref{it:gibbs extrapolation with gadgets} then gives a $(\delta/n,\tilde{M},\tilde{R})$-analytic approximation $\tilde{f}_{\beta,i}(z)$, where
        \begin{align}
            \tilde{M} = \exp\left(\bigO( \log^D(n \delta^{-1})) \right)\ ,\quad \tilde{R} = 1/\bigO(\log^D(n\delta^{-1}))\ .
        \end{align}
        Hence $\tilde{f}_{\beta}(z) := \sum_i \tilde{f}_{\beta,i}(z)$ is a $(\delta,M,R)$-analytic approximation for $f'_\beta$, where
        \begin{align}
            M = n\tilde{M} = n\exp\left(\bigO(\log^D(n\delta^{-1}))\right)\ ,\quad R = \tilde{R} = 1/\bigO(\log^D(n\delta^{-1}))\ .
        \end{align}
        \item Similarly, for every $i \in \Gamma_{\eff}$, we define $f_{\ground,i}(x) := \bra{\psi_0'(x)} O_{[i]} \ket{\psi_0'(x)}$. For $\delta > 0$, Theorem~\ref{thm:extrapolation with gadgets}\eqref{it:gs extrapolation with gadgets} then gives a $(\delta/n,\tilde{M},\tilde{R})$-analytic approximation $\tilde{f}_{\ground,i}(z)$, where
        \begin{align}
            \tilde{M} = \exp\left( \bigO(\log^{D}(n\delta^{-1})) \right)\ ,\quad \tilde{R} = 1/\bigO(\log^{D+1}(n\delta^{-1}))\ .
        \end{align}
        Hence $\tilde{f}_{\ground}(z) := \sum_i \tilde{f}_{\ground,i}(z)$ is a $(\delta,M,R)$-analytic approximation for $f'_{\ground}$, where
        \begin{align}
            M = n\tilde{M} = n\exp\left(\bigO(\log^{D}(n\delta^{-1})) \right) \ ,\quad R = \tilde{R} = 1/\bigO(\log^{D+1}(n\delta^{-1}))\ .
        \end{align}
    \end{enumerate}
\end{proof}

\begin{corollary}[Extrapolating extensive observables with simulator Hamiltonians]\label{cor:extensive extrapolation simulator}
    The value of $\tr[O\rho_{\beta,\tar}]$ (respectively $\bra{\psi_{0,\tar}} O \ket{\psi_{0,\tar}}$) can be calculated up to any desired accuracy $\epsilon > 0$, given the values of $f'_\beta(x_k)$ (respectively $f_{\ground}(x_k)$) at $m$ Chebyshev sample points $\{x_k\}_{k=1}^m$, where each $x_k$ is bounded above zero by $x_{\min} := \min_k x_k$, where:
    \begin{enumerate}[(I)]
        \item For $f_\beta'$,
        \begin{align}
            m = \bigO(\log^D(n\epsilon^{-1})) ,\quad x_{\min} = \frac{1}{\bigO(\log^{3D}(n\epsilon^{-1}))} \ .
        \end{align}
        \item For $f_{\ground}'$,
        \begin{align}
            m = \bigO(\log^D(n\epsilon^{-1}))\ ,\quad x_{\min} = \frac{1}{\bigO(\log^{3D+1}(n\epsilon^{-1}))}\ .
        \end{align}
    \end{enumerate}
    In particular, this process requires simulator Hamiltonians with interaction strengths scaling as $x_{\min}^{-d} \sim \bigO( \poly \log (n\epsilon^{-1}))$. These conclusions also hold using noisy estimates of $f_\beta'(x_k)$ (respectively $f'_{\ground}(x_k)$), each with additive error $\delta = \Theta(\epsilon / \log\log (n\epsilon^{-1}))$.
\end{corollary}

\begin{proof}[*cor:extensive extrapolation simulator]
    We are guaranteed by Corollary~\ref{cor:richardson extrapolation approx} that extrapolation to error $\epsilon > 0$ is possible with $m$ Chebyshev samples from a $(\delta,M,R)$-analytic approximation, assuming that
    \begin{align}
        \epsilon = (\delta + 2^{-m} M) \bigO(\log m)\ .
    \end{align}
    By Theorem~\ref{thm:extensive extrapolation simulator}, we have $M = n\exp(\bigO(\log^D(n\delta^{-1}))$. Taking $\delta \sim \epsilon / \log\log(n\epsilon^{-1})$, this gives
    \begin{align}
        M = n\exp(\bigO(\log^D(n\epsilon^{-1})) = \exp(\bigO(\log^D(n\epsilon^{-1}))\ .
    \end{align}
    $M = n\exp(\bigO(\log^D(n\epsilon^{-1}))$. Hence it is sufficient to take $m = \bigO(\log^D(n\epsilon^{-1}))$. Using $x_{\min} \sim R/m^2$ we arrive at the desired result.
\end{proof}

\end{document}

%% file: commands.tex
\usepackage{graphicx}
\usepackage{url}
\usepackage{amsmath,amssymb}
\usepackage{amsfonts}
\usepackage{hyperref}
\usepackage{ifthen}
\usepackage{dsfont}
\usepackage{enumerate}
\usepackage{float}
\usepackage[capitalise]{cleveref}
\usepackage{xcolor}
\usepackage{subcaption}
\usepackage{tikz}
\usepackage{authblk}
\usetikzlibrary{calc,shapes.geometric,decorations.pathmorphing}
\DeclareMathOperator{\RR}{\mathbb{R}}
\DeclareMathOperator{\CC}{\mathbb{C}}
\DeclareMathOperator{\ZZ}{\mathbb{Z}}

\DeclareMathOperator{\cH}{\mathcal{H}}
\DeclareMathOperator{\cS}{\mathcal{S}}
\DeclareMathOperator{\cX}{\mathcal{X}}

\DeclareMathOperator{\cU}{\mathcal{U}}
\DeclareMathOperator{\diff}{\mathrm{d}\!}
\DeclareMathOperator{\id}{\mathds{1}}

\DeclareMathOperator{\cL}{\mathcal{L}}
\DeclareMathOperator{\U}{U}
\DeclareMathOperator{\Lin}{Lin}
\DeclareMathOperator{\Herm}{Herm}

\DeclareMathOperator{\sign}{\mathsf{sign}}
\DeclareMathOperator{\loc}{loc}

\DeclareMathOperator{\eff}{\mathsf{eff}}
\DeclareMathOperator{\tar}{tar}

\DeclareMathOperator{\Proj}{Proj}
\DeclareMathOperator{\Cov}{Cov}
\DeclareMathOperator{\dist}{dist}
\DeclareMathOperator{\real}{Re}

\DeclareMathOperator{\texp}{exp_\mathcal{T}}
\DeclareMathOperator{\QMA}{\textsf{QMA}}

\DeclareMathOperator{\bigO}{\mathcal{O}}

\DeclareMathOperator{\vx}{\mathbf{x}}
\DeclareMathOperator{\ondiag}{\mathcal{D}}
\DeclareMathOperator{\offdiag}{\mathcal{O}}
\DeclareMathOperator{\anc}{\mathsf{anc}}

\DeclareMathOperator{\ground}{\mathsf{ground}}

\newcommand{\ket}[1]{|#1\rangle}
\newcommand{\bra}[1]{\langle#1|}
\newcommand{\proj}[1]{|#1\rangle \langle #1 |}

\usepackage{phfqit}
\usepackage{phfparen}
\usepackage{phfcc}
\usepackage[smallproofs=false]{phfthm}

\phfMakeTheorem[defstar=true,defnostar=true,thmstyle=plain,proofref=false]{example}{Example}
\crefname{conjecture}{Conjecture}{Conjectures}
\phfMakeTheorem[defstar=true,defnostar=true,thmstyle=plain,proofref=false,counter=theorem]{setup}{Setup}
\crefname{setup}{Setup}{Setups}
\phfMakeTheorem[defstar=true,defnostar=true,thmstyle=plain,proofref=false,counter=theorem]{assumption}{Assumption}
\crefname{assumption}{Assumption}{Assumptions}
\phfMakeTheorem[defstar=true,defnostar=true,thmstyle=plain,proofref=false,counter=theorem]{definitionnc}{Definition}
\phfMakeTheorem[defstar=true,defnostar=true,thmstyle=plain,proofref=true,counter=theorem]{claim}{Claim}
\phfMakeTheorem[defstar=true,defnostar=true,thmstyle=plain,proofref=false,counter=theorem]{result}{Result}
\crefname{assumption}{Claim}{Claims}
\numberwithin{theorem}{section}
\numberwithin{lemma}{section}
\numberwithin{corollary}{section}
\numberwithin{definition}{section}
\numberwithin{assumption}{section}

\makeatletter
\def\phfthm@proofrefstyle@sec@setup{%
  \phfthm@proofrefstyle@default@setup
  \def\phfthm@proofref@impl@fmt##1##2{
  \parfillskip=0pt\relax%
    \hfil\null\hfil\null\hfil%
    \hbox{\proofrefsize{(Proof in \cref{sec:proof:##1})}}\par
  }%
}

\phfMakeTheorem[defstar=true,defnostar=true,thmstyle=plain,proofref=true,counter=theorem,proofrefstyle=sec]{theoremsec}{Theorem}
\phfMakeTheorem[defstar=true,defnostar=true,thmstyle=plain,proofref=true,counter=theorem,proofrefstyle=sec]{propositionsec}{Proposition}
\phfMakeTheorem[defstar=true,defnostar=true,thmstyle=plain,proofref=true,counter=theorem,proofrefstyle=sec]{lemmasec}{Lemma}

%% file: figures/gadgetextrapolation.tex
\begin{subfigure}{0.45\textwidth}
    \centering
    \begin{tikzpicture}
        \def\sitespacing{.6};
        \def\siterad{0.05};
        \def\intthickness{0.2};
        \def\loweroffset{2};
        \def\siteheight{.5};
        \def\arrowbuffer{0.3};
    
        \coordinate (Q1) at (0,0);
        \coordinate (Q2) at (0,-\loweroffset);

        \draw[thick,dashed,draw=red!40!black,fill=red!50,fill opacity=0.3,rounded corners=\intthickness cm] ($(Q1) + (-\intthickness, -\intthickness)$) rectangle ($(Q1) + (2*\sitespacing + \intthickness,\intthickness)$);
        \node[font=\footnotesize,red!80!black,above] at ($(Q1) + (\sitespacing,\intthickness)$) {$h_1$};
        \draw[thick,dashed,draw=green!40!black,fill=green!50,fill opacity=0.3,rounded corners=\intthickness cm] ($(Q1) + (2*\sitespacing-\intthickness, -\intthickness)$) rectangle ($(Q1) + (4*\sitespacing + \intthickness,\intthickness)$);
        \node[font=\footnotesize,green!80!black,above] at ($(Q1) + (3*\sitespacing,\intthickness)$) {$h_2$};
        \draw[thick,dashed,draw=blue!40!black,fill=blue!50,fill opacity=0.3,rounded corners=\intthickness cm] ($(Q1) + (4*\sitespacing-\intthickness, -\intthickness)$) rectangle ($(Q1) + (6*\sitespacing + \intthickness,\intthickness)$);
        \node[font=\footnotesize,blue!80!black,above] at ($(Q1) + (5*\sitespacing,\intthickness)$) {$h_3$};

        \foreach\x in {0,...,6}{
            \draw[fill=black] ($(Q1) + (\x*\sitespacing,0)$) circle (\siterad);
        }
        \node at ($(Q1) + (-\sitespacing,0)$) {$\dots$};
        \node at ($(Q1) + (7*\sitespacing,0)$) {$\dots$};
        
        \coordinate (Q2) at ($(Q1) + (0,-\loweroffset)$);
        \coordinate (arrowstart) at ($(Q1) + (3*\sitespacing,-\intthickness-\arrowbuffer)$);
        \coordinate (arrowend) at ($(Q2) + (3*\sitespacing,\siteheight+\arrowbuffer)$);

        \draw[thick,->] (arrowstart) -- (arrowend);

        \draw[thick,draw=red!80!black] ($(Q2)$) -- ($(Q2) + (\sitespacing,\siteheight)$) -- ($(Q2) + (2*\sitespacing,0)$) -- ($(Q2) + (\sitespacing,-0.5*\siteheight)$) -- ($(Q2)$) -- ($(Q2) + (2*\sitespacing,0)$);
        \draw[thick,draw=red!80!black] ($(Q2) + (\sitespacing,\siteheight)$) -- ($(Q2) + (\sitespacing,-0.5*\siteheight)$);
        \node[red!80!black,font=\footnotesize,below] at ($(Q2) + (\sitespacing,-0.5*\siteheight)$) {$h_1'(x)$};

        \coordinate (Q3) at ($(Q2) + (2*\sitespacing,0)$);
        \draw[thick,draw=green!80!black] ($(Q3)$) -- ($(Q3) + (\sitespacing,\siteheight)$) -- ($(Q3) + (2*\sitespacing,0)$) -- ($(Q3) + (\sitespacing,-0.5*\siteheight)$) -- ($(Q3)$) -- ($(Q3) + (2*\sitespacing,0)$);
        \draw[thick,draw=green!80!black] ($(Q3) + (\sitespacing,\siteheight)$) -- ($(Q3) + (\sitespacing,-0.5*\siteheight)$);
        \node[green!80!black,font=\footnotesize,below] at ($(Q3) + (\sitespacing,-0.5*\siteheight)$) {$h_2'(x)$};

        \coordinate (Q4) at ($(Q3) + (2*\sitespacing,0)$);
        \draw[thick,draw=blue!80!black] ($(Q4)$) -- ($(Q4) + (\sitespacing,\siteheight)$) -- ($(Q4) + (2*\sitespacing,0)$) -- ($(Q4) + (\sitespacing,-0.5*\siteheight)$) -- ($(Q4)$) -- ($(Q4) + (2*\sitespacing,0)$);
        \draw[thick,draw=blue!80!black] ($(Q4) + (\sitespacing,\siteheight)$) -- ($(Q4) + (\sitespacing,-0.5*\siteheight)$);
        \node[blue!80!black,font=\footnotesize,below] at ($(Q4) + (\sitespacing,-0.5*\siteheight)$) {$h_3'(x)$};

        \foreach\x in {0,...,2}{
            \draw[fill=black] ($(Q2) + (2*\x*\sitespacing,0)$) circle (\siterad);
            \draw[fill=black] ($(Q2) + (\sitespacing+2*\x*\sitespacing,-0.5*\siteheight)$) circle (\siterad);
            \draw[fill=black] ($(Q2) + (\sitespacing+2*\x*\sitespacing,\siteheight)$) circle (\siterad);
        }
        \draw[fill=black] ($(Q2) + (6*\sitespacing,0)$) circle (\siterad);
        \node at ($(Q2) + (-\sitespacing,0)$) {$\dots$};
        \node at ($(Q2) + (7*\sitespacing,0)$) {$\dots$};
        
    \end{tikzpicture}
    \caption{}
\end{subfigure}
\begin{subfigure}{0.45\textwidth}
    \centering
    \begin{tikzpicture}
        \def\yaxis{2.5};
        \def\xaxis{4};
        \def\notch{0.1};

        \coordinate (origin) at (0,0);

        \draw[thick,->] (origin) -- ($(origin) + (\xaxis,0)$) node[right,font=\footnotesize]{$x$};
        \draw[thick,->] (origin) -- ($(origin) + (0,\yaxis)$) node[above,font=\footnotesize]{$f(x)$};

        \def\correct{1.8};
        \def\approx{1.7};
        \def\eps{0.25}

        \draw[thick] ($(origin) + (0,\correct)$) -- ($(origin) + (-\notch,\correct)$) node[left,font=\footnotesize]{$f(0)$};
        \draw[red!80!black] ($(origin) + (0,\correct)$) -- ($(origin) + (\xaxis,\correct)$);
        \draw[draw=none,fill=red!50,fill opacity=0.3] ($(origin) + (0,\correct + \eps)$) rectangle ($(origin) + (\xaxis,\correct-\eps)$);
        \draw[<->,thick,red!80!black] ($(origin) + (\notch+\xaxis,\correct-\eps)$) -- ($(origin) + (\notch+\xaxis,\correct)$) node[font=\footnotesize, below right]{$\epsilon$};

        \def\smallx{1};
        \def\largex{2.5};

        \draw[dashed] ($(origin) + (\smallx,0)$) -- ($(origin) + (\smallx,\yaxis)$);
        \draw[thick] ($(origin) + (\smallx,0)$) -- ($(origin) + (\smallx,-\notch)$) node[below,font=\footnotesize]{$x_\text{sim}$}; 
        \draw[dashed] ($(origin) + (\largex,0)$) -- ($(origin) + (\largex,\yaxis)$);
        \draw[thick] ($(origin) + (\largex,0)$) -- ($(origin) + (\largex,-\notch)$) node[below,font=\footnotesize]{$x_\text{ext}$}; 
        \draw[thick] ($(origin)$) -- ($(origin) + (0,-\notch)$) node[below,font=\footnotesize]{$0$}; 

        \draw[domain=0:\xaxis,smooth,variable=\x] plot (\x,{\correct+0.3*\x-0.4*\x^2 + 0.06*\x^3});

        \draw[domain=0:\xaxis,smooth,variable=\x,blue,dashed,thick] plot (\x,{\approx+0.3*\x-0.4*\x^2 + 0.065*\x^3});

        \foreach\x in {2.7,3,3.3,3.8} {
        \draw[draw=blue,fill=blue!30,thick] (\x,{\approx+0.3*\x-0.4*\x^2 + 0.065*\x^3}) circle (0.05);
        }
    \end{tikzpicture}
    \caption{}
\end{subfigure}

%% file: figures/partitionzeroes.tex
\begin{tikzpicture}

    \def\minx{-3};
    \def\maxx{5};
    \def\miny{-3};
    \def\maxy{3};
    \def\roc{2};
    \def\crossSize{0.06};
    \def\complexZero(#1,#2){
    \draw[thick,blue] ($(#1,#2) + (-\crossSize,-\crossSize)$) -- ($(#1,#2) + (\crossSize,\crossSize)$);
    \draw[thick,blue] ($(#1,#2) + (-\crossSize,\crossSize)$) -- ($(#1,#2) + (\crossSize,-\crossSize)$);

    \draw[thick,blue] ($(#1,-#2) + (-\crossSize,-\crossSize)$) -- ($(#1,-#2) + (\crossSize,\crossSize)$);
    \draw[thick,blue] ($(#1,-#2) + (-\crossSize,\crossSize)$) -- ($(#1,-#2) + (\crossSize,-\crossSize)$);
    }

    \draw[thick,dashed,fill=orange,fill opacity=0.3,draw=orange!80!black] (0,0) circle (\roc);
    
    \draw[->,thick] (\minx,0) -- (\maxx,0) node[right]{$\mathsf{Re}(z)$};
    \draw[->,thick] (0,\miny) -- (0,\maxy) node[above]{$\mathsf{Im}(z)$};

    \complexZero(1,2);
    \complexZero(3,2.5);
    \complexZero(-2,1);
    \complexZero(-1.7,1.9);
    \complexZero(-0.5,2.5);
    \complexZero(2.1,0.1);
    \complexZero(2.1,0.3);
    \complexZero(3.5,2);

    \draw[thick,blue,->,font={\footnotesize}] (2.1,1) -- (2.1,0.4) node[above right]{$n\rightarrow\infty$};

    \draw[thick] (2.1,0) -- ($(2.1,0) + (0.2,-0.2)$) node[right=0.04,font={\footnotesize}] {$z_{\mathsf{crit}}$};

    \node at (0,0) [below right]{$0$};
    
\end{tikzpicture}

%% file: figures/localisation.tex
\begin{tikzpicture}
    \def\sitesx{5};
    \def\sitesy{4};
    \def\sitespacing{0.5};
    \def\siterad{0.06};
    \def\obsrad{0.15};
    \def\fillbuffer{0.3};
    \def\boxsizex{3.5};
    \def\boxsizey{3.5};
    \def\boxrounding{0.5};
    \def\distance{2};
    \def\arrowbuffer{0.5};
    \def\oradius{1.5};

    \coordinate (centreO1) at (0,0);
    \coordinate (centreO2) at ($(centreO1) + (2*\sitesx*\sitespacing + \distance,0)$);
    \coordinate (arrowstart) at ($(centreO1) + (\sitesx*\sitespacing + \arrowbuffer,0)$);
    \coordinate (arrowend) at ($(centreO2) + (-\sitesx*\sitespacing - \arrowbuffer,0)$);
    \coordinate (label1) at ($(centreO1) + (0,\sitesy*\sitespacing + \fillbuffer)$);
    \coordinate (label2) at ($(centreO2) + (0,\sitesy*\sitespacing + \fillbuffer)$);
    \coordinate (belowarrowstart) at ($(centreO2) + (\oradius*\sitespacing,-\sitesy*\sitespacing - \arrowbuffer)$);
    \coordinate (belowarrowend) at ($(centreO2) + (\boxsizex*\sitespacing,-\sitesy*\sitespacing - \arrowbuffer)$);

    \draw[fill=orange,fill opacity=0.3,draw=none, decorate, decoration={random steps, segment length=3mm, amplitude=1mm}] ($(centreO1) + (-\sitesx*\sitespacing - \fillbuffer,-\sitesy*\sitespacing - \fillbuffer)$) rectangle ($(centreO1) + (\sitesx*\sitespacing + \fillbuffer,\sitesy*\sitespacing + \fillbuffer)$);

    \draw[fill=orange,fill opacity=0.3,draw=orange!80!black,thick,dashed,rounded corners=\boxrounding cm] ($(centreO2) + (-\boxsizex*\sitespacing,-\boxsizey*\sitespacing)$) rectangle ($(centreO2) + (\boxsizex*\sitespacing,\boxsizey*\sitespacing)$);

    \node[orange,above] at (label1) {$H(x)$};
    \node[above] at (label2) {$H(0) + $\textcolor{orange}{$H_{\loc}(x)$}};

    % \draw[fill=blue,draw=none] (centreO1) circle (\obsrad) node[blue,below left]{$O$};
    % \draw[fill=blue,draw=none] (centreO2) circle (\obsrad) node[blue,below left]{$O$};

    \draw[fill=blue!10!white,draw=blue,thick,dashed,rounded corners=\boxrounding cm] ($(centreO1) - (\oradius*\sitespacing,\oradius*\sitespacing)$) rectangle ($(centreO1) + (\oradius*\sitespacing,\oradius*\sitespacing)$);
    \draw[fill=blue!10!white,draw=blue,thick,dashed,rounded corners=\boxrounding cm] ($(centreO2) - (\oradius*\sitespacing,\oradius*\sitespacing)$) rectangle ($(centreO2) + (\oradius*\sitespacing,\oradius*\sitespacing)$);

    \node[blue,above] at (centreO1) {$O$};
    \node[blue,above] at (centreO2) {$O$};
    
    \foreach\x in {-\sitesx,...,\sitesx} {
    \foreach\y in {-\sitesy,...,\sitesy} {
    \draw[fill=black] ($(centreO1) + (\x*\sitespacing,\y*\sitespacing)$) circle (\siterad);
    \draw[fill=black] ($(centreO2) + (\x*\sitespacing,\y*\sitespacing)$) circle (\siterad);
    }
    }
    \draw[thick,->] (arrowstart) -- (arrowend);

    \draw[thick,orange!80!black,<->] (belowarrowstart) -- (belowarrowend);
    \draw[dashed,orange!80!black] ($(centreO2) + (\oradius*\sitespacing,0)$)  -- (belowarrowstart);
    \draw[dashed,orange!80!black] ($(centreO2) + (\boxsizex*\sitespacing,0)$) -- (belowarrowend);
    \node[orange!80!black,below] at ($(centreO2) + (0.5*\boxsizex*\sitespacing,-\sitesy*\sitespacing - \arrowbuffer)$) {$\sim \log\epsilon^{-1}$};

\end{tikzpicture}